\documentclass[aps,
10pt,
prl,
groupedaddress,
twocolumn,
notitlepage,
  bibnotes 
  ]{revtex4-2} 
\usepackage{amsmath, amsfonts, amscd, amsthm, xspace, mdframed}
\usepackage{mathtools, amssymb}              
\usepackage{dsfont, mathrsfs, nicefrac, stmaryrd, bm, bbm, 
}
\usepackage{enumitem}

\usepackage{graphicx, graphics, color}
\usepackage[unicode=true, pdfusetitle, bookmarks=true, bookmarksnumbered=false,%
    bookmarksopen=false, breaklinks=false, pdfborder={0 0 0}, pdfborderstyle={},%
    backref=false, colorlinks=true]{hyperref}
\usepackage{tikz-cd}
\usepackage{changepage}
\usepackage{iftex}

\ifXeTeX%
  \usepackage{fontspec}
  \usepackage{unicode-math}           
  
\else
  \ifLuaTeX%
    \usepackage{fontspec}
    \usepackage{unicode-math} 
    
  \else
    \usepackage[T1]{fontenc}
    \usepackage{lmodern}
  \fi
\fi

\usepackage{savesym}
\savesymbol{degree}
\savesymbol{corresponds}
\restoresymbol{whatever}{degree}
\restoresymbol{whatever2}{corresponds}

\usepackage{tikz}
\usetikzlibrary{arrows, chains, matrix, positioning, scopes}
\usepackage[most]{tcolorbox}

\theoremstyle{definition}
\newtheorem{definition}{Definition}[section]

\newtheorem*{example*}{Example}
\newtheorem*{remark}{Remark}
\theoremstyle{plain}
\newtheorem{theorem}{Theorem}[section]
\newtheorem{theoremNS}{Theorem}
\newtheorem{lemma}[theorem]{Lemma}
\newtheorem{lemmaNS}{Lemma}

\newtheorem{corollary}[theorem]{Corollary}
\newtheorem{claim}[theorem]{Claim}

\newtheorem{observation}[theorem]{Observation}

\usepackage{thmtools}
\usepackage{thm-restate}

\newcommand{\AssumptionAnchor}[2]{%
  \hypertarget{ass:#1}{\textbf{#2}.}\hspace{-7pt}
}

\newcommand{\divergenceSeparator}{\delimsize\|\mathopen{}}

\newcommand{\smallgiven}{\nonscript\,\delimsize\vert\nonscript\,\mathopen{}}
\DeclarePairedDelimiterX{\pa}[1]{(}{)}{#1} 
\DeclarePairedDelimiterX{\pb}[1]{[}{]}{#1} 
\DeclarePairedDelimiterX{\pc}[1]{\{}{\}}{#1}
\DeclarePairedDelimiterX{\braket}[1]{\langle}{\rangle}{#1}
\DeclarePairedDelimiterX{\verts}[1]{\lvert}{\rvert}{#1}
\DeclarePairedDelimiterX{\Verts}[1]{\lVert}{\rVert}{#1}
\DeclarePairedDelimiter{\norm}{\lVert}{\rVert}

\newcommand{\divergence}[2]{\pa*{#1\divergenceSeparator#2}}

\DeclareRobustCommand{\divergence}[2]{\pa*{#1\divergenceSeparator#2}}
   \def\mf{\mathfrak}   \def\md{\mathbb}   \def\mb{\mathbf}

\def\8{\infty}
    \def\p{\mb{p}}

  \def\u{\mb{u}}  \def\0{\mathbf{0}}
\def\1{\mathbf{1}}
\def\N{\md{N}}   \def\R{\md{R}}  \def\C{\md{C}}
\def\D{\md{D}}   

\def\states{\mathfrak{D}}

\DeclareMathOperator{\tr}{Tr}
\DeclareMathOperator{\rank}{rk}
\DeclareMathOperator{\ima}{Im}

\DeclareMathOperator{\supp}{supp}

\DeclareMathOperator*{\opt}{opt}

\newcommand{\bra}[1]{ \langle#1|}
\newcommand{\ket}[1]{|#1\rangle}

\newcommand{\defeq}{\stackrel{\textup{\tiny def}}{=}}

\def\id{\mathsf{id}}

  \def\pos{\mathsf{Pos}}  
\def\CPTP{{\rm CPTP}}  \def\cptp{\CPTP}

\def\sym{{\rm Sym}}

\def\free{\mf{F}}           

\newcommand{\yield}[1][0]{{\mathsf{Yield}}^{#1}}
\newcommand{\cost}[1][0]{{\mathsf{Cost}}^{#1}}

\def\dmax{D_{\rm max}}
\def\dmin{D_{\rm min}}

\renewcommand{\ge}{\geqslant}  \renewcommand{\le}{\leqslant}
  
\renewcommand{\qedsymbol}{\nobreak\ifvmode\relax\else
  \ifdim\lastskip<1.5em \hskip-\lastskip\relax\hskip1.5em plus0em minus0.5em\fi
  \nobreak\vrule height0.75em width0.5em depth0.25em\fi}

\newcommand{\scriptveryshortarrow}[1][3pt]{{%
  \hbox{\rule[\scriptratio\dimexpr\fontdimen22\textfont2-.2pt\relax]
    {\scriptratio\dimexpr#1\relax}{\scriptratio\dimexpr.4pt\relax}}%
  \mkern-4mu\hbox{\let\f@size\sf@size\usefont{U}{lasy}{m}{n}\symbol{41}}}}

\newtcolorbox{myt}[2][]{%
  attach boxed title to top center={yshift=-4pt},
  colback=blue!5!white, colframe=blue!75!black,
  halign=flush left, fonttitle=\bfseries\sffamily,
  colbacktitle=blue!65!black, title=#2,#1, enhanced,
}
\newtcolorbox{myd}[2][]{%
  attach boxed title to top center={yshift=-4pt},
  colback=violet!5!white, colframe=violet!75!black,
  halign=flush left, fonttitle=\bfseries\sffamily,
  colbacktitle=violet!65!black, title=#2,#1, enhanced,
}
\newtcolorbox{mye}[2][]{%
  attach boxed title to top center={yshift=-4pt},
  colback=purple!5!white, colframe=purple!75!black,
  halign=flush left, fonttitle=\bfseries\sffamily,
  colbacktitle=purple!65!black, title=#2,#1, enhanced,
}
\newtcolorbox{myg}[2][]{%
  attach boxed title to top center={yshift=-4pt},
  colback=green!5!white, colframe=green!75!black,
  halign=flush left, fonttitle=\bfseries\sffamily,
  colbacktitle=green!65!black, title=#2,#1, enhanced,
}
\def\bmyd{\begin{myd}{} \begin{definition}}
\def\emyd{\end{definition}\end{myd}}
\def\bmyl{\begin{myg}{} \begin{lemma}}
\def\emyl{\end{lemma}\end{myg}}
\def\bmyt{\begin{myt}{} \begin{theorem}}
\def\emyt{\end{theorem}\end{myt}}
\def\bmyc{\begin{myg}{} \begin{corollary}}
\def\emyc{\end{corollary}\end{myg}}

\newcommand{\explain}[1]{{\color{red}\text{#1}\rightarrow}}
\newcommand{\explainabove}[2]{{\color{red}\overbrace{\color{black}#1}^{#2}}}

\usepackage{makecell}

\usepackage{anyfontsize}
\newcommand{\supportProjection}[1]{#1^0}
\newcommand{\terminal}[1]{#1_{\infty}}
\newcommand{\scalarTerminal}{\mathfrak{I}}
\newcommand{\currency}[1][]{\mathdollar_{#1}}
\def\gibbs{\gamma}
\usepackage{tikz}
\usetikzlibrary{calc}
\newcommand{\convto}[1][]{\succ^{\! #1}}

\newcommand{\cat}{\rm cat}
\newcommand{\bat}{\rm bat}

\newcommand{\yieldcat}[1][0]{\yield[#1]_{\cat}}
\newcommand{\yieldbat}[1][0]{\yield[#1]_{\bat}}

\newcommand{\monotone}{\mathsf{M}}
\newcommand{\region}{\mathfrak{R}}
\begin{document}
\expandafter\def\csname frontmatter@abstractwidth\endcsname{410pt}
	\title{Instability as a Quantum Resource}
	\date{\today}
	\author{\vspace{-5pt}Goni Yoeli}
	\email{yoeli.goni@campus.technion.ac.il}
  	\affiliation{%
		Department of Mathematics, Technion --- Israel Institute of Technology, Haifa 3200000, Israel
	}
	\author{Gilad Gour}%
	\email{gilad.gour@technion.ac.il}
	\affiliation{%
		Department of Mathematics, Technion --- Israel Institute of Technology, Haifa 3200000, Israel
	}%

\begin{abstract}
We consolidate coherence, athermality, and nonuniformity as sub-resources within an underlying quantum resource theory: \emph{instability}.
We formulate instability axiomatically as the transient information within a decaying physical system.
Specifying a decay mechanism (e.g., dephasing, thermalization) recovers these familiar resources as specific manifestations of instability.
We compute the one-shot distillation yield and dilution cost in various operational paradigms, and use them to pin down the extremal additive monotones.
In the asymptotic regime, we show that all conversion rates are governed by a single additive monotone, and thereby establish a universal second law for instability.
\end{abstract}

\expandafter\def\csname frontmatter@finalspace\endcsname{\addvspace{6pt}}
\maketitle

Dissipative evolution drives a physical system from an \emph{unstable} state \(\rho\) toward a steady state \(\terminal \rho\), thereby assigning clear roles: \(\rho\) carries a resource, while \(\terminal{\rho}\) is resource-free.
Mechanisms such as decoherence, thermalization, and depolarization manifest coherence, athermality, and nonuniformity
 as resources.
The corresponding resource theories have been investigated extensively, yet largely independently~\cite{BCP2014,WY2016,MS2016,CG2016,YMG+2016,Chitambar2018,SKW+2018,RFXA2018,ZLY+2018,ZLY+2019,TEZP2019,RNBG2020,HFW2021,HHO2003,GMN+2015,HO2013,BHOR2013,FOR2015,WW2019,Gour2022}.

The long-time dynamics \(\Delta:\rho \mapsto \terminal{\rho}\) defines an idempotent (\(\Delta\circ\Delta = \Delta\)) known as the ``resource-destroying map'' (RDM)~\cite{LHL2017}.
A resource theory admitting an RDM satisfies specific structural properties~\cite{LHL2017,Gour2017,LBT2019,CG2019,RBTL2020,Gour2025}.

Here we invert the RDM paradigm: rather than treating destruction as a feature of specific resources, we define \emph{instability} axiomatically as the resource induced by an arbitrary idempotent (``destruction'') channel.
Encompassing the full spectrum of destruction mechanisms, the resource theory of instability absorbs coherence, athermality, and nonuniformity (plus their conditional counterparts~\cite{NG2017,BRLW2017,MLA2019,JGW2025,GGH+2018,BGG2022,GWBG2024}) as sub-resource theories.
Previously treated in isolation, these resources are now recast as \emph{forms of instability} and become interconvertible; for example, coherence can transform into athermality.
  \begin{figure}[b]
    \vspace{-15pt}
  \centering\includegraphics[width=0.4\textwidth]{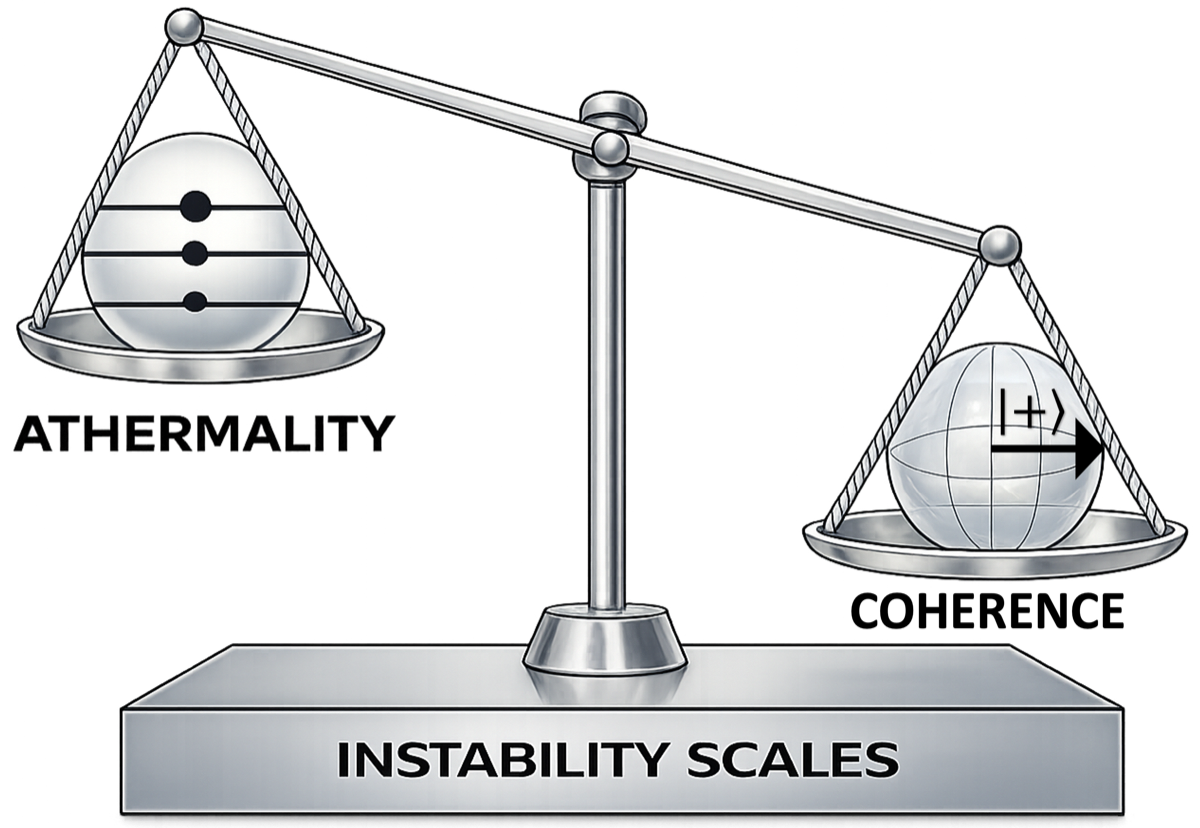}\vspace{-11pt}
  \caption{
   \!\!Athermality and coherence are now formally comparable as \emph{forms of instability}.
  \vspace{-13pt}}
\end{figure}

Diverse tasks such as purity (i.e., nonuniformity) concentration~\cite{HHO2003,GMN+2015}, gambling with correlated sources~\cite{BGG2022}, and conditional thermodynamic work extraction~\cite{BRLW2017,MLA2019,JGW2025} thus emerge as instances of instability conversion.
The operational hierarchy between instability resources subsumes familiar forms of majorization (standard~\cite{HLP1929,LY1999,Partovi2011,FGG2013}, relative~\cite{Blackwell1953,BST2019}, and conditional~\cite{BGG2022,GWBG2024}).

Despite its wide scope, the resource theory of instability remains tightly structured.
Among the extensive measures of instability, we identify the minimal and maximal ones and construct a three-parameter family in between.
Operationally, we compute distillation yield and dilution cost in the one-shot regime and then show that their asymptotic rates coincide.
This establishes a universal second law for instability.

\paragraph*{Framework.\!\!}
A quantum resource is defined relative to additional structure on the bare Hilbert space;
  entanglement relies on a partition into tensor factors, coherence on a distinguished orthonormal basis, and athermality on a Gibbs state (equivalently, on a Hamiltonian and environment temperature~\cite{BHOR2013}).
Unlike entanglement (see~\cite{Gour2017,LBT2019}), the extra structure in coherence and athermality can be encoded by an idempotent channel; a \emph{dephaser} \(\rho \mapsto \sum_x \bra{x} \rho \ket x \ket{x}\! \bra x\) determines the basis \(\pc{\ket{x}}\), and a \emph{replacer} \(\rho \mapsto \tr [\rho] \gibbs\) targets the Gibbs state \(\gibbs\).

The resource theory of instability therefore considers systems as pairs \((A,\Delta^{\!A})\), where \(A\cong \C^d\) and \(\Delta^{\!A}\) is an idempotent ``destruction'' channel.
To exclude degenerate regimes such as zero-temperature thermodynamics, 
we assume that \(\Delta^{\!A}\) fixes at least one invertible state.
Mathematically, this places us within the well-established theory of conditional expectations (details in Appendix~\ref{appendix:conditional-expectation}).

A quantum resource theory is defined by its \emph{free channels}---those operations allowed by the underlying constraint~\cite{Gour2025,CFS2016}. 
In instability theory free channels respect destruction. 
A channel \(\mathcal{N}^{A \to B}\) can circumvent destruction in two ways: by destabilizing fixed states or by protecting information~\cite{kni2006}.

\!\!\!\!\! The first is excluded by resource-non-generation (RNG): \(\mathcal{N}\) must map \(\Delta^{\!A}\)-fixed (free) states to \(\Delta^{\!B}\)-fixed states.
Since the fixed states span the image of an idempotent, this reads 
\vskip -18pt
\begin{equation} \label{eq:rng}
  \mathcal{N}(\ima \Delta^{\!A}) \subseteq \ima \Delta^{\!B}\;.
\end{equation}

To demonstrate protection, consider two noisy qubit registers in a hybrid platform: \(Q_1\) suffers depolarization (i.e., decays toward the uniform state \(\u^{Q_1}=I^{Q_1}/2\)), while \(Q_2\) dephases in \(\pc{\ket 0,\ket 1}\).
A classical bit \(x\in\{0,1\}\) is encoded as \(\ket{x}\!\bra{x}^{Q_1}\), where it is sensitive to depolarization.
To secure it, we can transfer the qubit state to \(Q_2\):
\(\ket{x}\!\bra{x}^{Q_2}\) is already incoherent, so \(Q_2\)'s dephasing leaves it unchanged.
Excluding such protection requires \emph{nonactivation}: information doomed to be lost must not be rescued.
Formally, if \(\Delta^{\!A}(\rho)=\Delta^{\!A}(\sigma)\), then \(\Delta^{\!B}(\mathcal N(\rho))=\Delta^{\!B}(\mathcal N(\sigma))\), or equivalently~\cite{YMG+2016,LHL2017}, 
\vskip -13pt
\begin{equation}\label{eq:nonactivation}
  \mathcal N(\ker\Delta^{\!A})\subseteq\ker\Delta^{\!B}\;.
\end{equation}
\vskip -3pt
  In the Heisenberg picture, where dual maps \(\mathcal M^*\) act on observables according to \(\tr[\rho\,\mathcal M^*(X)]=\tr[\mathcal M(\rho)X]\), nonactivation~(\ref{eq:nonactivation}) reads \(\mathcal{N}^*(\ima \Delta^{\!B*}) \subseteq \ima \Delta^{\!A*\!\!}\).
  Equivalently, \(\mathcal{N}^*\) sends
  observables fixed by \(\Delta^{\!B*}\) to those fixed by \(\Delta^{\!A*}\)~\cite{TEZP2019,RNBG2020}.
  Thus, if these fixed observables are regarded as free, nonactivation is the Heisenberg-dual requirement of RNG~\cite{LR2024,Gour2017}. 

Together, these requirements~(\ref{eq:rng})--(\ref{eq:nonactivation}) single out the \emph{destruction-covariant} channels~\cite{LHL2017,RNBG2020},
\vskip -13pt
\begin{align}\label{eq:dcov}
 \mathcal{N}^{A \to B}\;&:& \mathcal{N} \circ \Delta^{\!A}=
 \Delta^{\!B}\circ\mathcal{N}\,,
 \tag{\textsc{DC}}
\end{align}
\vskip -3pt
as the free operations of instability.

When \(\Delta^{\!A}\) and \(\Delta^{\!B}\) are both dephasers, Eq.~(\ref{eq:dcov}) specifies the dephasing-covariant operations (DIO)~\cite{CG2016,MS2016}; when they are replacers, it yields the Gibbs-preserving operations (GPO)~\cite{FOR2015,WW2019}.
Thus instability subsumes prominent resource theories such as coherence under DIO and athermality under GPO.
Figure~\ref{fig:reach} organizes these \emph{forms of instability} and displays their corresponding destruction mechanisms.
\begin{figure}[b]
  \centering
  \vskip -5pt
  \includegraphics[width=\columnwidth]{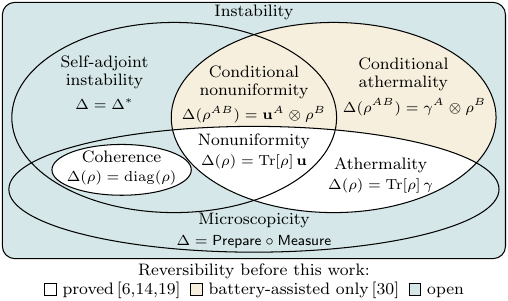}
  \setlength{\abovecaptionskip}{-10pt}
  \setlength{\belowcaptionskip}{-13pt}
  \caption{Venn diagram of forms of instability and their defining destruction mechanisms.
  The destruction mechanisms of coherence and (conditional) nonuniformity are self adjoint; the general self-adjoint subtheory was unified in~\cite{GR2024} as ``coherence relative to a subalgebra.''
 (Conditional) athermality is self adjoint exactly when \(\gamma\) is uniform (infinite temperature), in which case it reduces to nonuniformity. 
  Augmenting a thermodynamic system \(A\) with a perfect quantum memory \((B,\Delta^B=\id^B)\) gives conditional athermality whose \ref{eq:dcov} operations coincide with \emph{quantum-feedback-assisted GPO}~\cite{JGW2025}; trivial memory \(B\cong\mathbb C\) recovers standard athermality.  
Generalizing coherence and athermality, \emph{microscopicity} is information that a specified Bayesian retrodiction process fails to recover and instead progressively erodes~\cite{bss+2024,NWKB2024,NWKB2025}.}
  \label{fig:reach}
\end{figure}

Because Eq.~(\ref{eq:dcov}) also applies when \(A\) and \(B\) have different destruction mechanisms, the free channels in instability theory interlink these distinct forms.
In our working example, \(Q_1\)'s instability is nonuniformity, whereas \(Q_2\)'s is coherence. Let
\begin{equation}
  \mathcal{N}^{Q_1 \to Q_2}(\rho) \coloneqq
  \bra 0 \rho \ket 0 \ket +\! \bra +
  + \bra 1 \rho \ket 1 \ket -\! \bra -\;,
\end{equation}
where \(\ket \pm \coloneqq \frac{1}{\sqrt{2}}(\ket 0 \pm \ket 1)\). \(\mathcal{N}^{Q_1 \to Q_2}\) obeys Eq.~\eqref{eq:dcov} and maps \(\ket 0\! \bra 0^{Q_1}\) to \(\ket +\! \bra +^{Q_2}\), thus converting \(Q_1\)'s nonuniformity into \(Q_2\)'s coherence.
If the platform contains a third, thermalizing register \(Q_3\), a free \(Q_2\to Q_3\) channel can further convert \(Q_2\)'s coherence into \(Q_3\)'s athermality.


\paragraph*{Quantifying instability.\!}
Let \(\states\) be the collection of all states (across all systems).
A \emph{monotone} is a function \(\monotone:\states \to \R\) that does not increase under free channels.
On a thermodynamic system with Gibbs state \(\gibbs\), \(\monotone\) takes the form of a quantum divergence~\cite{GT2020}:
\(\monotone(\rho)=\D_\monotone\divergence{\rho}{\gibbs}\).
Conversely, any quantum divergence \(\D\) extends to two canonical instability monotones~\cite{LHL2017}:
\begin{align}
  \D&\divergence{\rho}{\Delta(\rho)} \;,\;\text{and}\label{eq:endpoint-local}\\
 \D&\divergence{\rho}{\free} \coloneqq \inf_{\sigma\in\free}\D\divergence{\rho}{\sigma} = \inf_{\sigma\in\states}\D\divergence{\rho}{\Delta(\sigma)}\;.\label{optimized-extension}
\end{align}
\vspace{-11pt}
\\
Here \(\Delta:\states \to \states\) denotes the collective destruction map, acting on each system \(A\) by its own \(\Delta^{\!A}\), and \(\free\coloneqq \Delta(\states)\) is the family of fixed states across all systems.
While \(\D\divergence{\rho}{\free}\) (known generally as \emph{the resource divergence}) is monotone under any RNG channel~\cite{VP1998,Gour2025},
monotonicity of \(\D\divergence{\rho}{\Delta(\rho)}\) requires \ref{eq:dcov}.
In Appendix~\ref{appendix:interpolation-optimized-direct} we bridge from \(\D\divergence{\rho}{\free}\) up to \(\D\divergence{\rho}{\Delta(\rho)}\) with intermediate monotone extensions.

We focus on the family of \(\alpha\)--\(z\) Rényi divergences~\cite{AD2015}:
\footnotetext{Trace functionals of the form \(\tr[(XS^rX^*)^z]\) are evaluated directly when \(S\) is invertible, and otherwise by the corresponding semicontinuous extension; see Appendix~\ref{appendix:optimization}.}
\begin{equation}\label{eq:alpha-z-renyi}
  \!\!D_{\alpha,z} \divergence{\rho}{\gibbs} \!=\! \frac{1}{\!\alpha\!-\!1\!} \!\log \tr \pb[\big]{\pa[\big]{\rho^{\frac{\alpha}{2z}} \!\gibbs^{\frac {1-\alpha} z}\! \rho^{\frac{\alpha}{2z}}}^{z}}, ~ (\alpha,z)\! \in \region\,,\!
\end{equation}
where all powers of operators are taken over their support and \(\region\) is the parameter region in which the data-processing inequality holds~\cite{Zhang2020, Note1}.
For particular parameter values within \(\region\), \(D_{\alpha,z}\) reduces to the Petz-Rényi family \(\bar D_\alpha = D_{\alpha,1}\)~\cite{Petz1985}, the sandwiched Rényi family \(\tilde D_\alpha =D_{\alpha,\alpha}\)~\cite{MDS+2013}, and, as limiting cases, to the Umegaki, min-, and max-relative entropies~\cite{Datta2009}:
\begin{align}\label{eq:umegaki}
  D \divergence{\rho}{\gibbs}&= \bar D_{1} \divergence{\rho}{\gibbs}= \tr [\rho \log \rho] - \tr [\rho \log \gibbs]\;,
  \\
\dmin \divergence{\rho}{\gibbs} &= \bar D_{0}\divergence{\rho}{\gibbs}=   - \log \tr [\gibbs \rho^0]\;, \label{dmin-heisenberg}
   \\ \label{eq:def-dmax}
  \dmax \divergence{\rho}{\gibbs} &= \tilde D_\infty \divergence{\rho}{\gibbs} = \min \pc*{m\ge 0 \smallgiven 2^m\gibbs \ge \rho}\;.
\end{align}
Here \(\rho^0\) is the projector onto the support of \(\rho\) and ``\( \ge\)''  denotes the Löwner order. 

For any \(\alpha\)--\(z\) Rényi divergence, the infimum in Eq.~\eqref{optimized-extension} is attained~\cite{RubboliTomamichel2024}: \(D_{\alpha,z}\divergence{\rho}{\free} = D_{\alpha,z}\divergence{\rho}{\sigma_*}\) for some minimizing free state \(\sigma_*\).
In Appendix~\ref{appendix:optimization}, we characterize \(\sigma_*\) as a solution of an implicit equation.

Closed forms for \(\sigma_*\) and \(D_{\alpha,z}\divergence{\rho}{\free}\) are available in particular cases.
 In the Umegaki case~(\ref{eq:umegaki}), \(\sigma_*= \Delta(\rho)\),
as follows from the chain rule \(D\divergence{\rho}{\Delta(\sigma)}
=D\divergence{\rho}{\Delta(\rho)}
+D\divergence{\Delta(\rho)}{\Delta(\sigma)}\) (Lemma 3.4 in~\cite{JLR2023}).
For the general Petz-Rényi case (\(z=1\)), we find \(\sigma_*\) explicitly and extend the Umegaki chain rule accordingly: \(\bar D_\alpha \divergence{\rho}{\Delta(\sigma)} = \bar D_\alpha \divergence{\rho}{\sigma_*} + \bar D_\alpha \divergence{\sigma_*}{\Delta(\sigma)}\). Evaluating at \(\sigma_*\) yields
\vskip -13pt 
\begin{equation}\label{eq:closed-form}
  \bar D_{\alpha} \divergence{\rho}{\free} 
  = \frac{1}{\alpha -1} \log \norm*{\Delta^{\!*}\!\pa[\Big]{{\Delta(I)}^{-\frac \alpha 2}\rho^\alpha{\Delta(I)}^{-\frac \alpha 2}}}
  _{\frac{1}{\alpha}}\,.
\end{equation}
\vskip -5pt
Here \(\norm{\cdot}_p\) denotes the Schatten \(p\)-norm (or quasinorm) for \(p \in [1/2,\infty]\) (see Appendix~\ref{appendix:renyi}).
In the particular case of \(\dmin\) (\(\alpha=0\)), Eq.~(\ref{eq:closed-form}) and its proof simplify significantly: by definition
\vskip -17pt
\begin{align}\label{eq:dmin-dual-start}
  \dmin\divergence{\rho}{\free}
  &= -\log\sup_{\sigma\in\states}\tr[\Delta(\sigma)\rho^0] \\
  &= -\log\sup_{\sigma\in\states}\tr[\sigma\,\Delta^{\!*}\!\pa*{\rho^0}] \\
  &= -\log\bigl\|\Delta^{\!*}\!\pa*{\rho^0}\bigr\|_\infty\;,\label{eq:dmin-dual-final}
\end{align}
where \(\norm{X}_\infty\) is the largest eigenvalue of a positive operator \(X\).
\(\dmin\) admits a smoothed, operational counterpart: the \emph{\(\epsilon\)-hypothesis-testing divergence},
\vskip -17pt
\begin{equation}\label{eq:def-hypo-testing}
D_{\!H}^\epsilon\divergence{\rho}{\gibbs}
=
-\log\!\!\!\!\!\!\min_{\begin{smallmatrix}
0\le\Gamma\le I\\
\tr[\rho\Gamma]\ge 1-\epsilon
\end{smallmatrix}}\!\!\!\!\!\!
\tr[\gibbs\Gamma]\;,\qquad \epsilon\in[0,1]\;.
\end{equation}
\vskip -3pt
\noindent
It quantifies how well \(\gibbs\) can be distinguished from \(\rho\) by a single test (\emph{quantum effect} \(0\le\Gamma\le I\)) in asymmetric settings~\cite{WR2012}.
Proceeding as in Eqs.~\eqref{eq:dmin-dual-start}--\eqref{eq:dmin-dual-final} (plus a standard minimax argument), we obtain
\begin{align}\label{eq:hypothesis-heisenberg}
  D_{\!H}^\epsilon\divergence{\rho}{\free}
  =
  -\log\!\!\!\!\!\min_{\begin{smallmatrix}
0\le\Gamma\le I\\
\tr[\rho\Gamma]\ge 1-\epsilon
\end{smallmatrix}}\!\!\!\!
\bigl\|\Delta^{\!*}\!\pa*{\Gamma}\bigr\|_\infty\;.
\end{align}\vspace{-7pt}
\\
We also consider the restricted hypothesis-testing divergence, denoted \(D_{\!H}^{\epsilon,\scalarTerminal\! }\divergence{\cdot}{\cdot}\), where the discriminating effects are restricted to a subset
  \(
  \scalarTerminal\!
  \)~\cite{BHA+2020,BBG+2023}.
Specifically, we take \(\scalarTerminal\) to be the effects \(\Gamma\) that cannot separate free states: \(\tr [\sigma \Gamma]\) constant over \(\sigma\in \free\).
With this we can unambiguously define 
\begin{align}\label{def:restricted-HT}
  h^\epsilon (\rho) &\coloneqq D_{\!H}^{\epsilon,\scalarTerminal }\divergence{\rho}{\sigma}=
  -\log\!\!\!\!\!\min_{\begin{smallmatrix}
\Gamma\in\scalarTerminal\\
\tr[\rho\Gamma]\ge 1-\epsilon
\end{smallmatrix}}\!\!\!\!\!\!
\tr[\sigma\Gamma]\;,
\end{align}
\vspace{-7pt}
\\
where \(\sigma\) can be any free state.
Clearly,
\(h^\epsilon(\rho)\le D_{\!H}^\epsilon \divergence{\rho}{\free}\), since restricting the optimization in Eq.~(\ref{eq:def-hypo-testing}) decreases the resulting divergence;
on a thermodynamic system, where only the Gibbs state is free, the restriction is vacuous and \(h^\epsilon(\rho)\) reduces to \(D_{\!H}^\epsilon \divergence{\rho}{\gibbs}\).

\paragraph*{One-shot analysis.\!\!}
For \(\rho,\sigma\in\states\) and \(\epsilon\in[0,1]\), we write \(\rho\convto[\epsilon]\sigma\) if \(\mathcal{N}(\rho)\approx_\epsilon\sigma\) for some free channel \(\mathcal N\) (i.e., \(\mathcal{N}(\rho)\) is \(\epsilon\)-close to \(\sigma\) in trace distance, \(\frac 12 \norm{\mathcal{N}(\rho) - \sigma}_1 \le \epsilon\)).
The exact-conversion preorder \(\convto[0]\) reduces to relative majorization in athermality~\cite{Blackwell1953} and recovers (conditional) majorization in (conditional) nonuniformity~\cite{JGW2025,BGG2022,GWBG2024}. 

Distillation and dilution of quantum resources are special conversion tasks defined relative to a \emph{currency}---a parameterized family of reference resource states \(\pc{\currency[m]}\subset \states\).
For error tolerance \(\epsilon \ge 0\), the \(\epsilon\)\emph{-one-shot distillable yield} and \(\epsilon\)\emph{-one-shot dilution cost} of \(\rho\in \states\) are:
\begin{align}\label{def:one-shot-yield}
  \yield[\epsilon](\rho) &= \sup\pc*{m \smallgiven   \rho \convto[\epsilon] \currency[m]} \;,
  \\
  \cost[\epsilon](\rho) & =  \hspace{1.4pt}\inf \hspace{1.4pt}\pc*{m \smallgiven  \currency[m]\! \convto[\epsilon]\! \rho~}\;.\label{def:one-shot-cost}
\end{align}

Individual subtheories use characteristic currencies, e.g., maximally coherent states in coherence, pure states in nonuniformity, and maximally entangled states in conditional nonuniformity.
Which currency, then, should instability adopt?
Resource-theoretically, this choice is immaterial: the corresponding resource states of these different currency families are reversibly interconvertible (see Appendix~\ref{appendix:numeraire}).
We therefore may as well pick the simplest among these families: the thermodynamic currency \(\currency[m]\coloneqq\ket 0\!\bra 0\) on a classical two-level system with Gibbs state \(\gibbs_m=2^{-m}\ket 0\!\bra 0+(1-2^{-m})\ket 1\!\bra 1\)~\cite{WW2019}. 
With this choice, distillation reduces to a binary measurement (POVM) and dilution to a preparation (cq) channel.
\begin{lemmaNS}\label{thm:one-shot}
   Let \(\epsilon \in [0,1]\) and \(\rho \in \states\). Then,
\begin{align}
    \yield[\epsilon](\rho) & = h^{\epsilon}(\rho)\;, \label{eq:distill} \\
    \cost[\epsilon](\rho)&   =  \min_{\tau \approx_\epsilon \rho}\!\dmax\divergence{\tau}{\Delta(\tau)}\,.\label{eq:cost}
  \end{align}
\end{lemmaNS}\vspace{-5pt}
\noindent The derivations appear in Appendices~\ref{appendix:one-shot-yield} (Eq.~(\ref{eq:distill})) and~\ref{appendix:cost-derivation} (Eq.~(\ref{eq:cost})). Cf.~\cite{LBT2019,RBTL2020}.

\paragraph*{Multipartite instability.\!\!}
To extend instability theory beyond the one-shot level, we impose a composition rule:
\\
\AssumptionAnchor{locality}{Locality of destruction} \emph{The composite system \(AB\) is equipped with the composite destruction \(\Delta^{\!AB} \coloneqq \Delta^{\!A}\otimes\Delta^{\!B}\).}
This axiom turns instability into a proper quantum resource theory~\cite{CFS2016,ZS2025} (see Appendix~\ref{appendix:fragility-theory}) and renders the currency additive: \(\currency[m]\otimes \currency[t]\) and \(\currency[m+t]\) are equivalent resources (see Appendix~\ref{appendix:numeraire} or Ref.~\cite{WW2019}).
Locality is standard for destruction-based resource theories (e.g., coherence) but excludes others like asymmetry~\cite{GS2008,MS2013,MS2014}.
\paragraph*{\!Assisted distillation.\!\!}
With the multipartite framework fixed, we can consider catalyst-assisted conversion tasks.
Specifically, the \(\epsilon\)-\emph{one-shot catalytic yield} and the \(\epsilon\)-\emph{one-shot battery-assisted yield}~\cite{JGW2025} are defined as
\begin{align}\label{eq:def-yield-cat}
  \yieldcat[\epsilon](\rho)
  &\coloneqq
  \sup\pc*{ m \smallgiven \exists\,\tau\in\states\!:\ \rho\otimes\tau \convto[\epsilon] \currency[m]\otimes\tau},\! \\
\label{eq:def-yield-bat}
   \yieldbat[\epsilon](\rho)
  &\coloneqq
  \sup \pc*{m-t \smallgiven \rho\otimes\currency[t]\convto[\epsilon]\currency[m]}\,.
\end{align}
The latter task is a special case of catalytic distillation with the catalyst restricted to the currency family (see Appendix~\ref{appendix:catalytic}).
Therefore, we have
\begin{equation}\label{eq:inequality-of-yield}
  \yield[\epsilon](\rho) \le \yieldbat[\epsilon](\rho) \le \yieldcat[\epsilon](\rho)\;.
\end{equation}
Without catalysts, \(\yield[\epsilon](\rho) \le D_{\!H\!}^\epsilon \divergence{\rho}{\free}\) (by Eq.~(\ref{eq:distill})).
Catalysis attains this bound:
\begin{restatable}{lemmaNS}{LemmaCatalytic}\label{lemma-catalytic}
  Let \(\epsilon \in [0,1]\) and \(\rho \in \states\). Then,
  \begin{align}\label{eq:catalytic}
   \yieldbat[\epsilon](\rho)
  = \yield[\epsilon](\rho\otimes\currency[1])-1=  D_{\!H}^\epsilon\divergence{\rho}{\free}  \,.
\end{align}
At \(\epsilon=0\), catalytic and battery-assisted yields coincide:
\begin{equation} \label{eq:exact-catalytic}
  \yieldcat[0](\rho)=\yieldbat[0](\rho)=\dmin \divergence{\rho}{\free}\;.
\end{equation}
\end{restatable}
\noindent
The proof is given in Appendix~\ref{appendix:catalytic} and follows Ref.~\cite{JGW2025}.
\paragraph*{Additive monotones.\!}
A monotone \(\monotone\) is called \emph{additive} if \(\monotone(\rho \otimes \tau)=\monotone(\rho)+\monotone(\tau)\) for all \(\rho,\tau\in\states\). 
\(\monotone\) is called \emph{weakly additive} if \(\monotone(\rho^{\otimes n}) = n \monotone(\rho)\) for all \(\rho\in\states\) and all \(n \in \mathbb{N}\).
Weakly additive monotones abound in every resource theory: \emph{regularizing} a monotone, \(\monotone_\infty(\rho) \coloneqq \lim_{n \to \infty} \frac{1}{n} \monotone(\rho^{\otimes n})\), yields one whenever the limit exists~\cite{HH1998,hht2001}.
Full additivity, on the other hand, has no such resource-independent recipe; known examples are scarce and resource-specific~\cite{cw2004,yhw2008,brt2026}.
In instability, a simple recipe arises from the locality of \(\Delta\): whenever \(\D\divergence{\cdot}{\cdot}\) is additive (e.g., any Rényi divergence), so is \(\D\divergence{\rho}{\Delta(\rho)}\).
Does the other extension, \(\D\divergence{\rho}{\free}\), inherit additivity as well?

For \(\bar D_\alpha \divergence{\rho}{\free}\), this is confirmed by Eq.~(\ref{eq:closed-form}):
the map \(\rho \mapsto \rho^{\alpha}\), the Schatten norm, the destruction map, and its dual \(\Delta^{\!*}\) all factorize on tensor products. 
The general \(D_{\alpha,z} \divergence{\rho}{\free}\) additivity follows from the characterization of the minimizer (Lemma~\ref{thm:optimized-renyi} in Appendix~\ref{appendix:renyi}).

Expanding this picture, we next introduce a three-parameter family \(\monotone_{\alpha,z}^\lambda\), running from \(D_{\alpha,z}\divergence{\rho}{\free}\) at \(\lambda=0\) up to \(D_{\alpha,z}\divergence{\rho}{\Delta(\rho)}\) at \(\lambda=1\).
Specifically, let  \(\displaystyle Y_\rho\coloneqq \rho^{\frac{\alpha}{2z}}\Delta(\rho)^{\lambda \frac{1-\alpha}{2z}}\)\!, let \(Y_\rho^*\) denote the adjoint, and define~\cite{Note1}  
\begin{align}\label{eq:3-parameter-entropy} 
&\monotone_{\alpha,z}^\lambda(\rho) \coloneqq \inf_{\sigma\in\free}\frac{1}{\alpha-1}\log \tr \pb*{(Y_\rho\sigma^{(1-\lambda)\frac{1-\alpha}z} Y_\rho^*)^{z}}\;.
\end{align}
\vskip -10pt
\begin{restatable}{theoremNS}{ThmRenyi}\label{thm:3-renyi}
  \(\monotone_{\alpha,z}^\lambda\) is an additive instability monotone for any \(\alpha,z\in\region\) and \(\lambda\in[0,1]\).
\end{restatable}
\noindent 
Our proof lifts the conditional-entropy framework of Refs.~\cite{RGT2025,rubboli-phd-2025}---essentially the conditional nonuniformity case---to general instability (see Appendix~\ref{appendix:renyi}).

\(\pc{\monotone_{\alpha,z}^\lambda}\) provides a wide family of additive monotones, but does not exhaust the class~\cite{mbv2024}. Yet its limiting elements turn out to be extremal even within the larger class of \emph{weakly} additive monotones (\emph{strictly} larger: see Claim~\ref{claim:weakly-additive-strict} in Appendix~\ref{appendix:extremals}).
Fixing a scale, we call a monotone \(\monotone\) \emph{normalized} if \(\monotone(\currency[1]) = 1\), aligning with~\cite{GT2020}.

\begin{restatable}{theoremNS}{ThmExtremal}\label{thm:extremal}
Let \(\monotone\) be a normalized weakly additive monotone of instability and let \(\rho\in \states\). Then,
\vskip -17pt
\begin{equation}\label{eq:optimal-entropies}
  \dmin\divergence{\rho}{\free} \;\le\;  \monotone(\rho)  \;\le\; \dmax\divergence{\rho}{\Delta(\rho)}
  \,.
\end{equation}
\end{restatable}
\vskip -2pt
 \noindent Specialized to athermality, Theorem~\ref{thm:extremal} yields the extremality of \(\dmin\) and \(\dmax\) among relative entropies~\cite{GT2020,GT2021}; in conditional nonuniformity, it reproduces the analogous claim for \emph{conditional entropies}~\cite{JGW2025}.

\paragraph*{Proof outline.\!\!\!}
  A general resource-theoretic argument bounds any normalized weakly additive monotone between the regularized exact yield and cost: \(\yield[0]_\infty \le \monotone \le \cost[0]_\infty\).
 The substantive step is that, in instability theory, these regularizations take the exact forms
 \begin{align}
  \yield[0]_\infty(\rho)& = \dmin\divergence{\rho}{\free}\;, \label{eq:exact-asym-yield}
  \\
\cost[0]_\infty(\rho) & = \dmax\divergence{\rho}{\Delta(\rho)}\;,
\label{eq:exact-asym-cost}
 \end{align}
 establishing the claim (full details in Appendix~\ref{appendix:extremals}).

\begin{figure}[h]
  \vskip -10pt
  \centering
  \includegraphics[width=\columnwidth]{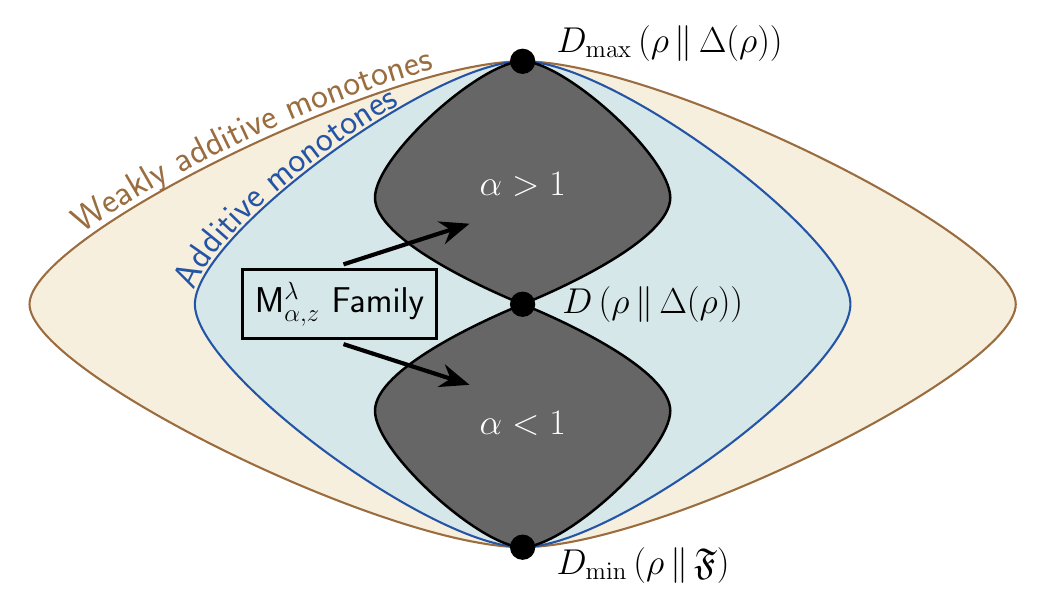}
    \setlength{\abovecaptionskip}{-13pt}
  \setlength{\belowcaptionskip}{-11pt}
  \caption{Schematic of instability monotones.
The \(\monotone_{\alpha,z}^{\lambda}\) family lies within the family of additive monotones, which in turn is a subset of the weakly additive monotones.
The limiting elements of the three-parameter family, \(\dmin\divergence{\rho}{\free}\) and \(\dmax\divergence{\rho}{\Delta(\rho)}\), are extremal even within the full class of (normalized) weakly additive monotones.
In the limit \(\alpha\to1\), the dependence on \(z\) and \(\lambda\) disappears, and the family collapses to \(D\divergence{\rho}{\Delta(\rho)}\).}
  \label{fig:additive-family}
\end{figure}

\paragraph*{Asymptotic regime.\!\!}
The regularizations \(\yield[0]_\infty\) and \(\cost[0]_\infty\) define the asymptotic rates of \emph{exact} distillation and dilution.
At error tolerance \(\epsilon\), they generalize to the \(\epsilon\)-\emph{achievable} rates \(\yield[\epsilon]_\infty\) and \(\cost[\epsilon]_\infty\)~\footnote{When the corresponding regularizations do not exist, we use the worst-case partial limits.}.
The \(\epsilon\to 0^+\) limits need not recover the exact rates; instead, they define the \emph{reliable} rates of distillation and dilution~\cite{Gour2025,trw2022}.
A resource theory is \emph{asymptotically reversible} when these reliable rates coincide, yielding a single measure that governs all asymptotic resource transitions.

Without enhancing the operational model~\cite{BP2008,BP2009,BP2010a,BG2015,RTNA2025}, most resource theories remain asymptotically \emph{irreversible} and typically admit only converse bounds on conversion rates~\cite{DH1999,SH2006,HOH2002}.
Entanglement, for instance, is irreversible beyond special classes such as pure states~\cite{HH1998,BPR+2000,VC2001,LR2023,BBPS1996,PR1997}.

Remarkably, several forms of instability were shown to be reversible, fully or with battery assistance (see Fig.~\ref{fig:reach}).
Our framework yields a stronger, fully general result.
\begin{theoremNS}\label{thm:second-law}
  Let \(\epsilon\in (0,1)\) and \(\rho\in \states\). Then 
 \begin{align} \label{eq:second-law}
 \yield[\epsilon]_\infty(\rho) =
   \cost[\epsilon]_\infty(\rho) = D \divergence{\rho}{\Delta(\rho)}\;.\!\!
 \end{align}
\end{theoremNS}
\noindent 
Theorem~\ref{thm:second-law} establishes full, unassisted asymptotic reversibility with \(D\divergence{\rho}{\Delta(\rho)}\) as the governing measure.
In fact, both \(\epsilon\)-achievable rates already equal this quantity for every fixed \(\epsilon\in(0,1)\).
Together, Eqs.~(\ref{eq:exact-asym-yield})--(\ref{eq:second-law}) expose the standard strong-converse profile~\cite{trw2022,pw2025}: \(\yield[\epsilon]_\infty\) and \(\cost[\epsilon]_\infty\) are constant for \(\epsilon\in(0,1)\) but discontinuous at \(\epsilon=0\), highlighting the operational difference between exact and reliable conversions.
The two equalities in Eq.~\eqref{eq:second-law} are proved separately (Appendices~\ref{appendix:asymptotic-distillable-fragility} and~\ref{appendix:cost-bounds}).

\paragraph*{Conclusion.\!\!}
We have consolidated the framework of resource destruction channels into a single quantum resource theory of instability.
We recognize coherence, athermality, and nonuniformity, along with their conditional extensions, as forms of one underlying resource.
By abstracting away mechanism-specific complexity, this unified perspective yields strong general results, including full asymptotic reversibility.

\paragraph*{Open directions.\!\!}
Much of our one-shot analysis extends to affine resources~\cite{Gour2017}, for which the destruction map need not be physical, as in imaginarity~\cite{HG2018}.
While unifying these resources is possible~\cite{GJJ+2026}, a \emph{natural} unification that additionally extends instability remains elusive.

A second problem concerns restrictions on free channels beyond the rather permissive law of destruction covariance.
In athermality, \ref{eq:dcov} channels are the GPO, which outperform the standard subclass of \emph{thermal operations}~\cite{BHOR2013,FOR2015,TT2025}.
In asymmetry, conventional theory demands covariance under every group element; instability, only under destruction---the group twirl~\cite{GS2008}.
For coherence, \emph{strictly incoherent operations} are often adopted instead of DIO~\cite{YMG+2016,LHL2017}.
Further restrictions arise in specific settings~\cite{BRLW2017,NG2017,RLS2018,MLA2019}.
We expect some such classes to fit into a generalization of instability theory.

Lastly, instability extends naturally to the dynamical domain~\cite{GW2019,LY2020,GS2020}, where channels themselves are resources and \emph{superchannels}~\cite{CDP2008,CDP2009,Gour2019} play the role of free operations.
For example, let \(A\) and \(B\) be \emph{static} admissible systems and define a \emph{destruction superchannel}:
\begin{equation}
  \Delta^{\!A \to B}[\mathcal{N}^{A \to B}]  \coloneqq \Delta^{\! B} \circ \mathcal{N}^{A \to B} \circ \Delta^{\!A}\;.
\end{equation}
If both \(\Delta^{\!A}\) and \(\Delta^{\!B}\) are dephasers, the instability-free channels are `classical' and the destruction-covariant superchannels are called DISC~\cite{SCG2020}.
Alternatively, if \(\Delta^{\!A}\) is the identity channel and \(\Delta^{\!B}\) is the depolarizer, the only instability-free channel is the uniform channel \(\rho^A \mapsto \tr[\rho^A]\u^B\), and the free operations are the \emph{completely uniformity-preserving superchannels}~\cite{Gour2019,GKN+2025}.

The definitions of instability yield and instability cost, as well as several one-shot results, extend to higher dynamical orders and generalized probabilistic theories~\cite{SZS2025,TR2019,Jencova2021} once the distance measure is adjusted.
A key next step is to determine which further features of static-system instability (e.g., reversibility) survive in these broader regimes.
\paragraph*{Acknowledgments.}
We thank Kaiyuan Ji for insightful correspondence that helped clarify key concepts. We are also grateful to Raz Firanko and Yosi Avron for reviewing the manuscript and providing valuable feedback.

\nocite{BP2010b,CFL2018,HT2016,HY2025,KA1979,Tomiyama1957,gour2025b,lami2024}
\bibliography{references}

\begin{thebibliography}{119}%
\makeatletter
\providecommand \@ifxundefined [1]{%
 \@ifx{#1\undefined}
}%
\providecommand \@ifnum [1]{%
 \ifnum #1\expandafter \@firstoftwo
 \else \expandafter \@secondoftwo
 \fi
}%
\providecommand \@ifx [1]{%
 \ifx #1\expandafter \@firstoftwo
 \else \expandafter \@secondoftwo
 \fi
}%
\providecommand \natexlab [1]{#1}%
\providecommand \enquote  [1]{``#1''}%
\providecommand \bibnamefont  [1]{#1}%
\providecommand \bibfnamefont [1]{#1}%
\providecommand \citenamefont [1]{#1}%
\providecommand \href@noop [0]{\@secondoftwo}%
\providecommand \href [0]{\begingroup \@sanitize@url \@href}%
\providecommand \@href[1]{\@@startlink{#1}\@@href}%
\providecommand \@@href[1]{\endgroup#1\@@endlink}%
\providecommand \@sanitize@url [0]{\catcode `\\12\catcode `\$12\catcode
  `\&12\catcode `\#12\catcode `\^12\catcode `\_12\catcode `\%12\relax}%
\providecommand \@@startlink[1]{}%
\providecommand \@@endlink[0]{}%
\providecommand \url  [0]{\begingroup\@sanitize@url \@url }%
\providecommand \@url [1]{\endgroup\@href {#1}{\urlprefix }}%
\providecommand \urlprefix  [0]{URL }%
\providecommand \Eprint [0]{\href }%
\providecommand \doibase [0]{https://doi.org/}%
\providecommand \selectlanguage [0]{\@gobble}%
\providecommand \bibinfo  [0]{\@secondoftwo}%
\providecommand \bibfield  [0]{\@secondoftwo}%
\providecommand \translation [1]{[#1]}%
\providecommand \BibitemOpen [0]{}%
\providecommand \bibitemStop [0]{}%
\providecommand \bibitemNoStop [0]{.\EOS\space}%
\providecommand \EOS [0]{\spacefactor3000\relax}%
\providecommand \BibitemShut  [1]{\csname bibitem#1\endcsname}%
\let\auto@bib@innerbib\@empty
\bibitem [{\citenamefont {Baumgratz}\ \emph {et~al.}(2014)\citenamefont
  {Baumgratz}, \citenamefont {Cramer},\ and\ \citenamefont {Plenio}}]{BCP2014}%
  \BibitemOpen
  \bibfield  {author} {\bibinfo {author} {\bibfnamefont {T.}~\bibnamefont
  {Baumgratz}}, \bibinfo {author} {\bibfnamefont {M.}~\bibnamefont {Cramer}},\
  and\ \bibinfo {author} {\bibfnamefont {M.~B.}\ \bibnamefont {Plenio}},\
  }\href {https://doi.org/10.1103/PhysRevLett.113.140401} {\bibfield  {journal}
  {\bibinfo  {journal} {Phys. Rev. Lett.}\ }\textbf {\bibinfo {volume} {113}},\
  \bibinfo {pages} {140401} (\bibinfo {year} {2014})}\BibitemShut {NoStop}%
\bibitem [{\citenamefont {Winter}\ and\ \citenamefont {Yang}(2016)}]{WY2016}%
  \BibitemOpen
  \bibfield  {author} {\bibinfo {author} {\bibfnamefont {A.}~\bibnamefont
  {Winter}}\ and\ \bibinfo {author} {\bibfnamefont {D.}~\bibnamefont {Yang}},\
  }\href {https://doi.org/10.1103/PhysRevLett.116.120404} {\bibfield  {journal}
  {\bibinfo  {journal} {Phys. Rev. Lett.}\ }\textbf {\bibinfo {volume} {116}},\
  \bibinfo {pages} {120404} (\bibinfo {year} {2016})}\BibitemShut {NoStop}%
\bibitem [{\citenamefont {Marvian}\ and\ \citenamefont
  {Spekkens}(2016)}]{MS2016}%
  \BibitemOpen
  \bibfield  {author} {\bibinfo {author} {\bibfnamefont {I.}~\bibnamefont
  {Marvian}}\ and\ \bibinfo {author} {\bibfnamefont {R.~W.}\ \bibnamefont
  {Spekkens}},\ }\href {https://doi.org/10.1103/PhysRevA.94.052324} {\bibfield
  {journal} {\bibinfo  {journal} {Phys. Rev. A}\ }\textbf {\bibinfo {volume}
  {94}},\ \bibinfo {pages} {052324} (\bibinfo {year} {2016})}\BibitemShut
  {NoStop}%
\bibitem [{\citenamefont {Chitambar}\ and\ \citenamefont
  {Gour}(2016)}]{CG2016}%
  \BibitemOpen
  \bibfield  {author} {\bibinfo {author} {\bibfnamefont {E.}~\bibnamefont
  {Chitambar}}\ and\ \bibinfo {author} {\bibfnamefont {G.}~\bibnamefont
  {Gour}},\ }\href {https://doi.org/10.1103/PhysRevLett.117.030401} {\bibfield
  {journal} {\bibinfo  {journal} {Phys. Rev. Lett.}\ }\textbf {\bibinfo
  {volume} {117}},\ \bibinfo {pages} {030401} (\bibinfo {year}
  {2016})}\BibitemShut {NoStop}%
\bibitem [{\citenamefont {Yadin}\ \emph {et~al.}(2016)\citenamefont {Yadin},
  \citenamefont {Ma}, \citenamefont {Girolami}, \citenamefont {Gu},\ and\
  \citenamefont {Vedral}}]{YMG+2016}%
  \BibitemOpen
  \bibfield  {author} {\bibinfo {author} {\bibfnamefont {B.}~\bibnamefont
  {Yadin}}, \bibinfo {author} {\bibfnamefont {J.}~\bibnamefont {Ma}}, \bibinfo
  {author} {\bibfnamefont {D.}~\bibnamefont {Girolami}}, \bibinfo {author}
  {\bibfnamefont {M.}~\bibnamefont {Gu}},\ and\ \bibinfo {author}
  {\bibfnamefont {V.}~\bibnamefont {Vedral}},\ }\href
  {https://doi.org/10.1103/PhysRevX.6.041028} {\bibfield  {journal} {\bibinfo
  {journal} {Phys. Rev. X}\ }\textbf {\bibinfo {volume} {6}},\ \bibinfo {pages}
  {041028} (\bibinfo {year} {2016})}\BibitemShut {NoStop}%
\bibitem [{\citenamefont {Chitambar}(2018)}]{Chitambar2018}%
  \BibitemOpen
  \bibfield  {author} {\bibinfo {author} {\bibfnamefont {E.}~\bibnamefont
  {Chitambar}},\ }\href {https://doi.org/10.1103/PhysRevA.97.050301} {\bibfield
   {journal} {\bibinfo  {journal} {Phys. Rev. A}\ }\textbf {\bibinfo {volume}
  {97}},\ \bibinfo {pages} {050301} (\bibinfo {year} {2018})}\BibitemShut
  {NoStop}%
\bibitem [{\citenamefont {Streltsov}\ \emph {et~al.}(2018)\citenamefont
  {Streltsov}, \citenamefont {Kampermann}, \citenamefont {Wölk}, \citenamefont
  {Gessner},\ and\ \citenamefont {Bruß}}]{SKW+2018}%
  \BibitemOpen
  \bibfield  {author} {\bibinfo {author} {\bibfnamefont {A.}~\bibnamefont
  {Streltsov}}, \bibinfo {author} {\bibfnamefont {H.}~\bibnamefont
  {Kampermann}}, \bibinfo {author} {\bibfnamefont {S.}~\bibnamefont {Wölk}},
  \bibinfo {author} {\bibfnamefont {M.}~\bibnamefont {Gessner}},\ and\ \bibinfo
  {author} {\bibfnamefont {D.}~\bibnamefont {Bruß}},\ }\href
  {https://doi.org/10.1088/1367-2630/aac484} {\bibfield  {journal} {\bibinfo
  {journal} {New Journal of Physics}\ }\textbf {\bibinfo {volume} {20}},\
  \bibinfo {pages} {053058} (\bibinfo {year} {2018})}\BibitemShut {NoStop}%
\bibitem [{\citenamefont {Regula}\ \emph
  {et~al.}(2018{\natexlab{a}})\citenamefont {Regula}, \citenamefont {Fang},
  \citenamefont {Wang},\ and\ \citenamefont {Adesso}}]{RFXA2018}%
  \BibitemOpen
  \bibfield  {author} {\bibinfo {author} {\bibfnamefont {B.}~\bibnamefont
  {Regula}}, \bibinfo {author} {\bibfnamefont {K.}~\bibnamefont {Fang}},
  \bibinfo {author} {\bibfnamefont {X.}~\bibnamefont {Wang}},\ and\ \bibinfo
  {author} {\bibfnamefont {G.}~\bibnamefont {Adesso}},\ }\href
  {https://doi.org/10.1103/PhysRevLett.121.010401} {\bibfield  {journal}
  {\bibinfo  {journal} {Phys. Rev. Lett.}\ }\textbf {\bibinfo {volume} {121}},\
  \bibinfo {pages} {010401} (\bibinfo {year} {2018}{\natexlab{a}})}\BibitemShut
  {NoStop}%
\bibitem [{\citenamefont {Zhao}\ \emph {et~al.}(2018)\citenamefont {Zhao},
  \citenamefont {Liu}, \citenamefont {Yuan}, \citenamefont {Chitambar},\ and\
  \citenamefont {Ma}}]{ZLY+2018}%
  \BibitemOpen
  \bibfield  {author} {\bibinfo {author} {\bibfnamefont {Q.}~\bibnamefont
  {Zhao}}, \bibinfo {author} {\bibfnamefont {Y.}~\bibnamefont {Liu}}, \bibinfo
  {author} {\bibfnamefont {X.}~\bibnamefont {Yuan}}, \bibinfo {author}
  {\bibfnamefont {E.}~\bibnamefont {Chitambar}},\ and\ \bibinfo {author}
  {\bibfnamefont {X.}~\bibnamefont {Ma}},\ }\href
  {https://doi.org/10.1103/PhysRevLett.120.070403} {\bibfield  {journal}
  {\bibinfo  {journal} {Phys. Rev. Lett.}\ }\textbf {\bibinfo {volume} {120}},\
  \bibinfo {pages} {070403} (\bibinfo {year} {2018})}\BibitemShut {NoStop}%
\bibitem [{\citenamefont {Zhao}\ \emph {et~al.}(2019)\citenamefont {Zhao},
  \citenamefont {Liu}, \citenamefont {Yuan}, \citenamefont {Chitambar},\ and\
  \citenamefont {Winter}}]{ZLY+2019}%
  \BibitemOpen
  \bibfield  {author} {\bibinfo {author} {\bibfnamefont {Q.}~\bibnamefont
  {Zhao}}, \bibinfo {author} {\bibfnamefont {Y.}~\bibnamefont {Liu}}, \bibinfo
  {author} {\bibfnamefont {X.}~\bibnamefont {Yuan}}, \bibinfo {author}
  {\bibfnamefont {E.}~\bibnamefont {Chitambar}},\ and\ \bibinfo {author}
  {\bibfnamefont {A.}~\bibnamefont {Winter}},\ }\href
  {https://doi.org/10.1109/TIT.2019.2911102} {\bibfield  {journal} {\bibinfo
  {journal} {IEEE Transactions on Information Theory}\ }\textbf {\bibinfo
  {volume} {65}},\ \bibinfo {pages} {6441} (\bibinfo {year}
  {2019})}\BibitemShut {NoStop}%
\bibitem [{\citenamefont {Theurer}\ \emph {et~al.}(2019)\citenamefont
  {Theurer}, \citenamefont {Egloff}, \citenamefont {Zhang},\ and\ \citenamefont
  {Plenio}}]{TEZP2019}%
  \BibitemOpen
  \bibfield  {author} {\bibinfo {author} {\bibfnamefont {T.}~\bibnamefont
  {Theurer}}, \bibinfo {author} {\bibfnamefont {D.}~\bibnamefont {Egloff}},
  \bibinfo {author} {\bibfnamefont {L.}~\bibnamefont {Zhang}},\ and\ \bibinfo
  {author} {\bibfnamefont {M.~B.}\ \bibnamefont {Plenio}},\ }\href
  {https://doi.org/10.1103/PhysRevLett.122.190405} {\bibfield  {journal}
  {\bibinfo  {journal} {Phys. Rev. Lett.}\ }\textbf {\bibinfo {volume} {122}},\
  \bibinfo {pages} {190405} (\bibinfo {year} {2019})}\BibitemShut {NoStop}%
\bibitem [{\citenamefont {Regula}\ \emph
  {et~al.}(2020{\natexlab{a}})\citenamefont {Regula}, \citenamefont
  {Narasimhachar}, \citenamefont {Buscemi},\ and\ \citenamefont
  {Gu}}]{RNBG2020}%
  \BibitemOpen
  \bibfield  {author} {\bibinfo {author} {\bibfnamefont {B.}~\bibnamefont
  {Regula}}, \bibinfo {author} {\bibfnamefont {V.}~\bibnamefont
  {Narasimhachar}}, \bibinfo {author} {\bibfnamefont {F.}~\bibnamefont
  {Buscemi}},\ and\ \bibinfo {author} {\bibfnamefont {M.}~\bibnamefont {Gu}},\
  }\href {https://doi.org/10.1103/PhysRevResearch.2.013109} {\bibfield
  {journal} {\bibinfo  {journal} {Phys. Rev. Res.}\ }\textbf {\bibinfo {volume}
  {2}},\ \bibinfo {pages} {013109} (\bibinfo {year}
  {2020}{\natexlab{a}})}\BibitemShut {NoStop}%
\bibitem [{\citenamefont {Hayashi}\ \emph {et~al.}(2021)\citenamefont
  {Hayashi}, \citenamefont {Fang},\ and\ \citenamefont {Wang}}]{HFW2021}%
  \BibitemOpen
  \bibfield  {author} {\bibinfo {author} {\bibfnamefont {M.}~\bibnamefont
  {Hayashi}}, \bibinfo {author} {\bibfnamefont {K.}~\bibnamefont {Fang}},\ and\
  \bibinfo {author} {\bibfnamefont {K.}~\bibnamefont {Wang}},\ }\href
  {https://doi.org/10.1109/TIT.2021.3064009} {\bibfield  {journal} {\bibinfo
  {journal} {IEEE Transactions on Information Theory}\ }\textbf {\bibinfo
  {volume} {67}},\ \bibinfo {pages} {3926} (\bibinfo {year}
  {2021})}\BibitemShut {NoStop}%
\bibitem [{\citenamefont {Horodecki}\ \emph {et~al.}(2003)\citenamefont
  {Horodecki}, \citenamefont {Horodecki},\ and\ \citenamefont
  {Oppenheim}}]{HHO2003}%
  \BibitemOpen
  \bibfield  {author} {\bibinfo {author} {\bibfnamefont {M.}~\bibnamefont
  {Horodecki}}, \bibinfo {author} {\bibfnamefont {P.}~\bibnamefont
  {Horodecki}},\ and\ \bibinfo {author} {\bibfnamefont {J.}~\bibnamefont
  {Oppenheim}},\ }\href {https://doi.org/10.1103/PhysRevA.67.062104} {\bibfield
   {journal} {\bibinfo  {journal} {Phys. Rev. A}\ }\textbf {\bibinfo {volume}
  {67}},\ \bibinfo {pages} {062104} (\bibinfo {year} {2003})}\BibitemShut
  {NoStop}%
\bibitem [{\citenamefont {Gour}\ \emph {et~al.}(2015)\citenamefont {Gour},
  \citenamefont {Müller}, \citenamefont {Narasimhachar}, \citenamefont
  {Spekkens},\ and\ \citenamefont {{Yunger Halpern}}}]{GMN+2015}%
  \BibitemOpen
  \bibfield  {author} {\bibinfo {author} {\bibfnamefont {G.}~\bibnamefont
  {Gour}}, \bibinfo {author} {\bibfnamefont {M.~P.}\ \bibnamefont {Müller}},
  \bibinfo {author} {\bibfnamefont {V.}~\bibnamefont {Narasimhachar}}, \bibinfo
  {author} {\bibfnamefont {R.~W.}\ \bibnamefont {Spekkens}},\ and\ \bibinfo
  {author} {\bibfnamefont {N.}~\bibnamefont {{Yunger Halpern}}},\ }\href
  {https://doi.org/https://doi.org/10.1016/j.physrep.2015.04.003} {\bibfield
  {journal} {\bibinfo  {journal} {Physics Reports}\ }\textbf {\bibinfo {volume}
  {583}},\ \bibinfo {pages} {1} (\bibinfo {year} {2015})},\ \bibinfo {note}
  {the resource theory of informational nonequilibrium in
  thermodynamics}\BibitemShut {NoStop}%
\bibitem [{\citenamefont {Horodecki}\ and\ \citenamefont
  {Oppenheim}(2013)}]{HO2013}%
  \BibitemOpen
  \bibfield  {author} {\bibinfo {author} {\bibfnamefont {M.}~\bibnamefont
  {Horodecki}}\ and\ \bibinfo {author} {\bibfnamefont {J.}~\bibnamefont
  {Oppenheim}},\ }\href {https://doi.org/10.1038/ncomms3059} {\bibfield
  {journal} {\bibinfo  {journal} {Nature Communications}\ }\textbf {\bibinfo
  {volume} {4}},\ \bibinfo {pages} {2059} (\bibinfo {year} {2013})}\BibitemShut
  {NoStop}%
\bibitem [{\citenamefont {Brand\~ao}\ \emph {et~al.}(2013)\citenamefont
  {Brand\~ao}, \citenamefont {Horodecki}, \citenamefont {Oppenheim},
  \citenamefont {Renes},\ and\ \citenamefont {Spekkens}}]{BHOR2013}%
  \BibitemOpen
  \bibfield  {author} {\bibinfo {author} {\bibfnamefont {F.~G. S.~L.}\
  \bibnamefont {Brand\~ao}}, \bibinfo {author} {\bibfnamefont {M.}~\bibnamefont
  {Horodecki}}, \bibinfo {author} {\bibfnamefont {J.}~\bibnamefont
  {Oppenheim}}, \bibinfo {author} {\bibfnamefont {J.~M.}\ \bibnamefont
  {Renes}},\ and\ \bibinfo {author} {\bibfnamefont {R.~W.}\ \bibnamefont
  {Spekkens}},\ }\href {https://doi.org/10.1103/PhysRevLett.111.250404}
  {\bibfield  {journal} {\bibinfo  {journal} {Phys. Rev. Lett.}\ }\textbf
  {\bibinfo {volume} {111}},\ \bibinfo {pages} {250404} (\bibinfo {year}
  {2013})}\BibitemShut {NoStop}%
\bibitem [{\citenamefont {Faist}\ \emph {et~al.}(2015)\citenamefont {Faist},
  \citenamefont {Oppenheim},\ and\ \citenamefont {Renner}}]{FOR2015}%
  \BibitemOpen
  \bibfield  {author} {\bibinfo {author} {\bibfnamefont {P.}~\bibnamefont
  {Faist}}, \bibinfo {author} {\bibfnamefont {J.}~\bibnamefont {Oppenheim}},\
  and\ \bibinfo {author} {\bibfnamefont {R.}~\bibnamefont {Renner}},\ }\href
  {https://doi.org/10.1088/1367-2630/17/4/043003} {\bibfield  {journal}
  {\bibinfo  {journal} {New Journal of Physics}\ }\textbf {\bibinfo {volume}
  {17}},\ \bibinfo {pages} {043003} (\bibinfo {year} {2015})}\BibitemShut
  {NoStop}%
\bibitem [{\citenamefont {Wang}\ and\ \citenamefont {Wilde}(2019)}]{WW2019}%
  \BibitemOpen
  \bibfield  {author} {\bibinfo {author} {\bibfnamefont {X.}~\bibnamefont
  {Wang}}\ and\ \bibinfo {author} {\bibfnamefont {M.~M.}\ \bibnamefont
  {Wilde}},\ }\href {https://doi.org/10.1103/PhysRevResearch.1.033170}
  {\bibfield  {journal} {\bibinfo  {journal} {Phys. Rev. Res.}\ }\textbf
  {\bibinfo {volume} {1}},\ \bibinfo {pages} {033170} (\bibinfo {year}
  {2019})}\BibitemShut {NoStop}%
\bibitem [{\citenamefont {Gour}(2022)}]{Gour2022}%
  \BibitemOpen
  \bibfield  {author} {\bibinfo {author} {\bibfnamefont {G.}~\bibnamefont
  {Gour}},\ }\href {https://doi.org/10.1103/PRXQuantum.3.040323} {\bibfield
  {journal} {\bibinfo  {journal} {PRX Quantum}\ }\textbf {\bibinfo {volume}
  {3}},\ \bibinfo {pages} {040323} (\bibinfo {year} {2022})}\BibitemShut
  {NoStop}%
\bibitem [{\citenamefont {Liu}\ \emph {et~al.}(2017)\citenamefont {Liu},
  \citenamefont {Hu},\ and\ \citenamefont {Lloyd}}]{LHL2017}%
  \BibitemOpen
  \bibfield  {author} {\bibinfo {author} {\bibfnamefont {Z.-W.}\ \bibnamefont
  {Liu}}, \bibinfo {author} {\bibfnamefont {X.}~\bibnamefont {Hu}},\ and\
  \bibinfo {author} {\bibfnamefont {S.}~\bibnamefont {Lloyd}},\ }\href
  {https://doi.org/10.1103/PhysRevLett.118.060502} {\bibfield  {journal}
  {\bibinfo  {journal} {Phys. Rev. Lett.}\ }\textbf {\bibinfo {volume} {118}},\
  \bibinfo {pages} {060502} (\bibinfo {year} {2017})}\BibitemShut {NoStop}%
\bibitem [{\citenamefont {Gour}(2017)}]{Gour2017}%
  \BibitemOpen
  \bibfield  {author} {\bibinfo {author} {\bibfnamefont {G.}~\bibnamefont
  {Gour}},\ }\href {https://doi.org/10.1103/PhysRevA.95.062314} {\bibfield
  {journal} {\bibinfo  {journal} {Phys. Rev. A}\ }\textbf {\bibinfo {volume}
  {95}},\ \bibinfo {pages} {062314} (\bibinfo {year} {2017})}\BibitemShut
  {NoStop}%
\bibitem [{\citenamefont {Liu}\ \emph {et~al.}(2019)\citenamefont {Liu},
  \citenamefont {Bu},\ and\ \citenamefont {Takagi}}]{LBT2019}%
  \BibitemOpen
  \bibfield  {author} {\bibinfo {author} {\bibfnamefont {Z.-W.}\ \bibnamefont
  {Liu}}, \bibinfo {author} {\bibfnamefont {K.}~\bibnamefont {Bu}},\ and\
  \bibinfo {author} {\bibfnamefont {R.}~\bibnamefont {Takagi}},\ }\href
  {https://doi.org/10.1103/PhysRevLett.123.020401} {\bibfield  {journal}
  {\bibinfo  {journal} {Phys. Rev. Lett.}\ }\textbf {\bibinfo {volume} {123}},\
  \bibinfo {pages} {020401} (\bibinfo {year} {2019})}\BibitemShut {NoStop}%
\bibitem [{\citenamefont {Chitambar}\ and\ \citenamefont
  {Gour}(2019)}]{CG2019}%
  \BibitemOpen
  \bibfield  {author} {\bibinfo {author} {\bibfnamefont {E.}~\bibnamefont
  {Chitambar}}\ and\ \bibinfo {author} {\bibfnamefont {G.}~\bibnamefont
  {Gour}},\ }\href {https://doi.org/10.1103/RevModPhys.91.025001} {\bibfield
  {journal} {\bibinfo  {journal} {Rev. Mod. Phys.}\ }\textbf {\bibinfo {volume}
  {91}},\ \bibinfo {pages} {025001} (\bibinfo {year} {2019})}\BibitemShut
  {NoStop}%
\bibitem [{\citenamefont {Regula}\ \emph
  {et~al.}(2020{\natexlab{b}})\citenamefont {Regula}, \citenamefont {Bu},
  \citenamefont {Takagi},\ and\ \citenamefont {Liu}}]{RBTL2020}%
  \BibitemOpen
  \bibfield  {author} {\bibinfo {author} {\bibfnamefont {B.}~\bibnamefont
  {Regula}}, \bibinfo {author} {\bibfnamefont {K.}~\bibnamefont {Bu}}, \bibinfo
  {author} {\bibfnamefont {R.}~\bibnamefont {Takagi}},\ and\ \bibinfo {author}
  {\bibfnamefont {Z.-W.}\ \bibnamefont {Liu}},\ }\href
  {https://doi.org/10.1103/PhysRevA.101.062315} {\bibfield  {journal} {\bibinfo
   {journal} {Phys. Rev. A}\ }\textbf {\bibinfo {volume} {101}},\ \bibinfo
  {pages} {062315} (\bibinfo {year} {2020}{\natexlab{b}})}\BibitemShut
  {NoStop}%
\bibitem [{\citenamefont {Gour}(2025{\natexlab{a}})}]{Gour2025}%
  \BibitemOpen
  \bibfield  {author} {\bibinfo {author} {\bibfnamefont {G.}~\bibnamefont
  {Gour}},\ }\href {https://doi.org/10.1017/9781009560870} {\emph {\bibinfo
  {title} {Quantum Resource Theories}}}\ (\bibinfo  {publisher} {Cambridge
  University Press},\ \bibinfo {address} {Cambridge},\ \bibinfo {year}
  {2025})\BibitemShut {NoStop}%
\bibitem [{\citenamefont {Narasimhachar}\ and\ \citenamefont
  {Gour}(2017)}]{NG2017}%
  \BibitemOpen
  \bibfield  {author} {\bibinfo {author} {\bibfnamefont {V.}~\bibnamefont
  {Narasimhachar}}\ and\ \bibinfo {author} {\bibfnamefont {G.}~\bibnamefont
  {Gour}},\ }\href {https://doi.org/10.1103/PhysRevA.95.012313} {\bibfield
  {journal} {\bibinfo  {journal} {Phys. Rev. A}\ }\textbf {\bibinfo {volume}
  {95}},\ \bibinfo {pages} {012313} (\bibinfo {year} {2017})}\BibitemShut
  {NoStop}%
\bibitem [{\citenamefont {Bera}\ \emph {et~al.}(2017)\citenamefont {Bera},
  \citenamefont {Riera}, \citenamefont {Lewenstein},\ and\ \citenamefont
  {Winter}}]{BRLW2017}%
  \BibitemOpen
  \bibfield  {author} {\bibinfo {author} {\bibfnamefont {M.~N.}\ \bibnamefont
  {Bera}}, \bibinfo {author} {\bibfnamefont {A.}~\bibnamefont {Riera}},
  \bibinfo {author} {\bibfnamefont {M.}~\bibnamefont {Lewenstein}},\ and\
  \bibinfo {author} {\bibfnamefont {A.}~\bibnamefont {Winter}},\ }\href
  {https://doi.org/10.1038/s41467-017-02370-x} {\bibfield  {journal} {\bibinfo
  {journal} {Nature Communications}\ }\textbf {\bibinfo {volume} {8}},\
  \bibinfo {pages} {2180} (\bibinfo {year} {2017})}\BibitemShut {NoStop}%
\bibitem [{\citenamefont {Morris}\ \emph {et~al.}(2019)\citenamefont {Morris},
  \citenamefont {Lami},\ and\ \citenamefont {Adesso}}]{MLA2019}%
  \BibitemOpen
  \bibfield  {author} {\bibinfo {author} {\bibfnamefont {B.}~\bibnamefont
  {Morris}}, \bibinfo {author} {\bibfnamefont {L.}~\bibnamefont {Lami}},\ and\
  \bibinfo {author} {\bibfnamefont {G.}~\bibnamefont {Adesso}},\ }\href
  {https://doi.org/10.1103/PhysRevLett.122.130601} {\bibfield  {journal}
  {\bibinfo  {journal} {Phys. Rev. Lett.}\ }\textbf {\bibinfo {volume} {122}},\
  \bibinfo {pages} {130601} (\bibinfo {year} {2019})}\BibitemShut {NoStop}%
\bibitem [{\citenamefont {Ji}\ \emph {et~al.}(2025)\citenamefont {Ji},
  \citenamefont {Gour},\ and\ \citenamefont {Wilde}}]{JGW2025}%
  \BibitemOpen
  \bibfield  {author} {\bibinfo {author} {\bibfnamefont {K.}~\bibnamefont
  {Ji}}, \bibinfo {author} {\bibfnamefont {G.}~\bibnamefont {Gour}},\ and\
  \bibinfo {author} {\bibfnamefont {M.~M.}\ \bibnamefont {Wilde}},\ }\href@noop
  {} {\bibinfo {title} {Fundamental work costs of preparation and erasure in
  the presence of quantum side information}} (\bibinfo {year} {2025}),\ \Eprint
  {https://arxiv.org/abs/2503.09012} {arXiv:2503.09012 [quant-ph]} \BibitemShut
  {NoStop}%
\bibitem [{\citenamefont {Gour}\ \emph {et~al.}(2018)\citenamefont {Gour},
  \citenamefont {Grudka}, \citenamefont {Horodecki}, \citenamefont {K\l{}obus},
  \citenamefont {\L{}odyga},\ and\ \citenamefont {Narasimhachar}}]{GGH+2018}%
  \BibitemOpen
  \bibfield  {author} {\bibinfo {author} {\bibfnamefont {G.}~\bibnamefont
  {Gour}}, \bibinfo {author} {\bibfnamefont {A.}~\bibnamefont {Grudka}},
  \bibinfo {author} {\bibfnamefont {M.}~\bibnamefont {Horodecki}}, \bibinfo
  {author} {\bibfnamefont {W.}~\bibnamefont {K\l{}obus}}, \bibinfo {author}
  {\bibfnamefont {J.}~\bibnamefont {\L{}odyga}},\ and\ \bibinfo {author}
  {\bibfnamefont {V.}~\bibnamefont {Narasimhachar}},\ }\href
  {https://doi.org/10.1103/PhysRevA.97.042130} {\bibfield  {journal} {\bibinfo
  {journal} {Phys. Rev. A}\ }\textbf {\bibinfo {volume} {97}},\ \bibinfo
  {pages} {042130} (\bibinfo {year} {2018})}\BibitemShut {NoStop}%
\bibitem [{\citenamefont {Brandsen}\ \emph {et~al.}(2022)\citenamefont
  {Brandsen}, \citenamefont {Geng},\ and\ \citenamefont {Gour}}]{BGG2022}%
  \BibitemOpen
  \bibfield  {author} {\bibinfo {author} {\bibfnamefont {S.}~\bibnamefont
  {Brandsen}}, \bibinfo {author} {\bibfnamefont {I.~J.}\ \bibnamefont {Geng}},\
  and\ \bibinfo {author} {\bibfnamefont {G.}~\bibnamefont {Gour}},\ }\href
  {https://doi.org/10.1103/PhysRevE.105.024117} {\bibfield  {journal} {\bibinfo
   {journal} {Phys. Rev. E}\ }\textbf {\bibinfo {volume} {105}},\ \bibinfo
  {pages} {024117} (\bibinfo {year} {2022})}\BibitemShut {NoStop}%
\bibitem [{\citenamefont {Gour}\ \emph {et~al.}(2024)\citenamefont {Gour},
  \citenamefont {Wilde}, \citenamefont {Brandsen},\ and\ \citenamefont
  {Geng}}]{GWBG2024}%
  \BibitemOpen
  \bibfield  {author} {\bibinfo {author} {\bibfnamefont {G.}~\bibnamefont
  {Gour}}, \bibinfo {author} {\bibfnamefont {M.~M.}\ \bibnamefont {Wilde}},
  \bibinfo {author} {\bibfnamefont {S.}~\bibnamefont {Brandsen}},\ and\
  \bibinfo {author} {\bibfnamefont {I.~J.}\ \bibnamefont {Geng}},\ }\href
  {https://doi.org/10.22331/q-2024-11-20-1529} {\bibfield  {journal} {\bibinfo
  {journal} {{Quantum}}\ }\textbf {\bibinfo {volume} {8}},\ \bibinfo {pages}
  {1529} (\bibinfo {year} {2024})}\BibitemShut {NoStop}%
\bibitem [{\citenamefont {Hardy}\ \emph {et~al.}(1929)\citenamefont {Hardy},
  \citenamefont {Littlewood},\ and\ \citenamefont {P{\'o}lya}}]{HLP1929}%
  \BibitemOpen
  \bibfield  {author} {\bibinfo {author} {\bibfnamefont {G.~H.}\ \bibnamefont
  {Hardy}}, \bibinfo {author} {\bibfnamefont {J.~E.}\ \bibnamefont
  {Littlewood}},\ and\ \bibinfo {author} {\bibfnamefont {G.}~\bibnamefont
  {P{\'o}lya}},\ }\href@noop {} {\bibfield  {journal} {\bibinfo  {journal}
  {Messenger of Mathematics}\ }\textbf {\bibinfo {volume} {58}},\ \bibinfo
  {pages} {145} (\bibinfo {year} {1929})}\BibitemShut {NoStop}%
\bibitem [{\citenamefont {Lieb}\ and\ \citenamefont {Yngvason}(1999)}]{LY1999}%
  \BibitemOpen
  \bibfield  {author} {\bibinfo {author} {\bibfnamefont {E.~H.}\ \bibnamefont
  {Lieb}}\ and\ \bibinfo {author} {\bibfnamefont {J.}~\bibnamefont
  {Yngvason}},\ }\href
  {https://doi.org/https://doi.org/10.1016/S0370-1573(98)00082-9} {\bibfield
  {journal} {\bibinfo  {journal} {Physics Reports}\ }\textbf {\bibinfo {volume}
  {310}},\ \bibinfo {pages} {1} (\bibinfo {year} {1999})}\BibitemShut {NoStop}%
\bibitem [{\citenamefont {Partovi}(2011)}]{Partovi2011}%
  \BibitemOpen
  \bibfield  {author} {\bibinfo {author} {\bibfnamefont {M.~H.}\ \bibnamefont
  {Partovi}},\ }\href {https://doi.org/10.1103/PhysRevA.84.052117} {\bibfield
  {journal} {\bibinfo  {journal} {Phys. Rev. A}\ }\textbf {\bibinfo {volume}
  {84}},\ \bibinfo {pages} {052117} (\bibinfo {year} {2011})}\BibitemShut
  {NoStop}%
\bibitem [{\citenamefont {Friedland}\ \emph {et~al.}(2013)\citenamefont
  {Friedland}, \citenamefont {Gheorghiu},\ and\ \citenamefont
  {Gour}}]{FGG2013}%
  \BibitemOpen
  \bibfield  {author} {\bibinfo {author} {\bibfnamefont {S.}~\bibnamefont
  {Friedland}}, \bibinfo {author} {\bibfnamefont {V.}~\bibnamefont
  {Gheorghiu}},\ and\ \bibinfo {author} {\bibfnamefont {G.}~\bibnamefont
  {Gour}},\ }\href {https://doi.org/10.1103/PhysRevLett.111.230401} {\bibfield
  {journal} {\bibinfo  {journal} {Phys. Rev. Lett.}\ }\textbf {\bibinfo
  {volume} {111}},\ \bibinfo {pages} {230401} (\bibinfo {year}
  {2013})}\BibitemShut {NoStop}%
\bibitem [{\citenamefont {Blackwell}(1953)}]{Blackwell1953}%
  \BibitemOpen
  \bibfield  {author} {\bibinfo {author} {\bibfnamefont {D.}~\bibnamefont
  {Blackwell}},\ }\href {https://doi.org/10.1214/aoms/1177729032} {\bibfield
  {journal} {\bibinfo  {journal} {The Annals of Mathematical Statistics}\
  }\textbf {\bibinfo {volume} {24}},\ \bibinfo {pages} {265} (\bibinfo {year}
  {1953})}\BibitemShut {NoStop}%
\bibitem [{\citenamefont {Buscemi}\ \emph {et~al.}(2019)\citenamefont
  {Buscemi}, \citenamefont {Sutter},\ and\ \citenamefont
  {Tomamichel}}]{BST2019}%
  \BibitemOpen
  \bibfield  {author} {\bibinfo {author} {\bibfnamefont {F.}~\bibnamefont
  {Buscemi}}, \bibinfo {author} {\bibfnamefont {D.}~\bibnamefont {Sutter}},\
  and\ \bibinfo {author} {\bibfnamefont {M.}~\bibnamefont {Tomamichel}},\
  }\href {https://doi.org/10.22331/q-2019-12-09-209} {\bibfield  {journal}
  {\bibinfo  {journal} {{Quantum}}\ }\textbf {\bibinfo {volume} {3}},\ \bibinfo
  {pages} {209} (\bibinfo {year} {2019})}\BibitemShut {NoStop}%
\bibitem [{\citenamefont {Coecke}\ \emph {et~al.}(2016)\citenamefont {Coecke},
  \citenamefont {Fritz},\ and\ \citenamefont {Spekkens}}]{CFS2016}%
  \BibitemOpen
  \bibfield  {author} {\bibinfo {author} {\bibfnamefont {B.}~\bibnamefont
  {Coecke}}, \bibinfo {author} {\bibfnamefont {T.}~\bibnamefont {Fritz}},\ and\
  \bibinfo {author} {\bibfnamefont {R.~W.}\ \bibnamefont {Spekkens}},\ }\href
  {https://doi.org/10.1016/j.ic.2016.02.008} {\bibfield  {journal} {\bibinfo
  {journal} {Information and Computation}\ }\textbf {\bibinfo {volume} {250}},\
  \bibinfo {pages} {59} (\bibinfo {year} {2016})},\ \bibinfo {note} {quantum
  Physics and Logic}\BibitemShut {NoStop}%
\bibitem [{\citenamefont {Knill}(2006)}]{kni2006}%
  \BibitemOpen
  \bibfield  {author} {\bibinfo {author} {\bibfnamefont {E.}~\bibnamefont
  {Knill}},\ }\href {https://doi.org/10.1103/PhysRevA.74.042301} {\bibfield
  {journal} {\bibinfo  {journal} {Physical Review A}\ }\textbf {\bibinfo
  {volume} {74}},\ \bibinfo {pages} {042301} (\bibinfo {year}
  {2006})}\BibitemShut {NoStop}%
\bibitem [{\citenamefont {Lami}\ and\ \citenamefont {Regula}(2024)}]{LR2024}%
  \BibitemOpen
  \bibfield  {author} {\bibinfo {author} {\bibfnamefont {L.}~\bibnamefont
  {Lami}}\ and\ \bibinfo {author} {\bibfnamefont {B.}~\bibnamefont {Regula}},\
  }\href {https://doi.org/10.1038/s41467-024-54201-5} {\bibfield  {journal}
  {\bibinfo  {journal} {Nature Communications}\ }\textbf {\bibinfo {volume}
  {15}},\ \bibinfo {pages} {10120} (\bibinfo {year} {2024})}\BibitemShut
  {NoStop}%
\bibitem [{\citenamefont {Gao}\ and\ \citenamefont {Rahaman}(2024)}]{GR2024}%
  \BibitemOpen
  \bibfield  {author} {\bibinfo {author} {\bibfnamefont {L.}~\bibnamefont
  {Gao}}\ and\ \bibinfo {author} {\bibfnamefont {M.}~\bibnamefont {Rahaman}},\
  }\href {https://arxiv.org/abs/2401.03090} {\bibinfo {title} {Generalized
  stein's lemma and asymptotic equipartition property for subalgebra
  entropies}} (\bibinfo {year} {2024})\BibitemShut {NoStop}%
\bibitem [{\citenamefont {Bai}\ \emph {et~al.}(2024)\citenamefont {Bai},
  \citenamefont {{\v S}afr{\'a}nek}, \citenamefont {Schindler}, \citenamefont
  {Buscemi},\ and\ \citenamefont {Scarani}}]{bss+2024}%
  \BibitemOpen
  \bibfield  {author} {\bibinfo {author} {\bibfnamefont {G.}~\bibnamefont
  {Bai}}, \bibinfo {author} {\bibfnamefont {D.}~\bibnamefont {{\v
  S}afr{\'a}nek}}, \bibinfo {author} {\bibfnamefont {J.}~\bibnamefont
  {Schindler}}, \bibinfo {author} {\bibfnamefont {F.}~\bibnamefont {Buscemi}},\
  and\ \bibinfo {author} {\bibfnamefont {V.}~\bibnamefont {Scarani}},\ }\href
  {https://doi.org/10.22331/q-2024-11-14-1524} {\bibfield  {journal} {\bibinfo
  {journal} {Quantum}\ }\textbf {\bibinfo {volume} {8}},\ \bibinfo {pages}
  {1524} (\bibinfo {year} {2024})},\ \Eprint {https://arxiv.org/abs/2308.08763}
  {arXiv:2308.08763 [quant-ph]} \BibitemShut {NoStop}%
\bibitem [{\citenamefont {Nagasawa}\ \emph {et~al.}(2024)\citenamefont
  {Nagasawa}, \citenamefont {Kato}, \citenamefont {Wakakuwa},\ and\
  \citenamefont {Buscemi}}]{NWKB2024}%
  \BibitemOpen
  \bibfield  {author} {\bibinfo {author} {\bibfnamefont {T.}~\bibnamefont
  {Nagasawa}}, \bibinfo {author} {\bibfnamefont {K.}~\bibnamefont {Kato}},
  \bibinfo {author} {\bibfnamefont {E.}~\bibnamefont {Wakakuwa}},\ and\
  \bibinfo {author} {\bibfnamefont {F.}~\bibnamefont {Buscemi}},\ }\href
  {https://doi.org/10.1103/PhysRevResearch.6.043327} {\bibfield  {journal}
  {\bibinfo  {journal} {Phys. Rev. Res.}\ }\textbf {\bibinfo {volume} {6}},\
  \bibinfo {pages} {043327} (\bibinfo {year} {2024})}\BibitemShut {NoStop}%
\bibitem [{\citenamefont {Nagasawa}\ \emph {et~al.}(2025)\citenamefont
  {Nagasawa}, \citenamefont {Wakakuwa}, \citenamefont {Kato},\ and\
  \citenamefont {Buscemi}}]{NWKB2025}%
  \BibitemOpen
  \bibfield  {author} {\bibinfo {author} {\bibfnamefont {T.}~\bibnamefont
  {Nagasawa}}, \bibinfo {author} {\bibfnamefont {E.}~\bibnamefont {Wakakuwa}},
  \bibinfo {author} {\bibfnamefont {K.}~\bibnamefont {Kato}},\ and\ \bibinfo
  {author} {\bibfnamefont {F.}~\bibnamefont {Buscemi}},\ }\href
  {https://doi.org/10.1088/1361-6633/ae140e} {\bibfield  {journal} {\bibinfo
  {journal} {Reports on Progress in Physics}\ }\textbf {\bibinfo {volume}
  {88}},\ \bibinfo {pages} {117601} (\bibinfo {year} {2025})}\BibitemShut
  {NoStop}%
\bibitem [{\citenamefont {Gour}\ and\ \citenamefont
  {Tomamichel}(2020)}]{GT2020}%
  \BibitemOpen
  \bibfield  {author} {\bibinfo {author} {\bibfnamefont {G.}~\bibnamefont
  {Gour}}\ and\ \bibinfo {author} {\bibfnamefont {M.}~\bibnamefont
  {Tomamichel}},\ }\href {https://doi.org/10.1103/PhysRevA.102.062401}
  {\bibfield  {journal} {\bibinfo  {journal} {Phys. Rev. A}\ }\textbf {\bibinfo
  {volume} {102}},\ \bibinfo {pages} {062401} (\bibinfo {year}
  {2020})}\BibitemShut {NoStop}%
\bibitem [{\citenamefont {Vedral}\ and\ \citenamefont {Plenio}(1998)}]{VP1998}%
  \BibitemOpen
  \bibfield  {author} {\bibinfo {author} {\bibfnamefont {V.}~\bibnamefont
  {Vedral}}\ and\ \bibinfo {author} {\bibfnamefont {M.~B.}\ \bibnamefont
  {Plenio}},\ }\href {https://doi.org/10.1103/PhysRevA.57.1619} {\bibfield
  {journal} {\bibinfo  {journal} {Phys. Rev. A}\ }\textbf {\bibinfo {volume}
  {57}},\ \bibinfo {pages} {1619} (\bibinfo {year} {1998})}\BibitemShut
  {NoStop}%
\bibitem [{\citenamefont {Audenaert}\ and\ \citenamefont
  {Datta}(2015)}]{AD2015}%
  \BibitemOpen
  \bibfield  {author} {\bibinfo {author} {\bibfnamefont {K.~M.~R.}\
  \bibnamefont {Audenaert}}\ and\ \bibinfo {author} {\bibfnamefont
  {N.}~\bibnamefont {Datta}},\ }\href {https://doi.org/10.1063/1.4906367}
  {\bibfield  {journal} {\bibinfo  {journal} {Journal of Mathematical Physics}\
  }\textbf {\bibinfo {volume} {56}},\ \bibinfo {pages} {022202} (\bibinfo
  {year} {2015})}\BibitemShut {NoStop}%
\bibitem [{\citenamefont {Zhang}(2020)}]{Zhang2020}%
  \BibitemOpen
  \bibfield  {author} {\bibinfo {author} {\bibfnamefont {H.}~\bibnamefont
  {Zhang}},\ }\href {https://doi.org/https://doi.org/10.1016/j.aim.2020.107053}
  {\bibfield  {journal} {\bibinfo  {journal} {Advances in Mathematics}\
  }\textbf {\bibinfo {volume} {365}},\ \bibinfo {pages} {107053} (\bibinfo
  {year} {2020})}\BibitemShut {NoStop}%
\bibitem [{Note1()}]{Note1}%
  \BibitemOpen
  \bibinfo {note} {Trace functionals of the form \(\protect \tr [(XS^rX^*)^z]\)
  are evaluated directly when \(S\) is invertible, and otherwise by the
  corresponding semicontinuous extension; see Appendix~\ref
  {appendix:optimization}.}\BibitemShut {Stop}%
\bibitem [{\citenamefont {Petz}(1985)}]{Petz1985}%
  \BibitemOpen
  \bibfield  {author} {\bibinfo {author} {\bibfnamefont {D.}~\bibnamefont
  {Petz}},\ }\href {https://doi.org/10.2977/PRIMS/1195178929} {\bibfield
  {journal} {\bibinfo  {journal} {Publications of the Research Institute for
  Mathematical Sciences}\ }\textbf {\bibinfo {volume} {21}},\ \bibinfo {pages}
  {787} (\bibinfo {year} {1985})}\BibitemShut {NoStop}%
\bibitem [{\citenamefont {Müller-Lennert}\ \emph {et~al.}(2013)\citenamefont
  {Müller-Lennert}, \citenamefont {Dupuis}, \citenamefont {Szehr},
  \citenamefont {Fehr},\ and\ \citenamefont {Tomamichel}}]{MDS+2013}%
  \BibitemOpen
  \bibfield  {author} {\bibinfo {author} {\bibfnamefont {M.}~\bibnamefont
  {Müller-Lennert}}, \bibinfo {author} {\bibfnamefont {F.}~\bibnamefont
  {Dupuis}}, \bibinfo {author} {\bibfnamefont {O.}~\bibnamefont {Szehr}},
  \bibinfo {author} {\bibfnamefont {S.}~\bibnamefont {Fehr}},\ and\ \bibinfo
  {author} {\bibfnamefont {M.}~\bibnamefont {Tomamichel}},\ }\href
  {https://doi.org/10.1063/1.4838856} {\bibfield  {journal} {\bibinfo
  {journal} {Journal of Mathematical Physics}\ }\textbf {\bibinfo {volume}
  {54}},\ \bibinfo {pages} {122203} (\bibinfo {year} {2013})}\BibitemShut
  {NoStop}%
\bibitem [{\citenamefont {Datta}(2009)}]{Datta2009}%
  \BibitemOpen
  \bibfield  {author} {\bibinfo {author} {\bibfnamefont {N.}~\bibnamefont
  {Datta}},\ }\href {https://doi.org/10.1109/TIT.2009.2018325} {\bibfield
  {journal} {\bibinfo  {journal} {IEEE Transactions on Information Theory}\
  }\textbf {\bibinfo {volume} {55}},\ \bibinfo {pages} {2816} (\bibinfo {year}
  {2009})}\BibitemShut {NoStop}%
\bibitem [{\citenamefont {Rubboli}\ and\ \citenamefont
  {Tomamichel}(2024)}]{RubboliTomamichel2024}%
  \BibitemOpen
  \bibfield  {author} {\bibinfo {author} {\bibfnamefont {R.}~\bibnamefont
  {Rubboli}}\ and\ \bibinfo {author} {\bibfnamefont {M.}~\bibnamefont
  {Tomamichel}},\ }\href {https://doi.org/10.1007/s00220-024-05025-3}
  {\bibfield  {journal} {\bibinfo  {journal} {Communications in Mathematical
  Physics}\ }\textbf {\bibinfo {volume} {405}},\ \bibinfo {pages} {162}
  (\bibinfo {year} {2024})}\BibitemShut {NoStop}%
\bibitem [{\citenamefont {Junge}\ \emph {et~al.}(2023)\citenamefont {Junge},
  \citenamefont {Laracuente},\ and\ \citenamefont {Rouz{\'e}}}]{JLR2023}%
  \BibitemOpen
  \bibfield  {author} {\bibinfo {author} {\bibfnamefont {M.}~\bibnamefont
  {Junge}}, \bibinfo {author} {\bibfnamefont {N.}~\bibnamefont {Laracuente}},\
  and\ \bibinfo {author} {\bibfnamefont {C.}~\bibnamefont {Rouz{\'e}}},\ }\href
  {https://doi.org/10.1007/s10955-022-03026-x} {\bibfield  {journal} {\bibinfo
  {journal} {Journal of Statistical Physics}\ }\textbf {\bibinfo {volume}
  {190}},\ \bibinfo {pages} {30} (\bibinfo {year} {2023})}\BibitemShut
  {NoStop}%
\bibitem [{\citenamefont {Wang}\ and\ \citenamefont {Renner}(2012)}]{WR2012}%
  \BibitemOpen
  \bibfield  {author} {\bibinfo {author} {\bibfnamefont {L.}~\bibnamefont
  {Wang}}\ and\ \bibinfo {author} {\bibfnamefont {R.}~\bibnamefont {Renner}},\
  }\href {https://doi.org/10.1103/PhysRevLett.108.200501} {\bibfield  {journal}
  {\bibinfo  {journal} {Phys. Rev. Lett.}\ }\textbf {\bibinfo {volume} {108}},\
  \bibinfo {pages} {200501} (\bibinfo {year} {2012})}\BibitemShut {NoStop}%
\bibitem [{\citenamefont {Brandão}\ \emph {et~al.}(2020)\citenamefont
  {Brandão}, \citenamefont {Harrow}, \citenamefont {Lee},\ and\ \citenamefont
  {Peres}}]{BHA+2020}%
  \BibitemOpen
  \bibfield  {author} {\bibinfo {author} {\bibfnamefont {F.~G. S.~L.}\
  \bibnamefont {Brandão}}, \bibinfo {author} {\bibfnamefont {A.~W.}\
  \bibnamefont {Harrow}}, \bibinfo {author} {\bibfnamefont {J.~R.}\
  \bibnamefont {Lee}},\ and\ \bibinfo {author} {\bibfnamefont {Y.}~\bibnamefont
  {Peres}},\ }\href {https://doi.org/10.1109/TIT.2020.2979704} {\bibfield
  {journal} {\bibinfo  {journal} {IEEE Transactions on Information Theory}\
  }\textbf {\bibinfo {volume} {66}},\ \bibinfo {pages} {5037} (\bibinfo {year}
  {2020})}\BibitemShut {NoStop}%
\bibitem [{\citenamefont {Berta}\ \emph {et~al.}(2023)\citenamefont {Berta},
  \citenamefont {Brand{\~{a}}o}, \citenamefont {Gour}, \citenamefont {Lami},
  \citenamefont {Plenio}, \citenamefont {Regula},\ and\ \citenamefont
  {Tomamichel}}]{BBG+2023}%
  \BibitemOpen
  \bibfield  {author} {\bibinfo {author} {\bibfnamefont {M.}~\bibnamefont
  {Berta}}, \bibinfo {author} {\bibfnamefont {F.~G. S.~L.}\ \bibnamefont
  {Brand{\~{a}}o}}, \bibinfo {author} {\bibfnamefont {G.}~\bibnamefont {Gour}},
  \bibinfo {author} {\bibfnamefont {L.}~\bibnamefont {Lami}}, \bibinfo {author}
  {\bibfnamefont {M.~B.}\ \bibnamefont {Plenio}}, \bibinfo {author}
  {\bibfnamefont {B.}~\bibnamefont {Regula}},\ and\ \bibinfo {author}
  {\bibfnamefont {M.}~\bibnamefont {Tomamichel}},\ }\href
  {https://doi.org/10.22331/q-2023-09-07-1103} {\bibfield  {journal} {\bibinfo
  {journal} {{Quantum}}\ }\textbf {\bibinfo {volume} {7}},\ \bibinfo {pages}
  {1103} (\bibinfo {year} {2023})}\BibitemShut {NoStop}%
\bibitem [{\citenamefont {Zanoni}\ and\ \citenamefont
  {Scandolo}(2025)}]{ZS2025}%
  \BibitemOpen
  \bibfield  {author} {\bibinfo {author} {\bibfnamefont {E.}~\bibnamefont
  {Zanoni}}\ and\ \bibinfo {author} {\bibfnamefont {C.~M.}\ \bibnamefont
  {Scandolo}},\ }\href {https://doi.org/10.1103/PhysRevA.111.062407} {\bibfield
   {journal} {\bibinfo  {journal} {Phys. Rev. A}\ }\textbf {\bibinfo {volume}
  {111}},\ \bibinfo {pages} {062407} (\bibinfo {year} {2025})}\BibitemShut
  {NoStop}%
\bibitem [{\citenamefont {Gour}\ and\ \citenamefont {Spekkens}(2008)}]{GS2008}%
  \BibitemOpen
  \bibfield  {author} {\bibinfo {author} {\bibfnamefont {G.}~\bibnamefont
  {Gour}}\ and\ \bibinfo {author} {\bibfnamefont {R.~W.}\ \bibnamefont
  {Spekkens}},\ }\href {https://doi.org/10.1088/1367-2630/10/3/033023}
  {\bibfield  {journal} {\bibinfo  {journal} {New Journal of Physics}\ }\textbf
  {\bibinfo {volume} {10}},\ \bibinfo {pages} {033023} (\bibinfo {year}
  {2008})}\BibitemShut {NoStop}%
\bibitem [{\citenamefont {Marvian}\ and\ \citenamefont
  {Spekkens}(2013)}]{MS2013}%
  \BibitemOpen
  \bibfield  {author} {\bibinfo {author} {\bibfnamefont {I.}~\bibnamefont
  {Marvian}}\ and\ \bibinfo {author} {\bibfnamefont {R.~W.}\ \bibnamefont
  {Spekkens}},\ }\href {https://doi.org/10.1088/1367-2630/15/3/033001}
  {\bibfield  {journal} {\bibinfo  {journal} {New Journal of Physics}\ }\textbf
  {\bibinfo {volume} {15}},\ \bibinfo {pages} {033001} (\bibinfo {year}
  {2013})}\BibitemShut {NoStop}%
\bibitem [{\citenamefont {Marvian}\ and\ \citenamefont
  {Spekkens}(2014)}]{MS2014}%
  \BibitemOpen
  \bibfield  {author} {\bibinfo {author} {\bibfnamefont {I.}~\bibnamefont
  {Marvian}}\ and\ \bibinfo {author} {\bibfnamefont {R.~W.}\ \bibnamefont
  {Spekkens}},\ }\href {https://doi.org/10.1103/PhysRevA.90.062110} {\bibfield
  {journal} {\bibinfo  {journal} {Phys. Rev. A}\ }\textbf {\bibinfo {volume}
  {90}},\ \bibinfo {pages} {062110} (\bibinfo {year} {2014})}\BibitemShut
  {NoStop}%
\bibitem [{\citenamefont {Horodecki}\ \emph {et~al.}(1998)\citenamefont
  {Horodecki}, \citenamefont {Horodecki},\ and\ \citenamefont
  {Horodecki}}]{HH1998}%
  \BibitemOpen
  \bibfield  {author} {\bibinfo {author} {\bibfnamefont {M.}~\bibnamefont
  {Horodecki}}, \bibinfo {author} {\bibfnamefont {P.}~\bibnamefont
  {Horodecki}},\ and\ \bibinfo {author} {\bibfnamefont {R.}~\bibnamefont
  {Horodecki}},\ }\href {https://doi.org/10.1103/PhysRevLett.80.5239}
  {\bibfield  {journal} {\bibinfo  {journal} {Phys. Rev. Lett.}\ }\textbf
  {\bibinfo {volume} {80}},\ \bibinfo {pages} {5239} (\bibinfo {year}
  {1998})}\BibitemShut {NoStop}%
\bibitem [{\citenamefont {Hayden}\ \emph {et~al.}(2001)\citenamefont {Hayden},
  \citenamefont {Horodecki},\ and\ \citenamefont {Terhal}}]{hht2001}%
  \BibitemOpen
  \bibfield  {author} {\bibinfo {author} {\bibfnamefont {P.~M.}\ \bibnamefont
  {Hayden}}, \bibinfo {author} {\bibfnamefont {M.}~\bibnamefont {Horodecki}},\
  and\ \bibinfo {author} {\bibfnamefont {B.~M.}\ \bibnamefont {Terhal}},\
  }\href {https://doi.org/10.1088/0305-4470/34/35/314} {\bibfield  {journal}
  {\bibinfo  {journal} {Journal of Physics A: Mathematical and General}\
  }\textbf {\bibinfo {volume} {34}},\ \bibinfo {pages} {6891} (\bibinfo {year}
  {2001})}\BibitemShut {NoStop}%
\bibitem [{\citenamefont {Christandl}\ and\ \citenamefont
  {Winter}(2004)}]{cw2004}%
  \BibitemOpen
  \bibfield  {author} {\bibinfo {author} {\bibfnamefont {M.}~\bibnamefont
  {Christandl}}\ and\ \bibinfo {author} {\bibfnamefont {A.}~\bibnamefont
  {Winter}},\ }\href {https://doi.org/10.1063/1.1643788} {\bibfield  {journal}
  {\bibinfo  {journal} {Journal of Mathematical Physics}\ }\textbf {\bibinfo
  {volume} {45}},\ \bibinfo {pages} {829} (\bibinfo {year} {2004})}\BibitemShut
  {NoStop}%
\bibitem [{\citenamefont {Yang}\ \emph {et~al.}(2008)\citenamefont {Yang},
  \citenamefont {Horodecki},\ and\ \citenamefont {Wang}}]{yhw2008}%
  \BibitemOpen
  \bibfield  {author} {\bibinfo {author} {\bibfnamefont {D.}~\bibnamefont
  {Yang}}, \bibinfo {author} {\bibfnamefont {M.}~\bibnamefont {Horodecki}},\
  and\ \bibinfo {author} {\bibfnamefont {Z.~D.}\ \bibnamefont {Wang}},\ }\href
  {https://doi.org/10.1103/PhysRevLett.101.140501} {\bibfield  {journal}
  {\bibinfo  {journal} {Physical Review Letters}\ }\textbf {\bibinfo {volume}
  {101}},\ \bibinfo {pages} {140501} (\bibinfo {year} {2008})}\BibitemShut
  {NoStop}%
\bibitem [{\citenamefont {Beigi}\ \emph {et~al.}(2026)\citenamefont {Beigi},
  \citenamefont {Rubboli},\ and\ \citenamefont {Tomamichel}}]{brt2026}%
  \BibitemOpen
  \bibfield  {author} {\bibinfo {author} {\bibfnamefont {S.}~\bibnamefont
  {Beigi}}, \bibinfo {author} {\bibfnamefont {R.}~\bibnamefont {Rubboli}},\
  and\ \bibinfo {author} {\bibfnamefont {M.}~\bibnamefont {Tomamichel}},\
  }\href {https://doi.org/10.1007/s00220-026-05653-x} {\bibfield  {journal}
  {\bibinfo  {journal} {Communications in Mathematical Physics}\ }\textbf
  {\bibinfo {volume} {407}},\ \bibinfo {pages} {134} (\bibinfo {year}
  {2026})}\BibitemShut {NoStop}%
\bibitem [{\citenamefont {Rubboli}\ \emph {et~al.}(2025)\citenamefont
  {Rubboli}, \citenamefont {Goodarzi},\ and\ \citenamefont
  {Tomamichel}}]{RGT2025}%
  \BibitemOpen
  \bibfield  {author} {\bibinfo {author} {\bibfnamefont {R.}~\bibnamefont
  {Rubboli}}, \bibinfo {author} {\bibfnamefont {M.~M.}\ \bibnamefont
  {Goodarzi}},\ and\ \bibinfo {author} {\bibfnamefont {M.}~\bibnamefont
  {Tomamichel}},\ }\href {https://arxiv.org/abs/2410.21976} {\bibinfo {title}
  {Quantum conditional entropies from convex trace functionals}} (\bibinfo
  {year} {2025})\BibitemShut {NoStop}%
\bibitem [{\citenamefont {Rubboli}(2025)}]{rubboli-phd-2025}%
  \BibitemOpen
  \bibfield  {author} {\bibinfo {author} {\bibfnamefont {R.}~\bibnamefont
  {Rubboli}},\ }\href {https://scholarbank.nus.edu.sg/handle/10635/311115}
  {\bibinfo {title} {Optimization and additivity of quantum relative
  entropies}} (\bibinfo {year} {2025}),\ \bibinfo {note} {\ PhD
  Thesis}\BibitemShut {NoStop}%
\bibitem [{\citenamefont {Mosonyi}\ \emph {et~al.}(2024)\citenamefont
  {Mosonyi}, \citenamefont {Bunth},\ and\ \citenamefont {Vrana}}]{mbv2024}%
  \BibitemOpen
  \bibfield  {author} {\bibinfo {author} {\bibfnamefont {M.}~\bibnamefont
  {Mosonyi}}, \bibinfo {author} {\bibfnamefont {G.}~\bibnamefont {Bunth}},\
  and\ \bibinfo {author} {\bibfnamefont {P.}~\bibnamefont {Vrana}},\ }\href
  {https://doi.org/10.1016/j.laa.2024.06.005} {\bibfield  {journal} {\bibinfo
  {journal} {Linear Algebra and its Applications}\ }\textbf {\bibinfo {volume}
  {699}},\ \bibinfo {pages} {159} (\bibinfo {year} {2024})}\BibitemShut
  {NoStop}%
\bibitem [{\citenamefont {Gour}\ and\ \citenamefont
  {Tomamichel}(2021)}]{GT2021}%
  \BibitemOpen
  \bibfield  {author} {\bibinfo {author} {\bibfnamefont {G.}~\bibnamefont
  {Gour}}\ and\ \bibinfo {author} {\bibfnamefont {M.}~\bibnamefont
  {Tomamichel}},\ }\href {https://doi.org/10.1109/TIT.2021.3078337} {\bibfield
  {journal} {\bibinfo  {journal} {IEEE Transactions on Information Theory}\
  }\textbf {\bibinfo {volume} {67}},\ \bibinfo {pages} {6313} (\bibinfo {year}
  {2021})}\BibitemShut {NoStop}%
\bibitem [{Note2()}]{Note2}%
  \BibitemOpen
  \bibinfo {note} {When the corresponding regularizations do not exist, we use
  the worst-case partial limits.}\BibitemShut {Stop}%
\bibitem [{\citenamefont {Takagi}\ \emph {et~al.}(2022)\citenamefont {Takagi},
  \citenamefont {Regula},\ and\ \citenamefont {Wilde}}]{trw2022}%
  \BibitemOpen
  \bibfield  {author} {\bibinfo {author} {\bibfnamefont {R.}~\bibnamefont
  {Takagi}}, \bibinfo {author} {\bibfnamefont {B.}~\bibnamefont {Regula}},\
  and\ \bibinfo {author} {\bibfnamefont {M.~M.}\ \bibnamefont {Wilde}},\ }\href
  {https://doi.org/10.1103/PRXQuantum.3.010348} {\bibfield  {journal} {\bibinfo
   {journal} {PRX Quantum}\ }\textbf {\bibinfo {volume} {3}},\ \bibinfo {pages}
  {010348} (\bibinfo {year} {2022})}\BibitemShut {NoStop}%
\bibitem [{\citenamefont {Brand{\~a}o}\ and\ \citenamefont
  {Plenio}(2008)}]{BP2008}%
  \BibitemOpen
  \bibfield  {author} {\bibinfo {author} {\bibfnamefont {F.~G. S.~L.}\
  \bibnamefont {Brand{\~a}o}}\ and\ \bibinfo {author} {\bibfnamefont {M.~B.}\
  \bibnamefont {Plenio}},\ }\href {https://doi.org/10.1038/nphys1100}
  {\bibfield  {journal} {\bibinfo  {journal} {Nature Physics}\ }\textbf
  {\bibinfo {volume} {4}},\ \bibinfo {pages} {873} (\bibinfo {year}
  {2008})}\BibitemShut {NoStop}%
\bibitem [{\citenamefont {Brandao}\ and\ \citenamefont
  {Plenio}(2009)}]{BP2009}%
  \BibitemOpen
  \bibfield  {author} {\bibinfo {author} {\bibfnamefont {F.~G. S.~L.}\
  \bibnamefont {Brandao}}\ and\ \bibinfo {author} {\bibfnamefont {M.~B.}\
  \bibnamefont {Plenio}},\ }\href
  {https://doi.org/10.1088/1742-6596/143/1/012009} {\bibfield  {journal}
  {\bibinfo  {journal} {Journal of Physics: Conference Series}\ }\textbf
  {\bibinfo {volume} {143}},\ \bibinfo {pages} {012009} (\bibinfo {year}
  {2009})}\BibitemShut {NoStop}%
\bibitem [{\citenamefont {Brand{\~a}o}\ and\ \citenamefont
  {Plenio}(2010)}]{BP2010a}%
  \BibitemOpen
  \bibfield  {author} {\bibinfo {author} {\bibfnamefont {F.~G. S.~L.}\
  \bibnamefont {Brand{\~a}o}}\ and\ \bibinfo {author} {\bibfnamefont {M.~B.}\
  \bibnamefont {Plenio}},\ }\href {https://doi.org/10.1007/s00220-010-1003-1}
  {\bibfield  {journal} {\bibinfo  {journal} {Communications in Mathematical
  Physics}\ }\textbf {\bibinfo {volume} {295}},\ \bibinfo {pages} {829}
  (\bibinfo {year} {2010})}\BibitemShut {NoStop}%
\bibitem [{\citenamefont {Brand\~ao}\ and\ \citenamefont
  {Gour}(2015)}]{BG2015}%
  \BibitemOpen
  \bibfield  {author} {\bibinfo {author} {\bibfnamefont {F.~G. S.~L.}\
  \bibnamefont {Brand\~ao}}\ and\ \bibinfo {author} {\bibfnamefont
  {G.}~\bibnamefont {Gour}},\ }\href
  {https://doi.org/10.1103/PhysRevLett.115.070503} {\bibfield  {journal}
  {\bibinfo  {journal} {Phys. Rev. Lett.}\ }\textbf {\bibinfo {volume} {115}},\
  \bibinfo {pages} {070503} (\bibinfo {year} {2015})}\BibitemShut {NoStop}%
\bibitem [{\citenamefont {Ganardi}\ \emph {et~al.}(2025)\citenamefont
  {Ganardi}, \citenamefont {Kondra}, \citenamefont {Ng},\ and\ \citenamefont
  {Streltsov}}]{RTNA2025}%
  \BibitemOpen
  \bibfield  {author} {\bibinfo {author} {\bibfnamefont {R.}~\bibnamefont
  {Ganardi}}, \bibinfo {author} {\bibfnamefont {T.~V.}\ \bibnamefont {Kondra}},
  \bibinfo {author} {\bibfnamefont {N.~H.~Y.}\ \bibnamefont {Ng}},\ and\
  \bibinfo {author} {\bibfnamefont {A.}~\bibnamefont {Streltsov}},\ }\href
  {https://doi.org/10.1103/kl56-p2vb} {\bibfield  {journal} {\bibinfo
  {journal} {Phys. Rev. Lett.}\ }\textbf {\bibinfo {volume} {135}},\ \bibinfo
  {pages} {010202} (\bibinfo {year} {2025})}\BibitemShut {NoStop}%
\bibitem [{\citenamefont {Donald}\ and\ \citenamefont
  {Horodecki}(1999)}]{DH1999}%
  \BibitemOpen
  \bibfield  {author} {\bibinfo {author} {\bibfnamefont {M.~J.}\ \bibnamefont
  {Donald}}\ and\ \bibinfo {author} {\bibfnamefont {M.}~\bibnamefont
  {Horodecki}},\ }\href
  {https://doi.org/https://doi.org/10.1016/S0375-9601(99)00813-0} {\bibfield
  {journal} {\bibinfo  {journal} {Physics Letters A}\ }\textbf {\bibinfo
  {volume} {264}},\ \bibinfo {pages} {257} (\bibinfo {year}
  {1999})}\BibitemShut {NoStop}%
\bibitem [{\citenamefont {Synak-Radtke}\ and\ \citenamefont
  {Horodecki}(2006)}]{SH2006}%
  \BibitemOpen
  \bibfield  {author} {\bibinfo {author} {\bibfnamefont {B.}~\bibnamefont
  {Synak-Radtke}}\ and\ \bibinfo {author} {\bibfnamefont {M.}~\bibnamefont
  {Horodecki}},\ }\href {https://doi.org/10.1088/0305-4470/39/26/L02}
  {\bibfield  {journal} {\bibinfo  {journal} {Journal of Physics A:
  Mathematical and General}\ }\textbf {\bibinfo {volume} {39}},\ \bibinfo
  {pages} {L423} (\bibinfo {year} {2006})}\BibitemShut {NoStop}%
\bibitem [{\citenamefont {Horodecki}\ \emph {et~al.}(2002)\citenamefont
  {Horodecki}, \citenamefont {Oppenheim},\ and\ \citenamefont
  {Horodecki}}]{HOH2002}%
  \BibitemOpen
  \bibfield  {author} {\bibinfo {author} {\bibfnamefont {M.}~\bibnamefont
  {Horodecki}}, \bibinfo {author} {\bibfnamefont {J.}~\bibnamefont
  {Oppenheim}},\ and\ \bibinfo {author} {\bibfnamefont {R.}~\bibnamefont
  {Horodecki}},\ }\href {https://doi.org/10.1103/PhysRevLett.89.240403}
  {\bibfield  {journal} {\bibinfo  {journal} {Phys. Rev. Lett.}\ }\textbf
  {\bibinfo {volume} {89}},\ \bibinfo {pages} {240403} (\bibinfo {year}
  {2002})}\BibitemShut {NoStop}%
\bibitem [{\citenamefont {Bennett}\ \emph {et~al.}(2000)\citenamefont
  {Bennett}, \citenamefont {Popescu}, \citenamefont {Rohrlich}, \citenamefont
  {Smolin},\ and\ \citenamefont {Thapliyal}}]{BPR+2000}%
  \BibitemOpen
  \bibfield  {author} {\bibinfo {author} {\bibfnamefont {C.~H.}\ \bibnamefont
  {Bennett}}, \bibinfo {author} {\bibfnamefont {S.}~\bibnamefont {Popescu}},
  \bibinfo {author} {\bibfnamefont {D.}~\bibnamefont {Rohrlich}}, \bibinfo
  {author} {\bibfnamefont {J.~A.}\ \bibnamefont {Smolin}},\ and\ \bibinfo
  {author} {\bibfnamefont {A.~V.}\ \bibnamefont {Thapliyal}},\ }\href
  {https://doi.org/10.1103/PhysRevA.63.012307} {\bibfield  {journal} {\bibinfo
  {journal} {Phys. Rev. A}\ }\textbf {\bibinfo {volume} {63}},\ \bibinfo
  {pages} {012307} (\bibinfo {year} {2000})}\BibitemShut {NoStop}%
\bibitem [{\citenamefont {Vidal}\ and\ \citenamefont {Cirac}(2001)}]{VC2001}%
  \BibitemOpen
  \bibfield  {author} {\bibinfo {author} {\bibfnamefont {G.}~\bibnamefont
  {Vidal}}\ and\ \bibinfo {author} {\bibfnamefont {J.~I.}\ \bibnamefont
  {Cirac}},\ }\href {https://doi.org/10.1103/PhysRevLett.86.5803} {\bibfield
  {journal} {\bibinfo  {journal} {Phys. Rev. Lett.}\ }\textbf {\bibinfo
  {volume} {86}},\ \bibinfo {pages} {5803} (\bibinfo {year}
  {2001})}\BibitemShut {NoStop}%
\bibitem [{\citenamefont {Lami}\ and\ \citenamefont {Regula}(2023)}]{LR2023}%
  \BibitemOpen
  \bibfield  {author} {\bibinfo {author} {\bibfnamefont {L.}~\bibnamefont
  {Lami}}\ and\ \bibinfo {author} {\bibfnamefont {B.}~\bibnamefont {Regula}},\
  }\href {https://doi.org/10.1038/s41567-022-01873-9} {\bibfield  {journal}
  {\bibinfo  {journal} {Nature Physics}\ }\textbf {\bibinfo {volume} {19}},\
  \bibinfo {pages} {184} (\bibinfo {year} {2023})}\BibitemShut {NoStop}%
\bibitem [{\citenamefont {Bennett}\ \emph {et~al.}(1996)\citenamefont
  {Bennett}, \citenamefont {Bernstein}, \citenamefont {Popescu},\ and\
  \citenamefont {Schumacher}}]{BBPS1996}%
  \BibitemOpen
  \bibfield  {author} {\bibinfo {author} {\bibfnamefont {C.~H.}\ \bibnamefont
  {Bennett}}, \bibinfo {author} {\bibfnamefont {H.~J.}\ \bibnamefont
  {Bernstein}}, \bibinfo {author} {\bibfnamefont {S.}~\bibnamefont {Popescu}},\
  and\ \bibinfo {author} {\bibfnamefont {B.}~\bibnamefont {Schumacher}},\
  }\href {https://doi.org/10.1103/PhysRevA.53.2046} {\bibfield  {journal}
  {\bibinfo  {journal} {Phys. Rev. A}\ }\textbf {\bibinfo {volume} {53}},\
  \bibinfo {pages} {2046} (\bibinfo {year} {1996})}\BibitemShut {NoStop}%
\bibitem [{\citenamefont {Popescu}\ and\ \citenamefont
  {Rohrlich}(1997)}]{PR1997}%
  \BibitemOpen
  \bibfield  {author} {\bibinfo {author} {\bibfnamefont {S.}~\bibnamefont
  {Popescu}}\ and\ \bibinfo {author} {\bibfnamefont {D.}~\bibnamefont
  {Rohrlich}},\ }\href {https://doi.org/10.1103/PhysRevA.56.R3319} {\bibfield
  {journal} {\bibinfo  {journal} {Phys. Rev. A}\ }\textbf {\bibinfo {volume}
  {56}},\ \bibinfo {pages} {R3319} (\bibinfo {year} {1997})}\BibitemShut
  {NoStop}%
\bibitem [{\citenamefont {Polyanskiy}\ and\ \citenamefont {Wu}(2025)}]{pw2025}%
  \BibitemOpen
  \bibfield  {author} {\bibinfo {author} {\bibfnamefont {Y.}~\bibnamefont
  {Polyanskiy}}\ and\ \bibinfo {author} {\bibfnamefont {Y.}~\bibnamefont
  {Wu}},\ }\href {https://doi.org/10.1017/9781108966351} {\bibinfo {title}
  {Information {{Theory}}: {{From Coding}} to {{Learning}}}} (\bibinfo {year}
  {2025})\BibitemShut {NoStop}%
\bibitem [{\citenamefont {Hickey}\ and\ \citenamefont {Gour}(2018)}]{HG2018}%
  \BibitemOpen
  \bibfield  {author} {\bibinfo {author} {\bibfnamefont {A.}~\bibnamefont
  {Hickey}}\ and\ \bibinfo {author} {\bibfnamefont {G.}~\bibnamefont {Gour}},\
  }\href {https://doi.org/10.1088/1751-8121/aabe9c} {\bibfield  {journal}
  {\bibinfo  {journal} {Journal of Physics A: Mathematical and Theoretical}\
  }\textbf {\bibinfo {volume} {51}},\ \bibinfo {pages} {414009} (\bibinfo
  {year} {2018})}\BibitemShut {NoStop}%
\bibitem [{\citenamefont {Ganardi}\ \emph {et~al.}(2026)\citenamefont
  {Ganardi}, \citenamefont {Son}, \citenamefont {Czartowski}, \citenamefont
  {Lie},\ and\ \citenamefont {Ng}}]{GJJ+2026}%
  \BibitemOpen
  \bibfield  {author} {\bibinfo {author} {\bibfnamefont {R.}~\bibnamefont
  {Ganardi}}, \bibinfo {author} {\bibfnamefont {J.}~\bibnamefont {Son}},
  \bibinfo {author} {\bibfnamefont {J.}~\bibnamefont {Czartowski}}, \bibinfo
  {author} {\bibfnamefont {S.~H.}\ \bibnamefont {Lie}},\ and\ \bibinfo {author}
  {\bibfnamefont {N.~H.~Y.}\ \bibnamefont {Ng}},\ }\href
  {https://arxiv.org/abs/2602.17803} {\bibinfo {title} {Manipulating
  heterogeneous quantum resources over a network}} (\bibinfo {year}
  {2026})\BibitemShut {NoStop}%
\bibitem [{\citenamefont {Tajima}\ and\ \citenamefont {Takagi}(2025)}]{TT2025}%
  \BibitemOpen
  \bibfield  {author} {\bibinfo {author} {\bibfnamefont {H.}~\bibnamefont
  {Tajima}}\ and\ \bibinfo {author} {\bibfnamefont {R.}~\bibnamefont
  {Takagi}},\ }\href {https://doi.org/10.1103/PhysRevLett.134.170201}
  {\bibfield  {journal} {\bibinfo  {journal} {Phys. Rev. Lett.}\ }\textbf
  {\bibinfo {volume} {134}},\ \bibinfo {pages} {170201} (\bibinfo {year}
  {2025})}\BibitemShut {NoStop}%
\bibitem [{\citenamefont {Regula}\ \emph
  {et~al.}(2018{\natexlab{b}})\citenamefont {Regula}, \citenamefont {Lami},\
  and\ \citenamefont {Streltsov}}]{RLS2018}%
  \BibitemOpen
  \bibfield  {author} {\bibinfo {author} {\bibfnamefont {B.}~\bibnamefont
  {Regula}}, \bibinfo {author} {\bibfnamefont {L.}~\bibnamefont {Lami}},\ and\
  \bibinfo {author} {\bibfnamefont {A.}~\bibnamefont {Streltsov}},\ }\href
  {https://doi.org/10.1103/PhysRevA.98.052329} {\bibfield  {journal} {\bibinfo
  {journal} {Phys. Rev. A}\ }\textbf {\bibinfo {volume} {98}},\ \bibinfo
  {pages} {052329} (\bibinfo {year} {2018}{\natexlab{b}})}\BibitemShut
  {NoStop}%
\bibitem [{\citenamefont {Gour}\ and\ \citenamefont {Winter}(2019)}]{GW2019}%
  \BibitemOpen
  \bibfield  {author} {\bibinfo {author} {\bibfnamefont {G.}~\bibnamefont
  {Gour}}\ and\ \bibinfo {author} {\bibfnamefont {A.}~\bibnamefont {Winter}},\
  }\href {https://doi.org/10.1103/PhysRevLett.123.150401} {\bibfield  {journal}
  {\bibinfo  {journal} {Phys. Rev. Lett.}\ }\textbf {\bibinfo {volume} {123}},\
  \bibinfo {pages} {150401} (\bibinfo {year} {2019})}\BibitemShut {NoStop}%
\bibitem [{\citenamefont {Liu}\ and\ \citenamefont {Yuan}(2020)}]{LY2020}%
  \BibitemOpen
  \bibfield  {author} {\bibinfo {author} {\bibfnamefont {Y.}~\bibnamefont
  {Liu}}\ and\ \bibinfo {author} {\bibfnamefont {X.}~\bibnamefont {Yuan}},\
  }\href {https://doi.org/10.1103/PhysRevResearch.2.012035} {\bibfield
  {journal} {\bibinfo  {journal} {Phys. Rev. Res.}\ }\textbf {\bibinfo {volume}
  {2}},\ \bibinfo {pages} {012035} (\bibinfo {year} {2020})}\BibitemShut
  {NoStop}%
\bibitem [{\citenamefont {Gour}\ and\ \citenamefont {Scandolo}(2020)}]{GS2020}%
  \BibitemOpen
  \bibfield  {author} {\bibinfo {author} {\bibfnamefont {G.}~\bibnamefont
  {Gour}}\ and\ \bibinfo {author} {\bibfnamefont {C.~M.}\ \bibnamefont
  {Scandolo}},\ }\href {https://arxiv.org/abs/2101.01552} {\bibinfo {title}
  {Dynamical resources}} (\bibinfo {year} {2020})\BibitemShut {NoStop}%
\bibitem [{\citenamefont {{Chiribella, G.}}\ \emph {et~al.}(2008)\citenamefont
  {{Chiribella, G.}}, \citenamefont {{D'Ariano, G. M.}},\ and\ \citenamefont
  {{Perinotti, P.}}}]{CDP2008}%
  \BibitemOpen
  \bibfield  {author} {\bibinfo {author} {\bibnamefont {{Chiribella, G.}}},
  \bibinfo {author} {\bibnamefont {{D'Ariano, G. M.}}},\ and\ \bibinfo {author}
  {\bibnamefont {{Perinotti, P.}}},\ }\href
  {https://doi.org/10.1209/0295-5075/83/30004} {\bibfield  {journal} {\bibinfo
  {journal} {EPL}\ }\textbf {\bibinfo {volume} {83}},\ \bibinfo {pages} {30004}
  (\bibinfo {year} {2008})}\BibitemShut {NoStop}%
\bibitem [{\citenamefont {Chiribella}\ \emph {et~al.}(2009)\citenamefont
  {Chiribella}, \citenamefont {D'Ariano},\ and\ \citenamefont
  {Perinotti}}]{CDP2009}%
  \BibitemOpen
  \bibfield  {author} {\bibinfo {author} {\bibfnamefont {G.}~\bibnamefont
  {Chiribella}}, \bibinfo {author} {\bibfnamefont {G.~M.}\ \bibnamefont
  {D'Ariano}},\ and\ \bibinfo {author} {\bibfnamefont {P.}~\bibnamefont
  {Perinotti}},\ }\href {https://doi.org/10.1103/PhysRevA.80.022339} {\bibfield
   {journal} {\bibinfo  {journal} {Phys. Rev. A}\ }\textbf {\bibinfo {volume}
  {80}},\ \bibinfo {pages} {022339} (\bibinfo {year} {2009})}\BibitemShut
  {NoStop}%
\bibitem [{\citenamefont {Gour}(2019)}]{Gour2019}%
  \BibitemOpen
  \bibfield  {author} {\bibinfo {author} {\bibfnamefont {G.}~\bibnamefont
  {Gour}},\ }\href {https://doi.org/10.1109/TIT.2019.2907989} {\bibfield
  {journal} {\bibinfo  {journal} {IEEE Transactions on Information Theory}\
  }\textbf {\bibinfo {volume} {65}},\ \bibinfo {pages} {5880} (\bibinfo {year}
  {2019})}\BibitemShut {NoStop}%
\bibitem [{\citenamefont {Saxena}\ \emph {et~al.}(2020)\citenamefont {Saxena},
  \citenamefont {Chitambar},\ and\ \citenamefont {Gour}}]{SCG2020}%
  \BibitemOpen
  \bibfield  {author} {\bibinfo {author} {\bibfnamefont {G.}~\bibnamefont
  {Saxena}}, \bibinfo {author} {\bibfnamefont {E.}~\bibnamefont {Chitambar}},\
  and\ \bibinfo {author} {\bibfnamefont {G.}~\bibnamefont {Gour}},\ }\href
  {https://doi.org/10.1103/PhysRevResearch.2.023298} {\bibfield  {journal}
  {\bibinfo  {journal} {Phys. Rev. Res.}\ }\textbf {\bibinfo {volume} {2}},\
  \bibinfo {pages} {023298} (\bibinfo {year} {2020})}\BibitemShut {NoStop}%
\bibitem [{\citenamefont {Gour}\ \emph {et~al.}(2025)\citenamefont {Gour},
  \citenamefont {Kim}, \citenamefont {Nateeboon}, \citenamefont {Shemesh},\
  and\ \citenamefont {Yoeli}}]{GKN+2025}%
  \BibitemOpen
  \bibfield  {author} {\bibinfo {author} {\bibfnamefont {G.}~\bibnamefont
  {Gour}}, \bibinfo {author} {\bibfnamefont {D.}~\bibnamefont {Kim}}, \bibinfo
  {author} {\bibfnamefont {T.}~\bibnamefont {Nateeboon}}, \bibinfo {author}
  {\bibfnamefont {G.}~\bibnamefont {Shemesh}},\ and\ \bibinfo {author}
  {\bibfnamefont {G.}~\bibnamefont {Yoeli}},\ }\href
  {https://doi.org/10.1103/PhysRevA.111.052424} {\bibfield  {journal} {\bibinfo
   {journal} {Phys. Rev. A}\ }\textbf {\bibinfo {volume} {111}},\ \bibinfo
  {pages} {052424} (\bibinfo {year} {2025})}\BibitemShut {NoStop}%
\bibitem [{\citenamefont {Steakley}\ \emph {et~al.}(2025)\citenamefont
  {Steakley}, \citenamefont {Zanoni},\ and\ \citenamefont
  {Scandolo}}]{SZS2025}%
  \BibitemOpen
  \bibfield  {author} {\bibinfo {author} {\bibfnamefont {S.~B.}\ \bibnamefont
  {Steakley}}, \bibinfo {author} {\bibfnamefont {E.}~\bibnamefont {Zanoni}},\
  and\ \bibinfo {author} {\bibfnamefont {C.~M.}\ \bibnamefont {Scandolo}},\
  }\href {https://arxiv.org/abs/2510.03622} {\bibinfo {title} {Towards the
  simulation of higher-order quantum resources: a general type-theoretic
  approach}} (\bibinfo {year} {2025})\BibitemShut {NoStop}%
\bibitem [{\citenamefont {Takagi}\ and\ \citenamefont {Regula}(2019)}]{TR2019}%
  \BibitemOpen
  \bibfield  {author} {\bibinfo {author} {\bibfnamefont {R.}~\bibnamefont
  {Takagi}}\ and\ \bibinfo {author} {\bibfnamefont {B.}~\bibnamefont
  {Regula}},\ }\href {https://doi.org/10.1103/PhysRevX.9.031053} {\bibfield
  {journal} {\bibinfo  {journal} {Phys. Rev. X}\ }\textbf {\bibinfo {volume}
  {9}},\ \bibinfo {pages} {031053} (\bibinfo {year} {2019})}\BibitemShut
  {NoStop}%
\bibitem [{\citenamefont {Jencova}(2021)}]{Jencova2021}%
  \BibitemOpen
  \bibfield  {author} {\bibinfo {author} {\bibfnamefont {A.}~\bibnamefont
  {Jencova}},\ }\href {https://doi.org/10.1109/TIT.2021.3070120} {\bibfield
  {journal} {\bibinfo  {journal} {IEEE Transactions on Information Theory}\
  }\textbf {\bibinfo {volume} {67}},\ \bibinfo {pages} {3945} (\bibinfo {year}
  {2021})}\BibitemShut {NoStop}%
\bibitem [{\citenamefont {Brandão}\ and\ \citenamefont
  {Plenio}(2010)}]{BP2010b}%
  \BibitemOpen
  \bibfield  {author} {\bibinfo {author} {\bibfnamefont {F.~G. S.~L.}\
  \bibnamefont {Brandão}}\ and\ \bibinfo {author} {\bibfnamefont {M.~B.}\
  \bibnamefont {Plenio}},\ }\href {https://doi.org/10.1007/s00220-010-1005-z}
  {\bibfield  {journal} {\bibinfo  {journal} {Communications in Mathematical
  Physics}\ }\textbf {\bibinfo {volume} {295}},\ \bibinfo {pages} {791}
  (\bibinfo {year} {2010})}\BibitemShut {NoStop}%
\bibitem [{\citenamefont {Carlen}\ \emph {et~al.}(2018)\citenamefont {Carlen},
  \citenamefont {Frank},\ and\ \citenamefont {Lieb}}]{CFL2018}%
  \BibitemOpen
  \bibfield  {author} {\bibinfo {author} {\bibfnamefont {E.}~\bibnamefont
  {Carlen}}, \bibinfo {author} {\bibfnamefont {R.}~\bibnamefont {Frank}},\ and\
  \bibinfo {author} {\bibfnamefont {E.}~\bibnamefont {Lieb}},\ }\href
  {https://doi.org/10.1088/1751-8121/aae8a3} {\bibfield  {journal} {\bibinfo
  {journal} {Journal of Physics A: Mathematical and Theoretical}\ }\textbf
  {\bibinfo {volume} {51}} (\bibinfo {year} {2018})}\BibitemShut {NoStop}%
\bibitem [{\citenamefont {Hayashi}\ and\ \citenamefont
  {Tomamichel}(2016)}]{HT2016}%
  \BibitemOpen
  \bibfield  {author} {\bibinfo {author} {\bibfnamefont {M.}~\bibnamefont
  {Hayashi}}\ and\ \bibinfo {author} {\bibfnamefont {M.}~\bibnamefont
  {Tomamichel}},\ }\href {https://doi.org/10.1063/1.4964755} {\bibfield
  {journal} {\bibinfo  {journal} {Journal of Mathematical Physics}\ }\textbf
  {\bibinfo {volume} {57}},\ \bibinfo {pages} {102201} (\bibinfo {year}
  {2016})}\BibitemShut {NoStop}%
\bibitem [{\citenamefont {Hayashi}\ and\ \citenamefont
  {Yamasaki}(2025)}]{HY2025}%
  \BibitemOpen
  \bibfield  {author} {\bibinfo {author} {\bibfnamefont {M.}~\bibnamefont
  {Hayashi}}\ and\ \bibinfo {author} {\bibfnamefont {H.}~\bibnamefont
  {Yamasaki}},\ }\href {https://arxiv.org/abs/2408.02722} {\bibinfo {title}
  {Generalized quantum stein's lemma and second law of quantum resource
  theories}} (\bibinfo {year} {2025})\BibitemShut {NoStop}%
\bibitem [{\citenamefont {Kubo}\ and\ \citenamefont {Ando}(1979)}]{KA1979}%
  \BibitemOpen
  \bibfield  {author} {\bibinfo {author} {\bibfnamefont {F.}~\bibnamefont
  {Kubo}}\ and\ \bibinfo {author} {\bibfnamefont {T.}~\bibnamefont {Ando}},\
  }\href
  {https://gdz.sub.uni-goettingen.de/download/pdf/PPN235181684_0246/PPN235181684_0246.pdf}
  {\bibfield  {journal} {\bibinfo  {journal} {Mathematische Annalen}\ }\textbf
  {\bibinfo {volume} {246}},\ \bibinfo {pages} {205} (\bibinfo {year}
  {1979})}\BibitemShut {NoStop}%
\bibitem [{\citenamefont {Tomiyama}(1957)}]{Tomiyama1957}%
  \BibitemOpen
  \bibfield  {author} {\bibinfo {author} {\bibfnamefont {J.}~\bibnamefont
  {Tomiyama}},\ }\href {https://doi.org/10.3792/pja/1195524885} {\bibfield
  {journal} {\bibinfo  {journal} {Proceedings of the Japan Academy}\ }\textbf
  {\bibinfo {volume} {33}},\ \bibinfo {pages} {608 } (\bibinfo {year}
  {1957})}\BibitemShut {NoStop}%
\bibitem [{\citenamefont {Gour}(2025{\natexlab{b}})}]{gour2025b}%
  \BibitemOpen
  \bibfield  {author} {\bibinfo {author} {\bibfnamefont {G.}~\bibnamefont
  {Gour}},\ }\href {https://arxiv.org/abs/2502.13669} {\bibinfo {title}
  {Induced quantum divergence: A new lens on communication and source coding}}
  (\bibinfo {year} {2025}{\natexlab{b}})\BibitemShut {NoStop}%
\bibitem [{\citenamefont {Lami}(2025)}]{lami2024}%
  \BibitemOpen
  \bibfield  {author} {\bibinfo {author} {\bibfnamefont {L.}~\bibnamefont
  {Lami}},\ }\href {https://doi.org/10.1109/TIT.2025.3543610} {\bibinfo {title}
  {A solution of the generalized quantum stein’s lemma}} (\bibinfo {year}
  {2025})\BibitemShut {NoStop}%
\bibitem [{\citenamefont {Takesaki}(1972)}]{Takesaki1972}%
  \BibitemOpen
  \bibfield  {author} {\bibinfo {author} {\bibfnamefont {M.}~\bibnamefont
  {Takesaki}},\ }\href
  {https://doi.org/https://doi.org/10.1016/0022-1236(72)90004-3} {\bibfield
  {journal} {\bibinfo  {journal} {Journal of Functional Analysis}\ }\textbf
  {\bibinfo {volume} {9}},\ \bibinfo {pages} {306} (\bibinfo {year}
  {1972})}\BibitemShut {NoStop}%
\bibitem [{\citenamefont {Petz}(1988)}]{Petz1988}%
  \BibitemOpen
  \bibfield  {author} {\bibinfo {author} {\bibfnamefont {D.}~\bibnamefont
  {Petz}},\ }in\ \href@noop {} {\emph {\bibinfo {booktitle} {Quantum
  Probability and Applications III}}},\ \bibinfo {editor} {edited by\ \bibinfo
  {editor} {\bibfnamefont {L.}~\bibnamefont {Accardi}}\ and\ \bibinfo {editor}
  {\bibfnamefont {W.}~\bibnamefont {von Waldenfels}}}\ (\bibinfo  {publisher}
  {Springer Berlin Heidelberg},\ \bibinfo {address} {Berlin, Heidelberg},\
  \bibinfo {year} {1988})\ pp.\ \bibinfo {pages} {251--260}\BibitemShut
  {NoStop}%
\bibitem [{\citenamefont {Kadison}(2004)}]{Kadison2004}%
  \BibitemOpen
  \bibfield  {author} {\bibinfo {author} {\bibfnamefont {R.~V.}\ \bibnamefont
  {Kadison}},\ }in\ \href {https://doi.org/10.1090/conm/365/06703} {\emph
  {\bibinfo {booktitle} {Operator algebras, quantization, and noncommutative
  geometry}}},\ \bibinfo {series} {Contemp. Math.}, Vol.\ \bibinfo {volume}
  {365}\ (\bibinfo  {publisher} {Amer. Math. Soc., Providence, RI},\ \bibinfo
  {year} {2004})\ pp.\ \bibinfo {pages} {143--179}\BibitemShut {NoStop}%
\bibitem [{\citenamefont {Arias}\ \emph {et~al.}(2002)\citenamefont {Arias},
  \citenamefont {Gheondea},\ and\ \citenamefont {Gudder}}]{AGG2002}%
  \BibitemOpen
  \bibfield  {author} {\bibinfo {author} {\bibfnamefont {A.}~\bibnamefont
  {Arias}}, \bibinfo {author} {\bibfnamefont {A.}~\bibnamefont {Gheondea}},\
  and\ \bibinfo {author} {\bibfnamefont {S.}~\bibnamefont {Gudder}},\ }\href
  {https://doi.org/10.1063/1.1519669} {\bibfield  {journal} {\bibinfo
  {journal} {Journal of Mathematical Physics}\ }\textbf {\bibinfo {volume}
  {43}},\ \bibinfo {pages} {5872} (\bibinfo {year} {2002})}\BibitemShut
  {NoStop}%
\bibitem [{\citenamefont {Wolf}(2012)}]{Wolf2012}%
  \BibitemOpen
  \bibfield  {author} {\bibinfo {author} {\bibfnamefont {M.~M.}\ \bibnamefont
  {Wolf}},\ }\href
  {https://citeseerx.ist.psu.edu/document?repid=rep1&type=pdf&doi=afa58291b0b8bd47504acb1ab8f553f0b37685cf}
  {\bibinfo {title} {Quantum channels and operations - guided tour}} (\bibinfo
  {year} {2012})\BibitemShut {NoStop}%
\bibitem [{\citenamefont {Gao}\ and\ \citenamefont {Rouz{\'e}}(2022)}]{GR2022}%
  \BibitemOpen
  \bibfield  {author} {\bibinfo {author} {\bibfnamefont {L.}~\bibnamefont
  {Gao}}\ and\ \bibinfo {author} {\bibfnamefont {C.}~\bibnamefont
  {Rouz{\'e}}},\ }\href {https://doi.org/10.1007/s00205-022-01785-1} {\bibfield
   {journal} {\bibinfo  {journal} {Archive for Rational Mechanics and
  Analysis}\ }\textbf {\bibinfo {volume} {245}},\ \bibinfo {pages} {183}
  (\bibinfo {year} {2022})}\BibitemShut {NoStop}%
\bibitem [{Note3()}]{Note3}%
  \BibitemOpen
  \bibinfo {note} {The generalized quantum Stein's lemma is most commonly
  formulated in terms of the hypothesis-testing divergence. Here we present it
  in an equivalent form using the smoothed max-divergence \(D_{\protect \rm
  max}^\eta \); this form is also known as the generalized asymptotic
  equipartition property~\cite {Gour2025}.}\BibitemShut {Stop}%
\bibitem [{\citenamefont {Choi}\ and\ \citenamefont {Effros}(1977)}]{ce1977}%
  \BibitemOpen
  \bibfield  {author} {\bibinfo {author} {\bibfnamefont {M.-D.}\ \bibnamefont
  {Choi}}\ and\ \bibinfo {author} {\bibfnamefont {E.~G.}\ \bibnamefont
  {Effros}},\ }\href {https://doi.org/10.1016/0022-1236(77)90052-0} {\bibfield
  {journal} {\bibinfo  {journal} {Journal of Functional Analysis}\ }\textbf
  {\bibinfo {volume} {24}},\ \bibinfo {pages} {156} (\bibinfo {year}
  {1977})}\BibitemShut {NoStop}%
\end{thebibliography}%
\onecolumngrid

\setcounter{secnumdepth}{3}\newpage
\appendix

\begin{center}
  {\large\bfseries APPENDIX}
\end{center}\vspace{-20pt}
\tableofcontents

\setlength{\parindent}{0pt}               
\setlength{\parskip}{6pt plus 2pt minus 1pt} 
\setcounter{section}{0}
\setcounter{subsection}{0}
\renewcommand{\thesection}{\Alph{section}}
\renewcommand{\thesubsection}{\thesection.\arabic{subsection}}
\numberwithin{equation}{section}
\renewcommand{\theequation}{\thesection\arabic{equation}}
\renewcommand{\theHequation}{\thesection\arabic{equation}}
\makeatletter
\def\p@subsection{}
\def\p@subsubsection{}
\makeatother

\renewcommand{\terminal}[1]{\Delta\pa*{#1}}
\newcommand{\condExpectation}{\Delta^{\!*}}
\newcommand{\expectation}[1]{\condExpectation\pa*{#1}}
\newcommand{\optimizer}{\sigma_*}
\newcommand{\tpExpectation}{\bar{\Delta}}
\newcommand{\tpTerminal}[1]{\tpExpectation\pa*{#1}}
\newcommand{\algebra}{\mathscr{A}}
\newcommand{\linearOperators}{\mathfrak{L}}
\newcommand{\positiveOperators}{\pos}
\newcommand{\creator}[1]{\mathcal{C}_{#1}}
\newcommand{\mteq}[1]{Eq.~(\ref{#1}) in the main text}
\newcommand{\mteqs}[2]{Eqs.~(\ref{#1}) and~(\ref{#2}) in the main text}
\newcommand{\mteqsrange}[2]{Eqs.~(\ref{#1})--(\ref{#2}) in the main text}
\newcommand{\mtthm}[1]{Theorem~\ref{#1} in the main text}
\newcommand{\mtlem}[1]{Lemma~\ref{#1} in the main text}
\section{Conditional Expectations}\label{appendix:conditional-expectation}
This appendix records standard facts used in the main text: an idempotent quantum channel
$\Delta$ that admits a full-rank fixed state induces a \emph{faithful conditional expectation}
in the Heisenberg picture, and such maps admit an explicit block structure.
We include these facts for completeness; for general treatments, see~\cite{Takesaki1972,Petz1988,Kadison2004,AGG2002,Wolf2012,GR2022,JLR2023} and references therein.

Throughout, $H$ is finite-dimensional and $\linearOperators(H)$ denotes the algebra of linear operators on $H$.
A \emph{channel} is a completely positive trace-preserving (CPTP) map.
For a linear map $\Phi:\linearOperators(H)\to\linearOperators(H')$ we write $\Phi^*$ for its Hilbert-Schmidt adjoint,
defined by $\tr[\Phi(X)Y]=\tr[X\Phi^*(Y)]$ for \(X\in \linearOperators(H)\) and \(Y\in \linearOperators(H')\).

\subsection{Characterizing Conditional Expectations}

A completely positive unital idempotent \(\condExpectation:\linearOperators(H)\to\linearOperators(H)\) with \(\ima \condExpectation = \algebra\subset \linearOperators(H)\) is a \emph{conditional expectation onto \(\algebra\)} if it is \(\algebra\)-bimodular:
\begin{equation}\label{eq:bimodularity}
  \condExpectation(X_1YX_2) = X_1\condExpectation(Y) X_2\;, \qquad \forall X_1,X_2\in \algebra\;,\;Y\in \linearOperators(H)\;.
\end{equation}
 In that case, \(\algebra\) is a *-algebra.
 \begin{remark}
  The image of a unital Hermitian-preserving map is generally an \emph{operator system}: a subspace closed under adjoints that contains the unit \(I\).
 \end{remark}

\begin{claim}\label{claim:conditional}
  Let \(\condExpectation: \linearOperators(H) \to \linearOperators(H)\) be a positive unital idempotent with \(\ima\condExpectation=\algebra\). The following are equivalent:
  \begin{enumerate}[label=(\roman*)] 
    \item\label{condition:cond-expectation} \(\condExpectation\) is a conditional expectation onto \(\algebra\).
    \item\label{condition:algebra} \(\algebra\) is a *-algebra.
    \item\label{condition:closedness-under-powers} \(\algebra\) is closed under multiplication by the adjoint: if \(X \in \algebra\), then \(X^*X \in \algebra\).
    \item\label{condition:Kraus} \(\condExpectation(\cdot) = \sum_{x\le n} K_x^*(\cdot) K_x\) and \(\algebra\) is the commutant of all the Kraus operators \(\pc{K_x,K_x^* \smallgiven x\le n}'\).
  \end{enumerate}
\end{claim}
\begin{proof}
  Items~\ref{condition:cond-expectation} and~\ref{condition:algebra} are equivalent. This follows directly from Tomiyama's conditional expectation theorem~\cite{Tomiyama1957}.
 
  Item~\ref{condition:algebra} implies item~\ref{condition:closedness-under-powers}. Trivial.
  
  Item~\ref{condition:closedness-under-powers} implies item~\ref{condition:algebra}. Let \(X,Y^*\in \algebra\).
  From the polarization identity
  \begin{equation}
    Y^*X = \frac{1}{4}
    \pa[\big]{{(X+Y)}^*(X+Y)-{(X-Y)}^*(X-Y)+i{(X+iY)}^*(X+iY)-i{(X-iY)}^*(X-iY)}\;,
  \end{equation}
  we deduce that \(Y^*X \in \algebra\). This shows that \(\algebra\) is an algebra.

 Item~\ref{condition:Kraus} implies item~\ref{condition:algebra}. The commutant of any set is an algebra.

 Item~\ref{condition:cond-expectation} implies item~\ref{condition:Kraus}. 
 Conditional expectations are completely positive by definition and therefore admit a Kraus form \(\condExpectation(\cdot) = \sum_{x\le n} K_x^*(\cdot) K_x\); we therefore need to establish:
 \begin{equation}
  \algebra = \pc{K_x,K_x^* \smallgiven x\le n}'\;.
 \end{equation}
 (\(\supseteq\))
If \(X\in \pc{K_x,K_x^* \smallgiven x\le n}'\) then \(\condExpectation(X)=\sum_{x \le n} K_x^*  X K_x = X \sum_x K_x^* K_x = X I = X\), so \(X \in \algebra\).
  
(\(\subseteq\)) Let \(X\in\algebra\).
   By symmetry, it suffices to check that \(X\) commutes with the operators \(K_x\) (since \(\algebra\) is an *-algebra and thus contains \(X^*\)).
   To verify this, let us show that \(\sum_x {[K_x,X]}^*[K_x,X] = 0\).
   Opening the commutator brackets and using \(\condExpectation(\cdot) = \sum_x K_x^*(\cdot) K_x\) yields
  \begin{align*}
    \sum_{x\le n} {[K_x,X]}^*[K_x,X] &
     = X^*\condExpectation(I) X 
     + \condExpectation(X^*X) 
     - \condExpectation(X^*) X 
     - X^*\condExpectation(X)\;.
  \end{align*}
  Using \(\condExpectation(X)=X\) (by assumption), \(\condExpectation(I)=I\), \(\condExpectation(X^*)=X^*\), and \(\condExpectation(X^*X)=X^*X\) (because \(\algebra\defeq \ima \condExpectation\) is a *-algebra) then implies:
  \begin{equation}
    \sum_{x\le n} {[K_x,X]}^*[K_x,X] = X^*X + X^*X - X^*X - X^*X = 0\;.
  \end{equation}
  Since each summand is non-negative, we conclude \([K_x,X]=0\) for all \(x\).
\end{proof}
\begin{definition}
  We say a conditional expectation \(\condExpectation\) is \emph{faithful} if \(\Delta(\sigma)=\sigma\) for some full-rank state \(\sigma\).
\end{definition}
Not every idempotent channel induces a conditional expectation in the Heisenberg picture, however:
\begin{claim}
  Let \(\Delta\) be an idempotent channel \textbf{that admits a full-rank fixed state} \(\sigma\).
  Then \(\condExpectation\) is a (faithful) conditional expectation.
\end{claim}
\begin{proof} 
    By Claim~\ref{claim:conditional} (item~\ref{condition:closedness-under-powers}), 
    it is enough to prove that if \(X\in \ima \condExpectation\) then \(X^*X \in \ima \condExpectation\).
    As a unital 2-positive map,
    \(\condExpectation\) satisfies the Kadison-Schwartz inequality:
    \begin{equation}\label{eq:kadison-inequality}
       \condExpectation( X^* X) \ge \condExpectation( X^*)\condExpectation( X) = X^*X\;,
    \end{equation}
    or equivalently, 
         \begin{equation}\label{eq:kadison-inequality2}
      K \coloneqq \condExpectation( X^* X) - X^* X \ge 0\;.
     \end{equation} 
Since \(\sigma\) is a fixed state, \( \tr [\sigma  X^* X] = \tr[\Delta(\sigma)X^*X] = \tr [\sigma \condExpectation( X^* X)]\), or, equivalently
     \( \tr [\sigma K ]= 0\).
   Since \(K \ge 0\) and \(\sigma\) is a full-rank state, \(K\)  must be \(0\), or equivalently, \(\condExpectation(X^*X) = X^*X\).
    \end{proof}
\begin{remark}
  By duality, \(\Delta\) is also \(\algebra\)-bimodular:
   \begin{equation}\label{eq:bimodularity-schrödinger}
  \Delta(X_1YX_2) = X_1\Delta(Y) X_2\;, \qquad \forall X_1,X_2\in \algebra\;,\;Y\in \linearOperators(H)\;.
\end{equation}
\end{remark}
\subsection{The Structure of a Faithful Conditional Expectation}
This section provides a constructive characterization of general destruction channels (idempotent channels with full-rank fixed state).
We follow Sec.~2.1 in Ref.~\cite{GR2022}.
Any \(*\)-subalgebra of \(\linearOperators(H)\) has the following structure: 
\begin{equation}\label{eq:algebra-decomposition}
  \algebra = \bigoplus_{i} \mathbb{C}I^{A_i} \otimes \linearOperators(B_i)\;,
\end{equation}
where \(H = \bigoplus_{i} H_i \) and each \(H_i\) decomposes into \(H_i = A_i \otimes B_i\) for all \(i\).
Any faithful conditional expectation \(\condExpectation\) onto \(\algebra\) is completely determined by full-rank density operators \(\tau_i\) on \(A_i\); in the Schrödinger-picture:
\begin{equation}\label{eq:terminal-decomposition}
  \Delta(\rho) = \bigoplus_{i} \tau_i \otimes \rho_i 
\end{equation}
where \(\rho_i \coloneqq \tr_{A_i} \Pi_i \rho \Pi_i\) for \(\Pi_i\) the projection onto the subspace \(H_i \subset H\).
Note that
\begin{equation}\label{eq:destruction-defining-state}
  \terminal{I} = \bigoplus_i d_{A_i} \tau_i \otimes I^{B_i}\;.
\end{equation}
\(\terminal{I}\) commutes with all fixed states \(\terminal{\rho}\) and all \(\Delta^*\)-fixed observables \(\expectation{X}\).
By Eq.~(\ref{eq:destruction-defining-state}), \(\Delta\) is self adjoint if and only if it is unital (\(\tau_i = \u^{A_i}\) for all \(i\)). In that case, it is called the \emph{trace-preserving conditional expectation} onto \(\algebra\) and has the form: 
\begin{equation}
  \Delta(\rho) = \bigoplus_{i} \u^{A_i} \otimes \rho_i \;.
\end{equation}
For a general faithful conditional expectation \(\condExpectation\) onto the algebra \(\algebra\), we denote the trace-preserving conditional expectation onto \(\algebra\) by \(\tpExpectation\).
In that case, 
\begin{align}\label{eq:pre-twisting1}
  \terminal{\rho} &=\tpTerminal{\rho} \terminal{I}=\terminal{I} \tpTerminal{\rho} = \terminal{I}^{\frac{1}{2}} \tpTerminal{\rho} \terminal{I}^{\frac{1}{2}}\;,\\ \label{eq:pre-twisting2}
  \expectation{X} &= \tpExpectation\pa[\big]{\terminal{I}^{1/2} X \terminal{I}^{1/2}}\;.
\end{align}
\newcommand{\twisting}{\mathcal{T}}

\section{The Quantum Resource Theory of Instability}\label{appendix:fragility-theory}
We adopt the definition of a quantum resource theory from~\cite{ZS2025}.
A resource theory \(\free\) specifies a collection of admissible physical systems \(A,B,\dots\) and a set of physical operations \(\free(A \to B)\subset \cptp(A \to B)\) between any two admissible systems \(A\) and \(B\), termed ``free operations''.
Here, \(\cptp(A \to B)\) denotes the set of all quantum channels from system \(A\) to system \(B\). 
We identify \(\free\) with the overall class of free operations.
The specification has to satisfy certain axioms (see below).

In \emph{instability}, the admissible physical systems are each equipped with a destruction channel, and the free operations are the destruction-covariant channels \(\mathcal{N}\) satisfying \(\Delta\circ \mathcal{N} = \mathcal{N}\circ \Delta\).

\paragraph*{Necessity of the full-rank assumption.}
The simplifying assumption in the main text excludes the following operational degeneracy.
\begin{claim}[Universal Source State]\label{claim:universal-source}
  Let \((A,\Delta^{\!A})\) be a system whose idempotent destruction channel has no full-rank fixed state.
  Then there exists a state \(\rho^A\) such that, for every admissible system \((B,\Delta^{\!B})\) and every \(\sigma^B\in\states(B)\), there is a destruction-covariant channel \(\mathcal{N}_{\sigma}^{A\to B}\) satisfying \(\mathcal{N}_{\sigma}(\rho^A)=\sigma^B\).
\end{claim}
\begin{proof}
  By assumption \(\Delta^{\!A}(\u^A)\) is a singular fixed state.
  Let \(P\) be the support projection of \(\Delta^{\!A}(\u^A)\), let \(Q\coloneqq I^A-P\ne0\), and choose any state \(\rho^A\) supported on \(Q\).

  Every output of \(\Delta^{\!A}\) is supported on \(P\).
  Indeed, for \(X\ge0\),
  \begin{equation}
    X\le \norm*{X}_{\infty}I^A
    \quad\Longrightarrow\quad
    \Delta^{\!A}(X)\le \norm*{X}_{\infty}\Delta^{\!A}(I^A),
  \end{equation}
  and therefore \(\supp\Delta^{\!A}(X)\subseteq\supp\Delta^{\!A}(I^A)\).
  By linearity, \(P\Delta^{\!A}(X)P=\Delta^{\!A}(X)\) for every operator \(X\).

  For an arbitrary target state \(\sigma^B\), define the measure-and-prepare channel
  \begin{equation}
    \mathcal{N}_{\sigma}(X)
    \coloneqq
    \tr[QX]\,\sigma^B
    +
    \tr[PX]\,\Delta^{\!B}(\sigma^B).
  \end{equation}
  Because \(\rho^A\) is supported on \(Q\), \(\mathcal{N}_{\sigma}(\rho^A)=\sigma^B\).
  Moreover,
  \begin{equation}
    (\mathcal{N}_{\sigma}\circ\Delta^{\!A})(X)
    =\tr[X]\,\Delta^{\!B}(\sigma^B),
  \end{equation}
  whereas idempotence of \(\Delta^{\!B}\) gives
  \begin{align}
    (\Delta^{\!B}\circ\mathcal{N}_{\sigma})(X)
    &=\tr[QX]\,\Delta^{\!B}(\sigma^B)
      +\tr[PX]\,(\Delta^{\!B})^2(\sigma^B)\\
    &=\tr[X]\,\Delta^{\!B}(\sigma^B).
  \end{align}
  Thus \(\mathcal{N}_{\sigma}\circ\Delta^{\!A}=\Delta^{\!B}\circ\mathcal{N}_{\sigma}\), so \(\mathcal{N}_{\sigma}\) is destruction-covariant and \(\rho^A\) is a universal source state.
\end{proof}
\begin{remark}
  The projector \(Q\) is a perfect flag for whether destruction has occurred: \(\tr[Q\rho^A]=1\) and \(\Delta^{\!A*}(Q)=0\).
  Thus, in the notation of Sec.~\ref{appendix:one-shot-yield}, \(Q\in\scalarTerminal_{\infty}\), and by \mtlem{thm:one-shot}, \(\yield[0](\rho^A)=h^0(\rho^A)=\infty\).
  The weaker condition \(\supp\rho\nsubseteq\supp\Delta(\rho)\) would not suffice: it need not make \(\rho\) perfectly distinguishable from every possible output of destruction.
\end{remark}

We now verify that this definition satisfies the axioms of a quantum resource theory.
\begin{enumerate}[label=(\arabic*)] 
  \item \textbf{The identity channel is free.} \(\id\) commutes with every channel; in particular, \(\id\circ\Delta = \Delta \circ \id\).
  \item \textbf{Swap is free.} To see this, apply the swap channel \(\mathsf{swap}^{AB\to BA} \) on an arbitrary product state \begin{equation}
    \mathsf{swap}^{AB\to BA} \circ \Delta(\rho^A \otimes \sigma^B)
    =  \Delta(\sigma^B)\otimes \Delta(\rho^A)
    = \Delta \circ \mathsf{swap}^{AB\to BA} (\rho^A \otimes \sigma^B)\;,
  \end{equation}
   which holds by locality of destruction. 
  Because separable states span \(\linearOperators(AB)\), we conclude by linearity that \(\mathsf{swap}\) is destruction-covariant.
  \item  \textbf{Discarding information is free.} Discarding a system \(A\) is the trace map into the trivial system \(\C\). The only channel on this system is \(\id^\C\) and therefore we declare \(\Delta^{\! \C} = \id ^\C\).
  
  Since \(\Delta^{\!A}\) is trace preserving, \(\tr_A \circ\Delta^{\! A} = \tr_A  = \id^\C \circ \tr_A = \Delta^{\! \C} \circ \tr _A\).
  \item \textbf{Sequential composition is free.} If \(\mathcal{N, M}\) are free (and sequentially composable), then 
 \begin{equation}
   \Delta\circ \mathcal{M\circ N} = \mathcal{M}\circ\Delta\circ \mathcal{N} = \mathcal{M \circ N}\circ\Delta \;.
 \end{equation}
  \item \textbf{Parallel composition is free.} By locality of destruction, \(\Delta^{AB}=\Delta^A\otimes\Delta^B\),
  \begin{equation}
    \Delta\circ(\mathcal{M\otimes N}) = (\Delta \circ \mathcal{M})\otimes(\Delta \circ \mathcal{N}) =(\mathcal{M}\circ \Delta )\otimes(\mathcal{N}\circ \Delta) = (\mathcal{M\otimes N})\circ \Delta\;.
  \end{equation}
\end{enumerate}
The relation \(\Delta\circ \mathcal{N} = \mathcal{N}\circ \Delta\) is linear in \(\mathcal{N}\). Therefore, \(\free(A \to B)\) is a convex set, and more specifically, an affine subset of the set of all quantum channels. This makes instability an \emph{affine} resource theory~\cite{Gour2017,CG2019}.

The free states in a quantum resource theory are determined by the free channels as follows: we identify each state \(\sigma\in \states(A)\) with its `creator' channel \(\creator{\sigma} \in \cptp(\C \to A)\) (that satisfy \(\creator{\sigma}(\lambda) = \lambda \sigma\)).
The free states are those \(\sigma\in \states(A)\) for which \(\creator{\sigma}\) is free.

In instability theory, because \(\Delta^{\!\C}= \id^\C\), \(\creator{\sigma}\) is destruction-covariant iff \(\sigma\) is fixed by \(\Delta^{\!A}\).
 Thus, the free states are exactly the \(\Delta\)-fixed states.
 We henceforth denote the set of free states on the system \(A\) by \(\free(A)\).

\section{Monotones}\label{appendix:monotones}
This appendix establishes the structural and mathematical properties of instability monotones.
We begin in Sec.~\ref{appendix:interpolation-optimized-direct} by constructing a broad class of intermediate monotones.
In Sec.~\ref{appendix:optimization} we analyze an optimization problem that generalizes the minimization in computing \(D_{\alpha,z} \divergence{\rho}{\free}\).
We apply this machinery in Sec.~\ref{appendix:renyi} to prove the additivity of the Rényi instability monotones [\mtthm{thm:3-renyi}].
Finally, in Sec.~\ref{appendix:extremals}, we establish the extremality of \(\dmin \divergence{\rho}{\free}\) and \(\dmax \divergence{ \rho}{\Delta(\rho)}\) axiomatically among the family of (normalized) weakly additive monotones.

Throughout this appendix \(\positiveOperators(A) \subset \linearOperators(A)\) denotes the cone of positive semidefinite operators on a Hilbert space \(A\), \(\states(A) \subset\positiveOperators(A) \) the subset of quantum states and \(\free(A)\) the \(\Delta\)-fixed states.
\subsection{Intermediate Monotone \(\D\)-Extensions}\label{appendix:interpolation-optimized-direct}
\newcommand{\mean}{\star}
Let \(\mean\) be a binary action on all positive operators across all systems, decomposing to individual systems as:
\begin{equation}
  \begin{aligned} \label{def:mean}
    \mean:{\positiveOperators(A)}^2 & \to \positiveOperators(A)\\
     (X,Y) & \mapsto X \mean Y\;.
  \end{aligned}
\end{equation}
Additionally, assume that \(\mathcal{N}(X\mean Y) \le \mathcal{N}(X) \mean \mathcal{N}(Y)\) for any \(\mathcal{N} \in \cptp(A \to B)\), and for all \(X,Y \in \positiveOperators(A)\).
\begin{remark}
  The complementary action \(\hat{\mean} \coloneqq (X,Y) \mapsto Y \mean X\) satisfies the same properties. 
\end{remark}
Specific examples of such binary actions can be arithmetic means
  \((X , Y) \mapsto (1-\lambda) X +\lambda Y\)
 and geometric means
 \begin{equation}\label{eq:def-geometric-mean}
  X \#_\lambda Y=X ^{1/2} {(X^{-1/2} Y X^{-1/2})}^\lambda X^{1/2}\;, \qquad \lambda\in [0,1]\;,
 \end{equation}
 with the continuous extension to singular \(X\).
More generally, \(\mean\) can represent any Kubo-Ando \emph{connection}~\cite{KA1979}. 

Recall the definition of quantum divergence~\cite{GT2020}:
\begin{definition}[\cite{GT2020}]
 A quantum divergence is a function \(\D: \bigcup_A {\states(A)}^2 \to (-\infty,\infty]\) acting on pairs of quantum states across all finite-dimensional quantum systems and satisfies the DPI, namely: 
  \begin{equation}
    \D \divergence{\mathcal{N}({\rho})}{\mathcal{N}(\sigma)} \le \D \divergence{\rho}{\sigma} \qquad\qquad \forall \rho,\sigma \in \states(A)\;,\; \mathcal{N} \in \cptp(A \to B)
  \end{equation}
\end{definition}
Following~\cite{gour2025b}, a divergence \(\D \divergence{\cdot}{\cdot}\) defined on pairs of quantum states may be naturally extended so that the second argument can be any positive semidefinite operator, so that \(\D: \bigcup_A \states(A) \times \positiveOperators(A) \to (-\infty,\infty]\). This is usually done by setting 
\begin{equation}\label{eq:Lowner-normalization}
  \D \divergence{\rho}{t\sigma} \coloneqq \D \divergence{\rho}{\sigma} - \log t\;, \qquad \forall t \in (0,\infty)\;.
\end{equation}
In fact, all the quantum divergences introduced in the main paper satisfy Eq.~(\ref{eq:Lowner-normalization}) inherently.
\begin{observation}[\cite{gour2025b}]
A quantum divergence  \(\D: \bigcup_A \states(A) \times \positiveOperators(A) \to (-\infty,\infty]\) satisfies: 
\begin{enumerate}
  \item Data-processing inequality: \(\displaystyle \D \divergence{\mathcal{N}({\rho})}{\mathcal{N}(X)} \le \D \divergence{\rho}{X}\;,\qquad \forall \rho \in \states(A)\;,\; X \in \positiveOperators(A)\;,\; \mathcal{N} \in \cptp(A \to B)\).
  \item\label{Löwner-monotonicity} Löwner monotonicity: \(\displaystyle X \ge Y \Rightarrow \D \divergence{\rho}{X} \le \D \divergence{\rho}{Y}\;\).
\end{enumerate}
\end{observation}
With this in mind, we can define many more instability monotones that extend \(\D\).
\begin{claim}\label{claim:mean-optimized}
  Let \(\D: \bigcup_A \states(A) \times \positiveOperators(A) \to (-\infty,\infty]\) be a divergence and \(\mean\) as defined around Eq.~(\ref{def:mean}), then 
  \begin{enumerate}
    \item \(\D \divergence{\rho}{\omega \mean \sigma} \ge \D \divergence{\mathcal{N}(\rho)}{\mathcal{N}(\omega) \mean \mathcal{N}(\sigma)}\) for all \(\rho,\omega,\sigma\in \states(A)\) and \(\mathcal{N}\in \cptp(A \to B)\).
    \item The functions
 \begin{equation}
  M_\D^\mean(\rho) \coloneqq \inf_{\sigma\in \free}\D \divergence{\rho}{\Delta(\rho)\mean \sigma}\;,\;
  \text{\emph{and}}\quad
   M_\D^{\hat{\mean}}(\rho) \coloneqq \inf_{\sigma\in \free}\D \divergence{\rho}{\sigma \mean \Delta(\rho)}
\end{equation}
 are instability monotones.
  \end{enumerate}
\end{claim}
\begin{remark}
  If \(\mean\) is normalized so that \(1 \mean 1 = 1\), then by monotonicity under channels \(\sigma \mean \sigma =\sigma\) for all states \(\sigma\).
  To observe this, let \(\creator{\sigma}\in \cptp(\C \to A)\) be the creator of \(\sigma\) and note that by the monotonicity of \(\mean\) under quantum channels, \(\sigma \mean \sigma = \creator{\sigma}(1) \mean \creator{\sigma}(1) \ge \creator{\sigma}(1 \mean 1) = \creator{\sigma}(1) = \sigma\). On the other hand, \(\tr [\sigma\mean \sigma] \le \tr [\sigma] \mean \tr [\sigma] = 1 \mean 1 = 1\).

  Under this normalization, \(M_\D^\mean\) restricts back to the divergence \(\D \divergence{\cdot}{\gamma}\) on thermodynamic systems with Gibbs state \(\gamma\). 
\end{remark}
\begin{proof}
  1.
    Monotonicity of \(\mean\) under CPTP maps means that \(\mathcal{N}\pa*{\omega}\mean \mathcal{N}\pa*{\sigma} \ge \mathcal{N}\pa[\big]{\omega\mean \sigma}\);
  Thus, by Löwner monotonicity of \(\D\),
  \begin{align}
     \D \divergence{\mathcal{N}(\rho)}{\mathcal{N}\pa*{\omega}\mean \mathcal{N}\pa*{\sigma}}
       & \le 
        \D \divergence{\mathcal{N}(\rho)}{\mathcal{N}\pa[\big]{\omega\mean \sigma}}
        \\ 
        \explain{DPI}& \le
         \D \divergence{\rho}{\omega\mean \sigma}\;.
  \end{align}

  2. Since the complementary action \(\hat{\mean} \coloneqq (X,Y) \mapsto Y \mean X\) satisfies the same properties as \(\mean\), it is enough to prove that \(M_\D^{\mean}\) is a monotone.
   
  Let \(\mathcal{N} \in \free(A \to B)\) and \(\rho\in \states(A)\).
  For fixed \(\sigma\in \free(A)\), by the previous item,
  \begin{align}
    \D \divergence{\rho}{\Delta(\rho)\mean \sigma} & \ge \D \divergence{\mathcal{N}(\rho)}{\mathcal{N} \pa[\big]{\Delta(\rho)} \mean \mathcal{N}(\sigma)} \\
    \explain{Destruction Covariance}& =
    \D \divergence{\mathcal{N}(\rho)}{\Delta \pa[\big]{\mathcal{N}(\rho)} \mean \mathcal{N}(\sigma)} \\
    \explain{\(\mathcal{N}(\sigma) \in \free(B)\)} &\ge \inf _{\omega \in \free(B)} \D \divergence{\mathcal{N}(\rho)}
    {(\Delta\circ\mathcal{N}\pa*{\rho})\mean \omega} 
    \\ & \defeq M_\D^\mean(\mathcal{N}(\rho))\;. 
  \end{align}
However, \(\sigma\in \free(A)\) was arbitrary, so
\( \displaystyle  M_{\D}^{\mean}(\rho) \defeq\!\! \inf_{\sigma \in \free(A)} \!\!
      \D \divergence{\rho}{\Delta\pa*{\rho}\mean  \sigma} 
      \ge M_\D^\mean(\mathcal{N}(\rho))\;.
\)
\end{proof}

\subsection{Optimizing Trace Functionals}\label{appendix:optimization}
In this section all operators act on system \(A\).

Recall the \(\alpha\)-\(z\) Rényi divergence for \(\alpha,z \in (0,\infty)\)
\begin{align}
  D_{\alpha,z} \divergence{\rho}{\sigma} \coloneqq \frac 1 {\alpha-1} \log Q_{\alpha,z} \divergence{\rho}{\sigma} \;,
\end{align}
where
\begin{align}\label{eq:Q-alpha_z}
    Q_{\alpha,z} \divergence{\rho}{\sigma} &= \tr \pb[\big]{\pa[\big]{\rho^{\frac \alpha {2z}}  \sigma ^{\frac{1-\alpha}{z}}\rho^{\frac \alpha {2z}}}^{z}}\;.
\end{align}
(Recall that powers of operators are taken over their support).
We restrict the parameters \(\alpha\) and \(z\) to the following DPI region~\cite{Zhang2020}
\begin{equation}\label{dpi-domain}
  \region \coloneqq\pc*{(\alpha, z) \in {(0,\infty)}^2 \smallgiven \begin{matrix}
    \alpha \in (0,1),z \ge \max \pc{\alpha, 1 - \alpha}\;, & \text{or}  \\
    \alpha >1 \,,\; \max\pc{\alpha/2, \alpha - 1} \le z \le \alpha
  \end{matrix}}
\end{equation} 

Our goal is to find the minimizing state \(\optimizer \in \free(A)\) such that 
\(
  D_{\alpha,z} \divergence{\rho}{\optimizer} = D_{\alpha,z} \divergence{\rho}{\free}
 \).
Equivalently, \(\optimizer\) optimizes \(\sigma\mapsto Q_{\alpha,z} \divergence{\rho}{\sigma}\) over \(\free(A)\) (maximizes for \(\alpha<1\), minimizes for \(\alpha>1\), and at \(\alpha=1\) corresponds to the limiting Umegaki case).
We will write this as 
\begin{equation}
  Q_{\alpha,z} \divergence{\rho}{\optimizer} = \opt Q_{\alpha,z} \divergence{\rho}{\free}\;,
\end{equation}
and interpret \(\opt= \min/\max\), depending on the objective. 

In what follows, we generalize this optimization problem to more general trace functionals.
We fix \(r\in [-1,1]\), \(z\in(0,\infty)\) and define a trace functional over \(X,S\in \positiveOperators(A)\):
\begin{align}\label{eq:functional}
  F_X(S) \coloneqq \begin{cases}
    +\infty & r<0 \wedge \supp X \not\subseteq \supp S \\
    \tr \pb[\big]{{(S^{\frac r 2} X S^{\frac r 2})}^z}=\tr\pb[\big]{{(Z S^r Z^*)}^z}\;, & \text{else}
  \end{cases}
\end{align}
where \(Z^*Z = X\) and the rightmost equality holds because \(Y^*Y\) and \(YY^*\) have the same eigenvalues.
Most of the times we will fix \(X\) and so it will be convenient to write \(F\coloneqq F_X\).
Following~\cite{rubboli-phd-2025} (Definition 3), 
suppose that the parameters \(r,z\) lie within the domain:
\begin{align}
  \region_F &\coloneqq \region_F^{\smile} \cup \region_F^{\frown}\;, \qquad \qquad \text{where}
  \\
  \nonumber
  \\
  \region_F^{\smile} & \coloneqq [-1,0) \times (0,\infty) \qquad \qquad \text{and} \\ 
  \region_F^{\frown}
& \coloneqq \pc*{(r, z) \ne (1,1) \smallgiven 
    0<r\le1 ,\; 0<z\le r^{-1} } \;.
\end{align}
Note that if \(r=\frac {1-\alpha}z\), then (i) \((\alpha,z)\in \region\) implies \((r,z) \in \region_F\) and (ii) if \(X=\rho^{\alpha/z}\), then \(Q_{\alpha,z}\divergence{\rho}{\sigma} = F(\sigma)\)
(the case \(\alpha=1, z>0\) is a limiting case that corresponds to the Umegaki relative entropy and we do not treat it directly).
By~\cite{CFL2018} (Proposition 5) and~\cite{rubboli-phd-2025} (Lemmas 3.2, 3.3), \(F\) is \emph{convex}, \emph{Löwner antitone} and \emph{lower semicontinuous} on \(\region_F^{\smile}\), and is \emph{concave}, \emph{Löwner monotone} and \emph{upper semicontinuous} on \(\region_F^{\frown}\).
It is moreover \emph{homogeneous}: \(F(cS) = c^{rz} F(S)\) for every \(c>0\).
By Theorem~3.6 in~\cite{rubboli-phd-2025} we have:
\begin{claim}\label{lemma:rubbuli}
  Let \(X\in \positiveOperators(A)\), \((r,z)\in \region_F\) and \(\sigma\in \free(A)\).
  Then, \(\sigma\) is an optimizer of \(F(S)=\tr \pb[\big]{{(S^{\frac r 2} X S^{\frac r 2})}^z}\) over the free states if and only if (i) \(\sigma\in \mathfrak{N}_r(X)\) and (ii)
\begin{equation}\label{eq:rubboli}
  \max_{\omega\in \free} \tr\pb*{\omega \Xi_\sigma} \le F(\sigma)\;.
\end{equation}
\end{claim}
Here,
\begin{align}\label{eq:def-region-N}
  \mathfrak{N}_r(X) & \coloneqq 
  \begin{cases}
    \pc*{\omega \in \states(A) \smallgiven \omega \gg X} , & r\in[-1,1)\setminus \{0\} \\ 
    \pc*{\omega \in \states(A) \smallgiven X^{1/2} \omega X^{1/2} \gg X}, & r =1\;,
  \end{cases}\\
  \label{eq:Xi}
  \Xi_\sigma &\coloneqq \begin{cases}
    \displaystyle \frac{{(\sigma^{ \frac r 2}X \sigma^{ \frac r 2})}^z}{\sigma}, &  r=\pm 1\;,\vspace{5pt} 
    \\
   \displaystyle \frac{\sin (\pi r)}{\pi r} \!\!\int_{0}^{\infty}\!\! \frac {{(\sigma^{ \frac r 2}X \sigma^{ \frac r 2})}^z}{\sigma^r {(\sigma + tI)}^{2}} t^r dt\;, & r \in (-1,1)\setminus \{0\},
  \end{cases}
\end{align}
where \(Y \gg Z\) means that the support of \(Z\) is contained in the support of \(Y\) and \(Y/Z\) denotes \(Z^{-1/2}YZ^{-1/2}\).
The next theorem characterizes the optimizer free states as the solution of an implicit equation.
\begin{theorem}\label{thm:implicit-form-for-optimizer}
  Let \(X\in \positiveOperators(A)\), \((r,z)\in \region_F\). Then, \(\sigma\in \free(A)\) is an optimizer of \(F(\omega) = \tr \pb[\big]{{(\omega^{\frac r 2} X \omega^{\frac r 2})}^z}\) over \(\free(A)\) if and only if (i) \(\sigma \in \mathfrak{N}_r(X)\) and (ii)
  \begin{equation} \label{eq:fixed-point-condition}
   F(\sigma) \sigma = 
   \Delta
   \pa*{\pa*{\sigma^{\frac r 2}X \sigma^{\frac r 2}}^z}\;.
  \end{equation}
\end{theorem}
\begin{remark}
  This condition is equivalent to the fixed-point condition \(\sigma \propto \Delta
   \pa*{\pa*{\sigma^{\frac r 2}X \sigma^{\frac r 2}}^z}\) (except in the trivial \(X=0\) case).
\end{remark}
  \begin{proof}
    We utilize the characterization of the optimizer from Claim~\ref{lemma:rubbuli}.
    Because any optimizer belongs to \(\mathfrak{N}_r(X)\), it is enough to take \(\sigma\in \mathfrak{N}_r(X)\) and prove that it is an optimizer if and only if it satisfies Eq.~(\ref{eq:fixed-point-condition}). 
    We prove for the case \(r\in (-1,1)\setminus\{0\}\), so \(\sigma \gg X\)
    \begin{equation}
      \Xi_\sigma = \frac{\sin (\pi r)}{\pi r} \!\!\int_{0}^{\infty}\!\! \frac {{(\sigma^{ \frac r 2}X \sigma^{ \frac r 2})}^z}{\sigma^r {(\sigma + tI)}^{2}} t^r dt\;.
    \end{equation}
The proof of the cases \(r = \pm 1\) is analogous with appropriate modifications.

We initially prove the theorem for the case in which \(\sigma\) has full support and \(\Delta\) is self adjoint, i.e., the case in which \(\Delta = \condExpectation \eqqcolon \tpExpectation\) is the trace-preserving conditional expectation onto a subalgebra \(\algebra \subset \linearOperators(A)\) and \(\free(A) = \algebra \cap \states(A)\).

    To use Eq.~(\ref{eq:rubboli}) from Claim~\ref{lemma:rubbuli}, note that 
   \begin{equation}
     \max_{\omega\in \free} \tr\pb*{\omega \Xi_\sigma} = \max_{\rho\in \states} \tr\pb*{\tpExpectation(\rho) \Xi_\sigma} = \max_{\rho\in \states} \tr\pb*{\rho \tpExpectation(\Xi_\sigma)}\;.
   \end{equation}
     The condition in Eq.~(\ref{eq:rubboli}) is that this expression is smaller than \(F(\sigma)\), i.e., that
\begin{equation} \label{eq:rubboli-operator-form}
  \tpExpectation(\Xi_\sigma) \le F(\sigma) I\;.
\end{equation}
Because \(\sigma\) has full support, Eq.~(\ref{eq:rubboli-operator-form}) is equivalent to
\begin{equation} \label{eq:rubboli-operator-form-2}
\sigma^{1/2}\tpExpectation(\Xi_\sigma)\sigma^{1/2} \le F(\sigma) \sigma\;.
\end{equation}

We can write
    \begin{equation}\label{eq:tp-Xi}
  \sigma^{1/2}\tpExpectation(\Xi_\sigma)\sigma^{1/2} = \sigma^{1/2} \tpExpectation\pa*{\frac{\sin (\pi r)}{\pi r} \!\!\int_{0}^{\infty}\!\! \frac {{(\sigma^{ \frac r 2}X \sigma^{ \frac r 2})}^z}{\sigma^r {(\sigma + tI)}^{2}} t^r dt}\sigma^{1/2}
       = \frac{\sin (\pi r)}{\pi r} \!\!\int_{0}^{\infty}\!\! \frac {\tpExpectation\pa*{{(\sigma^{ \frac r 2}X \sigma^{ \frac r 2})}^z}}{\sigma^{r-1} {(\sigma + tI)}^{2}} t^r dt \;,
    \end{equation} 
where we used the bimodularity of \(\tpExpectation\) with respect to operators in \(\algebra\) (Eq.~(\ref{eq:bimodularity})).
Let 
\begin{equation}
  \Psi(Y) \coloneqq \frac{\sin (\pi r)}{\pi r} \!\!\int_{0}^{\infty}\!\! \frac {Y t^r dt}{\sigma^{r-1} {(\sigma + tI)}^{2}} \;.
\end{equation}
By Eqs.~(\ref{eq:rubboli-operator-form-2}) and (\ref{eq:tp-Xi}), we finish the proof once we establish the equivalence: 
\begin{equation}\label{Eq:equivalence}
  \Psi\circ\tpExpectation({(\sigma^{ \frac r 2}X \sigma^{ \frac r 2})}^z) \le F(\sigma)\sigma \iff \tpExpectation({(\sigma^{ \frac r 2}X \sigma^{ \frac r 2})}^z) = F(\sigma) \sigma\;. 
\end{equation}
\begin{claim}[Auxiliary Claim]\label{Claim:aux-1}
  \(\Psi:\linearOperators(A) \to \linearOperators(A)\) is an invertible quantum channel and \(\Psi(\sigma) = \sigma\).
\end{claim}
The proof of this claim will be presented after we finish proving theorem.

Since \(\Psi\circ\tpExpectation\) is trace preserving, \(\tr [\Psi\circ\tpExpectation({(\sigma^{ \frac r 2}X \sigma^{ \frac r 2})}^z)] = 
\tr [{(\sigma^{ \frac r 2}X \sigma^{ \frac r 2})}^z] \defeq F(\sigma)\). 
Thus, the LHS condition in (\ref{Eq:equivalence}) is an inequality of positive semidefinite operators with the same trace, i.e., an equality.
Moreover, by Claim~\ref{Claim:aux-1}, \(\Psi(\sigma) = \sigma\), so the LHS condition is actually equivalent to
\begin{align}
  \Psi\circ\tpExpectation({(\sigma^{ \frac r 2}X \sigma^{ \frac r 2})}^z) & = F(\sigma)\Psi(\sigma) \\
  \explain{\(\Psi\) is linear}& = \Psi(F(\sigma)\sigma)\;.
\end{align}
But this equality is equivalent to the RHS condition because, by Claim~\ref{Claim:aux-1}, \(\Psi\) is injective.

 We now extend the proof beyond the self-adjoint case in which \(\Delta = \condExpectation\) using the following observation:
   \begin{observation}\label{observation:reduction-to-self-adjoint}
  Let  \(\algebra = \ima \condExpectation\) and note that the map \(\twisting(Y) = {\Delta(I)}^{1/2}Y{\Delta(I)}^{1/2}\) is a completely positive invertible map on \(\linearOperators(A)\) that maps \(\algebra \cap \states(A)\) onto \(\free(A)\).
  Moreover,
  \begin{itemize}
    \item \(\bar{\omega}^{r/2} \tilde{X} \bar{\omega}^{r/2} = \omega^{r/2} X \omega^{r/2}\) for all \(\omega\in \free(A)\) where we denote \(\tilde{X} \coloneqq {\Delta(I)}^{r/2}X{\Delta(I)}^{r/2}\) and \(\bar{\omega} \coloneqq {\Delta(I)}^{-\frac 12}\omega{\Delta(I)}^{-\frac 12}\in \states(A) \cap \algebra\) (since \(\omega\in \free(A)\) commutes with \(\Delta(I)\)).
    \item Therefore, \(F_X(\omega) = F_{\tilde{X}}(\bar{\omega})\).
    \item In particular, \(F_X(\sigma) = \opt F_X(\free)\)  if and only if \(F_X(\sigma)=F_{\tilde{X}}(\bar{\sigma}) = \opt F_{\tilde{X}}(\algebra\cap \states(A))\).
  \end{itemize}
\end{observation}

From the self-adjoint case, the last condition in the observation holds if and only if
\begin{align}\label{eq:proportionality-of-twisted-optimizer1}
  \explainabove{F_{\tilde{X}}(\bar{\sigma})}{=F_X(\sigma)} \ \bar{\sigma} & = \tpExpectation\pa*{\pa*{\bar{\sigma}^{ r / 2}\tilde{X} \bar{\sigma}^{ r /2}}^z} = \tpExpectation\pa*{\pa*{\sigma^{r/2}X \sigma^{r/2}}^z}\;.
\end{align}
The above chain happens if and only if
\begin{equation}
  F_X(\sigma)\sigma \defeq F_X(\sigma) \Delta(I)\bar{\sigma} = \Delta(I) \tpExpectation\pa*{\pa*{\sigma^{r/2} X \sigma^{r/2}}^z} = \Delta\pa*{\pa*{\sigma^{r/2} X \sigma^{r/2}}^z}\;,
\end{equation}
i.e., if and only if Eq.~(\ref{eq:fixed-point-condition}) holds.

Finally, the proof for the case in which \(\sigma\) does not have full rank follows by restricting the proof above to the support \(A_\sigma \coloneqq \ima \sigma \subset A\).
Indeed, all the operators that were involved are generated by \(\sigma\) and \(X\) (recall that \(\sigma\gg X\), i.e., \(X \in \linearOperators(A_\sigma)\)) with applications of \(\Delta\) and \(\condExpectation\) and \(\twisting^s(Y) \coloneqq {\Delta(I)}^{s/2}Y{\Delta(I)}^{s/2}\) (for relevant \(s\in \R\)).
To see that indeed the algebra \(\linearOperators(A_\sigma)\) is closed under applications of \(\mathcal{M}\) for \(\mathcal{M} =\Delta/\condExpectation/\twisting^s\), note that \(Y\in\linearOperators(A_\sigma)\) if and only if \(Y=\sigma^0 Y \sigma^0\).
Because \(\sigma\) is fixed by \(\Delta\), 
  \(\sigma^0\in \algebra\) (see, e.g., Eqs.~(\ref{eq:algebra-decomposition})--(\ref{eq:terminal-decomposition})), and hence, by bimodularity (Eqs.~(\ref{eq:bimodularity}),~(\ref{eq:bimodularity-schrödinger}) and since \(\Delta(I)\) commutes with \(\sigma\)) \(\mathcal{M}(Y) = \mathcal{M}(\sigma^0Y\sigma^0) = \sigma^0 \mathcal{M}(Y)\sigma^0\), which shows that \(\mathcal{M}(Y)\in \linearOperators(A_\sigma)\). 
  \end{proof}

\begin{proof}
  [Proof of Claim~\ref{Claim:aux-1}]
  \(\Psi\) is completely positive: it is an integral of maps \(Y \mapsto S^{1/2} Y S^{1/2}\) for \(S =  \sigma^{1-r}{(\sigma + tI)}^{-2}\) with positive weights.
  
  To prove the other properties, let us work in the basis \(\pc{\ket x}_x\) that diagonalizes \(\sigma\), i.e., in which \(\sigma = \sum_x p_x \ket {x}\! \bra {x}\) for \(\{p_x >0\}_x\).
    \begin{align}
    \Psi(\ket {x}\! \bra {y}) & = \frac{\sin (\pi r)}{\pi r} 
    \!\!\int_{0}^{\infty}\!\! \frac { \ket {x}\! \bra {y}t^r dt}{\sigma^{r-1}{(\sigma + tI)}^{2}} \\
    & = \frac{\sin (\pi r)}{\pi r} p_x^{\frac{1-r}{2}} p_y^{\frac{1-r}{2}} 
    \!\!\int_{0}^{\infty}\!\! \frac {t^r dt}{(p_x + t)(p_y + t)} \ket {x}\! \bra {y}
    \;.
   \end{align}
   Clearly, \(\Psi\) is diagonal in the basis \(\pc{\ket {x}\! \bra {y}}_{x,y}\) and moreover, its eigenvalues are all strictly positive: 
   \begin{equation}\label{eq:explicit-form}
    \lambda_{xy} \coloneqq \bra x \Psi(\ket {x}\! \bra {y}) \ket y =\frac{\sin (\pi r)}{\pi r} p_x^{\frac{1-r}{2}} p_y^{\frac{1-r}{2}} 
    \!\!\int_{0}^{\infty}\!\! \frac {t^r dt}{(p_x + t)(p_y + t)}>0\;.
   \end{equation}
   This proves that \(\Psi\) is invertible.
   Finally, to show that \(\Psi\) is trace preserving and that \(\Psi(\sigma)=\sigma\) let us 
   substitute \(y=x\) in Eq.~(\ref{eq:explicit-form}):
   \begin{align}
    \lambda_{xx} & = \frac{\sin (\pi r)}{\pi r} p_x^{1-r} 
    \!\!\int_{0}^{\infty}\!\! \frac {t^r dt}{(p_x + t)^2}
    \\ \label{eq:integral-computation}
    & = \frac{\sin (\pi r)}{\pi r} 
    p_x^{1-r} \frac{\pi r}{\sin (\pi r)} p_x^{r-1} \\
    & = 1\;,
   \end{align}
   where Eq.~(\ref{eq:integral-computation}) follows directly by integration. 
   With this, \(\Psi(\sigma) = \sum_x \lambda_{xx} p_x \ket {x}\! \bra {x}=\sum_x p_x \ket {x}\! \bra {x} = \sigma\).
   Similarly, trace preservation follows, since given \(Y = \sum_{xy} \alpha_{xy} \ket {x}\! \bra {y}\) we have \(\bra x \Psi (Y) \ket x = \lambda_{xx} \alpha_{xx} = \alpha_{xx}\) and thus
   \(\tr \circ \Psi (Y) = \sum_{x} \bra x \Psi(Y) \ket x = \sum_{x} \alpha_{xx} = \tr[Y]\).
\end{proof}

In the case \(z=1\), Eq.~(\ref{eq:functional}) simplifies and \(F_X(\sigma) = \tr [\sigma^r X]\). Because this functional is linear in \(X\), we are able to determine the optimal value explicitly.

\begin{lemma}[The case \(z=1\)]\label{functional:z=1}
  Let \(r\in [-1,1)\setminus\{0\}\) and \(X \in \pos(A)\). 
  The optimal value of 
\(\tr[\omega^r X]\) over \( \omega \in \free(A)\) is
\begin{equation}\label{eq:functional-petz-value}
\opt_{\omega\in \free} \tr[\omega^r X]
    = \norm*{\condExpectation
       \pa*{\frac{X}{{\Delta(I)}^p}}
       }_{\frac{1}{p}}\;,
  \end{equation}
 for \(p=1-r\).
 Moreover, there exists an optimizer state \(\sigma\in \free(A)\) such that for any \(\omega\in \free(A)\) we have:
  \begin{equation}\label{eq:pythagorean-identity}
    \tr[\omega^r X] = \tr[\sigma^r X]\tr[\omega^r \sigma^{1-r}]\;.
  \end{equation}
\end{lemma}
\begin{proof} The case \(X=0\) is trivial, henceforth we assume \(X \ne 0\).

  \textbf{Step 1: self-adjoint case}.
  Consider the case in which \(\Delta=\condExpectation\eqcolon\tpExpectation\) 
  is the trace-preserving conditional expectation onto an algebra \(\algebra\).
  In this case, Eq.~(\ref{eq:functional-petz-value}) simplifies and we need to prove:
   \begin{align}
    \label{eq:petz-form-self-adjoint}
    \opt_{\omega\in \free} \tr[\omega^r X] &= \norm*{\tpExpectation(X)}_{\frac{1}{p}}\;. 
   \end{align}
   Let \(Y \coloneqq {\tpExpectation(X)}^{\frac{1}{1-r}}\) and denote \(\displaystyle \sigma = Y/\tr[Y]\).
  For \(\omega\in\states(A) \cap \algebra\), \(\omega^r \in \algebra\), and hence 
  \begin{align}\label{eq:with-X}
    \tr [\omega^r X] & =  \tr [\omega^r \tpTerminal{X}] = \tr [\omega^r Y^{1-r}] = \tr[\omega^r {(\tr[Y]\sigma)}^{1-r}] = {\tr[Y]}^{1-r} \tr[\omega^r \sigma^{1-r}]\;.
  \end{align}
  By Hölder inequality (for \(r>0\), and reversed Hölder inequality for \(r<0\)), \(\omega \mapsto \tr[\omega^r \sigma^{1-r}]\) is optimized at \(\omega=\sigma\) (where it achieves the value 1).
  This proves Eq.~(\ref{eq:pythagorean-identity}) and, moreover, that
  \begin{equation}
    \opt_{\omega\in \free} \tr[\omega^r X] = \tr[\sigma^r X] = {\tr[Y]}^{1-r} = {\tr[{\tpExpectation(X)}^{\frac{1}{1-r}}]}^{1-r}\;,
  \end{equation}
which establishes Eq.~(\ref{eq:petz-form-self-adjoint}).
  
  \textbf{Step 2: Generalization to arbitrary \(\Delta\)}.
Let \(\tpExpectation\) be the trace-preserving conditional expectation onto \(\algebra \coloneqq \ima \condExpectation\).
Using Observation~\ref{observation:reduction-to-self-adjoint} and the self-adjoint case we conclude:
\begin{enumerate}[label=(\roman*)] 
  \item \(\opt_{\omega\in \free} \tr[\omega^r X] = \opt_{\bar \omega\in \algebra \cap \states} \tr[\bar {\omega}^r \tilde X]\) for \(\tilde{X} = {\Delta(I)}^{r/2} X {\Delta(I)}^{r/2}\). By Eq.~(\ref{eq:petz-form-self-adjoint}), that is \(\norm*{\tpExpectation(\tilde X)}_{\frac{1}{p}}\). However, by Eq.~(\ref{eq:pre-twisting2}) and the definition of \(\tilde X\), \(\tpExpectation(\tilde X) = \condExpectation({\Delta(I)}^{-\frac{p}{2}}X{\Delta(I)}^{-\frac{p}{2}})\), proving Eq.~(\ref{eq:functional-petz-value}).
  \item Let \(\bar{\sigma}\) be the optimizer of \(\opt_{\bar \omega\in \algebra \cap \states} \tr[\bar {\omega}^r \tilde X]\) as defined in the self-adjoint case and let \(\sigma={\Delta(I)}^{1/2}\bar \sigma {\Delta(I)}^{1/2}\).
   Given \(\omega \in \free(A)\), denote \(\bar \omega = {\Delta(I)}^{-1/2} \omega {\Delta(I)}^{-1/2} \in \algebra \cap \states(A)\). 
   
  From the self-adjoint step we have \(\tr[\bar {\omega}^r \tilde X] = \tr[\bar {\sigma}^r \tilde X]\tr[\bar {\omega}^r {\bar \sigma}^{1-r}]\).
  However, by Observation~\ref{observation:reduction-to-self-adjoint}, \(\tr[\bar {\omega}^r \tilde X] = \tr[\omega^r X]\) and \(\tr[\bar {\sigma}^r \tilde X] = \tr[\sigma^r X]\). 
  Thus, we only need to prove that \(\tr[\bar{\sigma}^{1-r} \bar \omega^r]=\tr[\sigma^{1-r} \omega^r]\). However, this follows since \(\sigma^{1-r} \omega^r = \Delta(I)\bar{\sigma}^{1-r} \bar \omega^r\), and hence \(\tr[\sigma^{1-r} \omega^r] = \tr[I \condExpectation(\bar{\sigma}^{1-r} \bar \omega^r)] \) but \(\bar{\sigma}^{1-r} \bar \omega^r \in \algebra\), so  \(\condExpectation(\bar{\sigma}^{1-r} \bar \omega^r) = \bar{\sigma}^{1-r} \bar \omega^r\). \qedhere
\end{enumerate}

\end{proof}
\subsection{Rényi Monotones}\label{appendix:renyi}
Substituting \(r=1-\alpha\) and \(X = \rho^\alpha\) in Lemma~\ref{functional:z=1} yields:
\begin{corollary}\label{corollary:petz-renyi}
  Let \(\rho \in \states(A)\), \(\sigma\in \free(A)\) and \(\alpha \in (0,2]\setminus\{1\}\), then
  \begin{equation}
    \bar D_\alpha \divergence{\rho}{\sigma} = \bar D_\alpha \divergence{\rho}{\optimizer} + \bar D_\alpha \divergence{\optimizer}{\sigma}
  \end{equation}
  where \(\optimizer\) is a minimizing state
for which \(\displaystyle
    \bar D_\alpha \divergence{\rho}{\optimizer} = \frac{1}{\alpha -1} \log \norm*{\condExpectation\pa*{\frac{\rho^\alpha}{{\Delta(I)}^\alpha}}}
  _{\frac{1}{\alpha}}
  \).
\end{corollary}
\begin{remark}
  The limit cases \(\alpha=0\) (min-relative entropy) and \(\alpha=1\) (Umegaki relative entropy) are discussed in the main text.
\end{remark}

\begin{lemma}\label{thm:optimized-renyi}
  Let \({\alpha,z}\) be parameters in the DPI region \(\region\) (Eq.~(\ref{dpi-domain})).
  Then \(D_{\alpha,z} \divergence{\rho}{\free}\) is additive.
\end{lemma}
\begin{proof}
   Setting \(r = (1-\alpha)/z\) yields \((r,z)\in \region_F\).
   We can therefore apply Theorem~\ref{thm:implicit-form-for-optimizer} substituting \(X=\rho^{\alpha/z}\) which yields \(Q_{\alpha,z} \divergence{\rho}{\sigma} = F_X(\sigma)\) (definitions in Eqs.~(\ref{eq:Q-alpha_z}) and (\ref{eq:functional})).
  
  Our goal is to prove that \(\sigma\otimes\sigma'\) is the optimizer, namely,
  \[
    Q_{\alpha,z}\divergence{\rho\otimes\rho'}{\sigma\otimes\sigma'} = Q_{\alpha,z}\divergence{\rho\otimes\rho'}{\free},
  \]
  assuming \(\sigma\) optimizes \(Q_{\alpha,z}\divergence{\rho}{\sigma} = Q_{\alpha,z}\divergence{\rho}{\free}\) and \(Q_{\alpha,z}\divergence{\rho'}{\sigma'} = Q_{\alpha,z}\divergence{\rho'}{\free}\).
With Theorem~\ref{thm:implicit-form-for-optimizer}, the assumptions mean that \(\sigma \propto \Delta\pa[\big]{\pa[\big]{\sigma^{\frac r 2} \rho^{\frac \alpha z}\sigma^{\frac r 2}}^z}\), \(\Delta\pa*{\sigma'} \propto \Delta\pa[\big]{\pa[\big]{{\sigma'}^{\frac r 2} {\rho'}^{\frac \alpha{z}}{\sigma'}^{\frac r 2}}^z}\) and \(\sigma\in \mathfrak{N}_r(\rho)\), \(\sigma' \in \mathfrak{N}_r(\rho')\) (see definition in Eq.~(\ref{eq:def-region-N})).
   Moreover, we only need to prove that \(\sigma \otimes \sigma' \in \mathfrak{N}_r(\rho \otimes \rho')\) (which follows straightforwardly by the definition of \(\mathfrak{N}_r\)) and that
    \begin{equation}
      \!\sigma\otimes\sigma' \propto \Delta\pa*{{\pa*{{(\sigma\otimes\sigma')}^{\frac r 2} {(\rho\otimes\rho')}^{\frac\alpha{z}}{(\sigma\otimes\sigma')}^{\frac r 2}}}^z}\,.
    \end{equation}
However, this equation follows by locality of destruction and since operator powers factorize over tensor products (\({(X\otimes Y)}^p = X^p\otimes Y^p\)).
\end{proof}
We are ready to prove \mtthm{thm:3-renyi}.
For \(\alpha,z\in\region\), \(\lambda\in[0,1]\) and \(\rho\in\states\), set \(Y_\rho = \rho^{\frac{\alpha}{2z}}{\Delta(\rho)}^{\frac{\lambda(1-\alpha)}{2z}}\) and
\begin{equation}
\monotone_{\alpha,z}^\lambda(\rho) \coloneqq \inf_{\sigma\in\free}\frac{1}{\alpha-1}\log \tr\pb[\big]{\pa[\big]{Y_\rho \sigma^{(1-\lambda)\frac{1-\alpha}{z}}Y_\rho^*}^z}\;,
\end{equation}
where the trace functional is read as \(+\infty\) when \(\alpha>1\) and \(Y_\rho^*Y_\rho \not \ll \sigma\). 
\ThmRenyi*
For \(\lambda=0\) we have \(Y_\rho^*Y_\rho = \rho^{\alpha/z}\), so the support condition reads \(\sigma\gg\rho\) and the definition reduces to \(D_{\alpha,z} \divergence{\rho}{\free}\).
For \(\lambda=1\) it reduces to \(D_{\alpha,z} \divergence{\rho}{\Delta(\rho)}\).
Indeed, \(\supp \pa*{Y_\rho^* Y_\rho} = \ima Y_\rho^*\), so for \(\alpha>1\) every candidate of finite value obeys \(\supportProjection{\sigma} Y_\rho^* = Y_\rho^*\) and hence \(Y_\rho \supportProjection{\sigma} Y_\rho^* = Y_\rho Y_\rho^*\); for \(\alpha<1\) we have \(\supportProjection{\sigma}\le I\), so \(\tr\pb[\big]{{(Y_\rho \supportProjection{\sigma} Y_\rho^*)}^z} \le \tr\pb[\big]{{(Y_\rho Y_\rho^*)}^z}\), and since \(\frac 1{\alpha-1}<0\) the infimum is attained at any full-rank free state \(\sigma\).
In both cases the value is \(\frac 1 {\alpha -1}\log \tr\pb[\big]{{(Y_\rho Y_\rho^*)}^z} = D_{\alpha,z} \divergence{\rho}{\Delta(\rho)}\).
The latter is an additive monotone as the ``direct'' extension of the additive divergence \(D_{\alpha,z}\) and the former is a monotone that was proved additive in Lemma~\ref{thm:optimized-renyi}.
We therefore henceforth assume \(\lambda\in(0,1)\).
  Moreover, a similar analysis as in the proof of Lemma~\ref{thm:optimized-renyi} shows that \(\monotone_{\alpha,z}^\lambda\) is additive (this time, substituting \(r=(1-\lambda)\frac{1-\alpha} z\) and \(X=Y_\rho^*Y_\rho\)).
  Thus, we only need to establish monotonicity under destruction-covariant channels.

  To do this, we establish an asymptotic form for \(\monotone_{\alpha,z}^\lambda\).
\begin{lemma}[Asymptotic form]\label{appendix:asymptotic-form}
  For \(\rho\in \states(A)\), we have:
\begin{align}
  \monotone_{\alpha,z}^\lambda(\rho) 
  = \lim_{n\to \infty} \frac 1 n \inf_{ \sigma\in \free \cap \sym}\!\!\!\! D_{\alpha,z} \divergence{\rho^{\otimes n}}{\sigma\#_\lambda \Delta\pa*{\rho}^{\otimes n}} .
\end{align}
\end{lemma}
Here, \(\#_\lambda\) is the operator geometric mean (see Eq.~(\ref{eq:def-geometric-mean})) and \(\sym\) denotes the family of states invariant under permutations of the copies of \(A\) in the composite system \(A^n\).
These states are those fixed by the twirling channel \(\mathcal{G}_n\), that averages the standard action of the symmetric group on \(\linearOperators(A^n)\) (we minimize over \(\sigma\in \free(A^n) \cap \sym(A^n)\)).
Note that \(\mathcal{G}_n \circ\Delta = \Delta \circ \mathcal{G}_n\), by locality of \(\Delta\).

Before we prove Lemma~\ref{appendix:asymptotic-form}, let us see how it establishes the monotonicity of \(\monotone_{\alpha,z}^\lambda\).
\begin{proof}[Proof of \mtthm{thm:3-renyi}]
  Additivity has already been established. 
  To prove monotonicity, let \(\mathcal{N}\) be \(\Delta\)-covariant. Note that \(\mathcal{N}^{\otimes n}\) is \(\Delta \circ \mathcal{G}_n\)-covariant.
Hence, we may apply Claim~\ref{claim:mean-optimized} with the effective destruction \(\Delta \circ \mathcal{G}_n\) and \(\mean = \#_\lambda\) (see definition in Eq.~(\ref{eq:def-geometric-mean})), obtaining
\begin{align}
  \inf_{ \sigma\in \free \cap \sym}\!\!\! D_{\alpha,z} \divergence{\rho^{\otimes n}}{\sigma\#_\lambda \Delta\pa*{\rho}^{\otimes n}} & \ge \inf_{ \sigma\in \free \cap \sym}\!\!\! D_{\alpha,z} \divergence{\mathcal{N}^{\otimes n}(\rho^{\otimes n})}{\sigma\#_\lambda \mathcal{N}^{\otimes n}\pa[\big]{\Delta\pa*{\rho}^{\otimes n}}} \\
  & = \inf_{ \sigma\in \free \cap \sym}\!\!\! D_{\alpha,z} \divergence{{\mathcal{N}(\rho)}^{\otimes n}}{\sigma\#_\lambda {\mathcal{N}(\Delta\pa*{\rho})}^{\otimes n}}\;.
\end{align}
Dividing by \(n\) and letting \(n\to\infty\) while using Lemma~\ref{appendix:asymptotic-form} establishes monotonicity of \(\monotone_{\alpha,z}^\lambda\).
\end{proof}
\begin{proof}[Proof of Lemma~\ref{appendix:asymptotic-form}]
Let us denote
\begin{equation}\label{eq:3-relative}
  f_\rho(S) = \frac 1 {\alpha-1} \log F_{X}(S)\;,
\end{equation}
for \(X = Y_\rho^*Y_\rho\) and \(r = (1-\lambda)\frac{1-\alpha}z\), so that \(
  \monotone_{\alpha,z}^\lambda(\rho) = \inf_{\sigma\in \free}
  f_\rho(\sigma) = f_\rho(\optimizer)
  \),
where \(\optimizer\) is the optimizing free state.
Invoke the universal permutation-invariant state \(\hat \omega_n\in \states(A^n)\) from Ref.~\cite{HT2016}, which
\begin{itemize}
  \item Commutes with every permutation-invariant operator.
  \item Satisfies \(g_n \hat \omega_n \ge \sigma\) for every permutation-invariant state \(\sigma\). Here, \(g_n \coloneqq \rank \mathcal{G}_n\) satisfies \( \log g_n = o(n)\).
\end{itemize}
The asymptotic form follows from the chain:
\begin{equation}
   n \monotone_{\alpha,z}^\lambda(\rho) + o(n)  \overset{\ref{item:asymptotic-chain-1}}{=} f_{\rho^{\otimes n}}(\Delta\pa*{\hat \omega_n}) \overset{\ref{item:asymptotic-chain-2}}{=} D_{\alpha,z} \divergence{\rho^{\otimes n}}{\Delta\pa*{\hat \omega_n} \#_\lambda \Delta\pa*{\rho}^{\otimes n}}  \overset{\ref{item:asymptotic-chain-3}}{=}\!\!\!\!\inf_{ \sigma\in \free \cap \sym}\!\!\! D_{\alpha,z} \divergence{\rho^{\otimes n}}{\sigma\#_\lambda \Delta\pa*{\rho}^{\otimes n}} +o(n)\;.  
\end{equation}

\begin{enumerate}[label=(\alph*)] 
  \item\label{item:asymptotic-chain-1} (\(\le\)) By additivity: \(\optimizer^{\otimes n}\)
   is the optimizing free state of \(f_{\rho^{\otimes n}}\) with value \(n \monotone_{\alpha,z}^\lambda (\rho)\), and \(\Delta\pa*{\hat \omega_n}\) is a free state: \(n \monotone_{\alpha,z}^\lambda(\rho)=f_{\rho^{\otimes n}}(\optimizer^{\otimes n}) \le f_{\rho^{\otimes n}}(\Delta\pa*{\hat \omega_n})\).

  (\(\ge\)) since \(\optimizer^{\otimes n}\le g_n \hat \omega_n\) and \(\Delta\) is positive, \(\optimizer^{\otimes n} =\Delta(\optimizer^{\otimes n}) \le g_n \Delta(\hat \omega_n) \).
   By the properties of \(F_X\) recorded below its definition in Eq.~(\ref{eq:functional}), \(f_{\rho^{\otimes n}}(c\,S) = f_{\rho^{\otimes n}}(S) - (1-\lambda)\log c\) for every \(c>0\) and \(f_{\rho^{\otimes n}}\) is Löwner antitone for either sign of \(\alpha-1\).
   Thus \(f_{\rho^{\otimes n}}(\optimizer ^{\otimes n}) \ge f_{\rho^{\otimes n}}\pa*{g_n \Delta\pa*{\hat \omega_n}} = f_{\rho^{\otimes n}}(\Delta\pa*{\hat \omega_n}) - (1-\lambda) \log (g_n)\), which concludes this direction since \((1-\lambda) \log (g_n) = o(n)\).

  \item\label{item:asymptotic-chain-2}
  Follows by the definitions of \(f\), \(D_{\alpha,z}\), and \(\#_\lambda\) [Eq.~(\ref{eq:3-relative}), \mteq{eq:alpha-z-renyi}, and Eq.~(\ref{eq:def-geometric-mean}), respectively], and from the fact that \(\Delta\pa*{\hat \omega_n}\) commutes with \(\Delta\pa*{\rho}^{\otimes n}\) (see Claim~\ref{claim:universal-state} below).
  \item\label{item:asymptotic-chain-3} Analogous to the proof of item~\ref{item:asymptotic-chain-1}.
\end{enumerate}
\end{proof}
\begin{claim}\label{claim:universal-state}
  Any permutation-invariant free state \(\sigma\in \free(A^n)\) commutes with \(\Delta(\hat \omega_n)\).
\end{claim}
\begin{proof}
 Let \(\tpExpectation\) be the trace-preserving conditional expectation onto \(\algebra= \ima \condExpectation\).
  Denote \(\bar{\sigma} \coloneqq {\Delta(I)}^{-1}\sigma\) and note that this state lies within \(\algebra\). Moreover, \(\Delta(\hat{\omega}_n) = \Delta(I)\tpExpectation(\hat{\omega}_n)\).
  Thus,
\begin{equation}
  \sigma \Delta\pa*{\hat \omega_n} = \Delta(I)^2\bar{\sigma}\tpExpectation(\hat{\omega}_n) \explainabove{=}{\text{\!\!\!\!\!\!bimodularity\!\!\!\!\!\!}} {\Delta(I)}^2 \tpExpectation(\bar{\sigma} \hat \omega_n)\;.
\end{equation} 
Similarly,
\(  \Delta\pa*{\hat \omega_n}\sigma = {\Delta(I)}^2 \tpExpectation(\hat \omega_n \bar{\sigma})\).
However,
 \(\bar{\sigma} = {\Delta(I^{A^n})}^{-1}\sigma ={\pa[\big]{{\Delta(I^{A})}^{-1}}}^{\otimes n} \sigma\) is permutation invariant, hence \(\bar{\sigma} \hat \omega_n = \hat \omega_n \bar{\sigma} \), thus concluding
\(\sigma \Delta\pa*{\hat \omega_n} = \Delta\pa*{\hat \omega_n}\sigma\). 
\end{proof}
\subsection{Extremal Weakly Additive Monotones}\label{appendix:extremals}
In this section we prove \mtthm{thm:extremal}, which establishes that \(\dmin \divergence{\rho}{\free}\) and \(\dmax \divergence{\rho}{\Delta(\rho)}\) are the minimal and maximal among all normalized weakly additive monotones.
\begin{remark}
  Both extremes are additive and normalized.
  Additivity of \(\dmin \divergence{\cdot}{\free}\) is immediate from \mteq{eq:dmin-dual-final}, since support projections, \(\Delta^{\!*}\) and \(\norm{\cdot}_\infty\) all factorize on tensor products; \(\dmax \divergence{\cdot}{\Delta(\cdot)}\) is additive because \(\Delta\) is local and \(\dmax = \tilde D_\infty\) is additive~\cite{Datta2009,MDS+2013}.
  Normalization follows from the definitions of the min- and max-relative entropies (Eqs.~(\ref{dmin-heisenberg}) and~(\ref{eq:def-dmax}) in the main text) evaluated on \(\currency[1]\): the  state \(\ket 0\! \bra 0\) in a two-dimensional thermodynamic system with Gibbs state \(I/2\)~\cite{GT2020}.
\end{remark}
\ThmExtremal*
\begin{remark}
  This theorem does not conflict with the fact that weakly additive conditional \emph{uncertainty measures} can exceed \(H_{\min}\)~\cite{GWBG2024}: the correspondence between conditional entropies and additive conditional-nonuniformity monotones requires strong additivity~\cite{JGW2025}.
\end{remark}
\begin{proof}
  Claim~\ref{claim:weakly-additive-monotones-normalization} establishes \(\yield[0]_\infty\) and \(\cost[0]_\infty\) as the minimal and maximal normalized weakly additive monotones in every resource theory.
  
  \(\cost[0](\rho) = \dmax \divergence{\rho}{\Delta(\rho)}\) by Eq.~(\ref{eq:cost}) in the main text. But \(\dmax \divergence{\rho}{\Delta(\rho)}\) is additive as the direct extension of the max-relative entropy, so \(\cost[0]_\infty(\rho) = \dmax \divergence{\rho}{\Delta(\rho)}\).
  The other side is item~\ref{item:asymptotic-yield-convergence} of Lemma~\ref{thm:unrestricted-HT-yield-connection} (the case \(\epsilon=0\)).
\end{proof}
\begin{claim}[Weak additivity is strictly weak]
  \label{claim:weakly-additive-strict}
  Let \(a(\rho) \coloneqq \dmin \divergence{\rho}{\free}\), \(b(\rho)\coloneqq \dmax \divergence{\rho}{\Delta(\rho)}\), and \(\monotone_{\rm{geo}} \coloneqq \sqrt{a\,b}\).
  Then \(\monotone_{\rm{geo}}\) is a normalized weakly additive monotone that is not additive; in particular, the weakly additive class is \emph{strictly} larger than the additive one.
\end{claim}
\begin{proof}
  Being the geometric mean of two nonnegative monotones, \(\monotone_{\text{geo}}\) is a monotone.
  By the remark opening this section \(a\) and \(b\) are additive and equal \(1\) on \(\currency[1]\), so \(\monotone_{\text{geo}}(\currency[1]) = 1\) and \(\monotone_{\text{geo}}\pa*{\rho^{\otimes n}} = \sqrt{n a(\rho) \cdot n b(\rho)} = n \monotone_{\text{geo}}(\rho)\).
  
  \(\monotone_{\text{geo}}\) is not additive, since \(\sqrt{a_1b_1} + \sqrt{a_2 b_2} \le \sqrt{(a_1+a_2)(b_1+b_2)}\) with equality only when \(a_1 b_2 = a_2 b_1\):
  for two thermodynamic bits with Gibbs state \(I/2\), \(\rho_1 = \ket 0 \!\bra 0 = \currency[1]\) and \(\rho_2 = \frac 2 3 \ket 0 \!\bra 0 + \frac 1 3 \ket 1\! \bra 1\), we have \((a_1,b_1) = (1,1)\) and \((a_2,b_2) = \pa*{0,\log \frac 4 3}\), so \(\monotone_{\text{geo}}\pa*{\rho_1 \otimes \rho_2} = \sqrt{3 - \log 3}\) while \(\monotone_{\text{geo}}(\rho_1) + \monotone_{\text{geo}}(\rho_2) = 1\).
\end{proof}

\section{The Currency of Instability}\label{appendix:numeraire}
Here \(\rho\sim\sigma\) denotes resource equivalence---\(\rho \convto[0]\sigma\convto[0]\rho\).

A \emph{currency} (also called a \emph{golden unit} or a \emph{family of reference states}) is a parametrized family of resource states \(\pc{\phi_m}_m\) that is \emph{additive} (\(\phi_m \otimes \phi_t \sim \phi_{m+t}\)) and \emph{cofinal} (for every \(\rho \in \states\) there exists \(m\) such that \(\phi_m \convto[0]\rho\)).
The parameter regime can be any additive sub-semigroup of the nonnegative reals.

In this work we use the thermodynamic currency \(\currency[m]\) discussed in the main text and in~\cite{WW2019}, for which \(\cost(\rho) = \dmax \divergence{\rho}{\Delta(\rho)}\) for all \(\rho \in \states\).
Because \(\Delta^{\!A*}\) is a faithful conditional expectation, the support of \(\Delta(\rho)\) always contains the support of \(\rho\), thus \(\dmax \divergence{\rho}{\Delta(\rho)}\) is always finite (i.e., \(\pc{\currency[m]}_m\) is cofinal); the supremum over \(\rho\in\states(A)\) is the log of the ``generalized Pimsner-Popa index'' of the conditional expectation \(\Delta^{\!A*}\)~\cite{GR2022}.

Moreover, by definition, \(\yield(\rho) \le \cost(\rho)\).
If equality holds, then \( \rho \sim \currency[m]\) for \(m = \yield(\rho) = \cost(\rho)\) (the supremum and infimum in the main-text definitions of the yield and cost are always attained).
The states satisfying \(\yield(\rho) = \cost(\rho)\) are totally ordered under \(\convto[0]\) and obey a second law already at the single-copy level.
For the following states \(\rho\), we can establish this equality using \mtlem{thm:one-shot} with \(\epsilon=0\), by direct calculations of the restricted hypothesis-testing quantity \(h^\epsilon(\rho)\) and of \(\dmax \divergence{\rho}{\Delta(\rho)}\). 
\begin{enumerate}[label=(\arabic*)] 
  \item Any \(\sigma\in \free\) satisfies \(\sigma\sim \currency[0]\).
  \item \(\currency[m]\otimes \currency[m'] \sim \currency[m+m']\).
  \item\label{item:Phi}
  Let \(A\) and \(B\) be systems with the same dimension \(d\) and suppose \(\Delta^{\!AB}(X) = \u^A \otimes \tr_A X\) (depolarizing \(A\); this is the destruction mechanism in the resource theory of conditional nonuniformity).
The maximally entangled state  \(\ket{\Phi}^{AB} = \frac{1}{\sqrt{d}} \sum \ket x_A \ket x_B\) satisfies \(\ket{\Phi}_{d\times d} \sim \currency[2 \log d]\).
  \item Similarly, the \(d\)-dimensional fully coherent state \( \ket + _d \coloneq \frac{1}{\sqrt{d}} \sum_x \ket x \) is equivalent to \(\currency[\log d]\).
\end{enumerate}
The last two examples provide concrete currency ``alternatives'' for which the one-shot yield and cost formulas from the main text also apply (up to the continuous parameterization in \(\currency[m]\)). 
This recovers known results in the literature of coherence and athermality~\cite{RFXA2018,LBT2019,WW2019,JGW2025}. 

Since \(\yield(\rho)\le \dmin\divergence{\rho}{\free}\), Eq.~(\ref{eq:optimal-entropies}) forces every normalized weakly additive instability monotone to take the same value on this class; for \(\ket\Phi^{AB}\) (item~\ref{item:Phi} above), this recovers the main statement of~\cite{GWBG2024} (and extending this argument to the dynamical regime recovers the main statement of~\cite{GKN+2025}).

\begin{claim}\label{claim:weakly-additive-monotones-normalization}
  Let \(\monotone\) be a weakly additive monotone that is normalized \((\monotone(\currency[1]) = 1)\). Then
  \begin{enumerate}
    \item \(\monotone(\currency[t]) = t\) for all \(t\in (0,\infty)\).
    \item \(\yield[0]_\infty \le \monotone \le \cost[0]_\infty\).
  \end{enumerate} 
\end{claim}
\begin{remark}
  This claim holds for any quantum resource theory and any currency if the monotones are normalized with respect to it.
  An essentially similar proof to item 1 is provided in~\cite{GT2021} (Lemma~3), for example.
\end{remark}
\begin{remark}
  These regularizations exist by Fekete's lemma. Indeed, \(\yield[0]\) is superadditive (\(\yield(\rho\otimes\sigma) \ge \yield(\rho) + \yield(\sigma)\)) and \(\cost[0]\) is subadditive, by the additivity of \(\currency[m]\) and their definitions. As regularizations, they themselves are weakly additive, and they are normalized by their definition and the additivity of the currency. The second item therefore shows that they are the extremal normalized weakly additive monotones. 
\end{remark}
\begin{proof}
  1. By additivity of the currency (item~(2) above the claim), we have \({\currency[\frac{m}{n}]}^{\otimes n} \sim \currency[m]\sim {\currency[1]}^{\otimes m}\) for all \(m,n \in \N\).
By weak additivity and normalization of \(\monotone\) we therefore infer \(
  n\, \monotone(\currency[\frac{m}{n}])=\monotone({\currency[\frac{m}{n}]}^{\otimes n}) = m\, \monotone(\currency[1]) = m
\), proving the claim for all positive rational numbers.
 The extension to \(t \in (0,\infty)\) follows from monotonicity of \(\monotone\).

 2. Let \(\rho\in \states\). If \(\rho\convto[0]\currency[m]\), then, by monotonicity and item 1, \(\monotone(\rho) \ge \monotone(\currency[m]) = m\). 
 Since this inequality holds for all such \(m\), it holds also for their upper bound, which is, by definition \(\yield[0](\rho)\).
 The relation \(\monotone \ge \yield[0]\) proceeds through regularization of both measures so \(\monotone_\infty \ge \yield[0]_\infty\), but \(\monotone_\infty = \monotone\) by its weak additivity.

 The second inequality is proved analogously.
\end{proof}

\section{The Dilution Cost}
This section analyzes the dilution task. 
In Sec.~\ref{appendix:cost-derivation} we calculate the one-shot instability cost and in Sec.~\ref{appendix:cost-bounds} we bound it via the smoothed max-relative entropy and use the generalized quantum Stein's lemma to conclude that \(\cost[\epsilon](\rho)\) regularizes asymptotically to the Umegaki monotone \(D \divergence{\rho}{\Delta(\rho)}\).

\subsection{Calculation of the One-Shot Dilution Cost}\label{appendix:cost-derivation}
\emph{The exact one-shot instability cost} of a state \(\rho\) is the infimum over \(m\ge 0\) for which \(\currency[m] \convto[0] \rho\), where \(\currency[m]\) is described by the state \(\ket 0\! \bra 0\) on a thermodynamic (classical, two-level) system \(B\) with Gibbs state
\begin{equation}
  \gamma = 2^{-m}\ket 0\! \bra 0 + (1-2^{-m})\ket 1\! \bra 1\;.
\end{equation}
We seek a free preparation channel \(\mathcal{P}\in \free(B\to A)\) with \( \mathcal{P}(\ket 0\! \bra 0) = \rho\).
Because \(B\) is thermodynamic, \(\Delta^{\!B}\) is the complete thermalization \(X \mapsto \tr[X]\gamma\); in particular \(\Delta^{\!B}(\ket 0\! \bra 0) = \gamma\).
Applying the covariance relation \(\Delta^{\!A}\circ \mathcal{P} = \mathcal{P}\circ \Delta^{\!B}\) to \(\ket 0\! \bra 0\) and using the linearity of \(\mathcal{P}\), we therefore obtain
\begin{align}
  \Delta\pa*{\rho} = \Delta^{\!A}\pa[\big]{\mathcal{P}(\ket 0\! \bra 0)}
  &= \mathcal{P}\pa[\big]{\Delta^{\!B}(\ket 0\! \bra 0)} = \mathcal{P}\pa*{\gamma} \nonumber
  \\
  &= 2^{-m} \mathcal{P}(\ket 0\! \bra 0) + (1-2^{-m})\mathcal{P}(\ket 1\! \bra 1)\;.
\end{align}
For \(m>0\), substituting \( \mathcal{P}(\ket 0\! \bra 0) = \rho\) again and isolating \(\mathcal{P}(\ket 1\! \bra 1)\) yields
\begin{equation}\label{eq:cost-off-branch}
  \mathcal{P}(\ket 1\! \bra 1) = \frac{2^m \Delta\pa*{\rho}-\rho}{2^m -1}\;.
\end{equation}
Thus, for \(m>0\), such a preparation channel is feasible if and only if \(2^m\Delta\pa*{\rho}-\rho \ge 0\).
At the remaining endpoint \(m=0\) the currency is itself free, \(\currency[0] = \ket 0\! \bra 0 = \gamma\), so \(\currency[0]\convto[0]\rho\) holds if and only if \(\rho\in \free(A)\): a free channel maps free states to free states, and conversely \(\mathcal{P}(X)\coloneqq \tr[X]\rho\) is destruction covariant whenever \(\Delta(\rho)=\rho\).
Since \(\dmax\divergence{\rho}{\Delta(\rho)}=0\) if and only if \(\rho = \Delta(\rho)\), the two cases combine: the minimal \(m\) with the above property is precisely the max-relative entropy between \(\rho\) and \(\Delta(\rho)\), concluding:
\begin{equation}\label{eq:exact-cost}
  \cost(\rho) \defeq \inf\pc*{m \smallgiven  \currency[m]\! \convto[0]\! \rho}
  = \dmax\divergence{\rho}{\Delta\pa*{\rho}}\,.
\end{equation}

The \(\epsilon\)-one-shot instability cost follows immediately by minimizing the cost over states \(\tau\) in the \(\epsilon\)-neighborhood of \(\rho\):
\begin{equation*}
    \cost[\epsilon](\rho) \defeq \inf\pc*{m \smallgiven  \currency[m]\! \convto[\epsilon]\! \rho}
    =  \min_{\tau \approx_\epsilon \rho}\dmax\divergence{\tau}{\Delta\pa*{\tau}}\,.
\end{equation*}\vspace{-12pt}\\
The minimum is attained because \(\pc{\tau\in\states \smallgiven \tau\approx_\epsilon\rho}\) is compact and the map
\(\tau\mapsto \dmax\divergence{\tau}{\Delta(\tau)}\) is lower semicontinuous: for every \(c\in\R\),
\(\pc{\tau \smallgiven \dmax\divergence{\tau}{\Delta(\tau)}\le c}
=\pc{\tau \smallgiven \tau\le 2^c\Delta(\tau)}\), which is closed since \(\tau\mapsto \tau-2^c\Delta(\tau)\) is continuous and the cone of positive operators is closed.
\(\cost[\epsilon](\rho)\) is generally larger than \(\dmax^{\epsilon}\divergence{\rho}{\free}\!\! 
= \displaystyle \inf _{\begin{smallmatrix}
  \sigma\in \states\\ \tau\approx_\epsilon \rho
\end{smallmatrix}}  \dmax\divergence{\tau}{\Delta\pa*{\sigma}}\)\vspace{-7pt}
 See~\cite{Datta2009}.
 In the following section we connect these quantities.

\subsection{Bounds of the One-Shot Cost and the Asymptotic Cost}\label{appendix:cost-bounds}
In this section we show that \(\cost[\epsilon]_\infty (\rho) \coloneqq \lim_{n \to \infty} \frac 1 n \cost[\epsilon](\rho^{\otimes n}) = D \divergence{\rho}{\Delta(\rho)}\) for all \(\rho\in \states\) and all \(\epsilon\in (0,1)\). 
First, we prove:
\begin{align}
  \label{eq:cost-bound}
    \dmax^\epsilon \divergence{\rho}{\free} &\le \cost[\epsilon] (\rho) \le \dmax^{\epsilon-\delta} \divergence{\rho}{\free}+ \log(1/\delta)\;,\qquad \forall\rho\in \states\;, \epsilon\in (0,1)\;, \delta\in (0, \epsilon)\;, \qquad \text{where}
 \\
  \dmax^\epsilon \divergence{\rho}{\free} &\coloneqq \inf_{\tau\approx_\epsilon \rho}\dmax \divergence{\tau}{\free} \defeq \inf_{\begin{smallmatrix}
      \tau \approx_\epsilon \rho \\
      \sigma \in \states
    \end{smallmatrix}}\dmax\divergence{\tau}{\Delta\pa*{\sigma}} \;.
 \end{align}
Before proving Eq.~(\ref{eq:cost-bound}), let us establish a claim.
\begin{claim}\label{claim:initial-cost-higher-bound}
\(\dmax \divergence{\rho}{\free} + \log 1/\epsilon \ge \cost[\epsilon] (\rho)\) for all \(\epsilon\in (0,1)\) and all \(\rho\in \states\).
\end{claim}
\begin{remark}
  Compare this result with Lemma~S5 in~\cite{JGW2025}.
\end{remark}
\begin{proof}
  Let \(\sigma\) be the optimizing free state so that \(\dmax \divergence{\rho}{\sigma} = \dmax \divergence{\rho}{\free}\) and let \(t \coloneqq  2^{\dmax \divergence{\rho}{\sigma}}\), so that \(t\sigma \ge \rho\). 
  It is therefore enough to find a state \(\tau\approx_\epsilon \rho\) such that \(\frac t \epsilon \Delta(\tau) \ge \tau\).
  
  We define \(\tau \coloneqq (1-\epsilon)\rho + \epsilon \sigma\); clearly \(\tau\approx_\epsilon \rho\).
  Moreover, 
  \begin{equation}\label{eq:freater-than-epsilon-sigma}
    \terminal{\tau} = (1-\epsilon )\terminal{\rho} + \epsilon\explainabove{\terminal{\sigma}}{=\sigma} \ge \epsilon\sigma\;,
  \end{equation}
 so:
\begin{enumerate}
  \item \(\frac t  \epsilon \terminal{\tau}\ge \rho\) because \(t \sigma \ge \rho\) by definition. 
  \item  \(\frac t  \epsilon \terminal{\tau}\ge \sigma\) since \(t\ge 1\).
\end{enumerate}  
Therefore, for all \(\lambda \in[0,1]\), \(\frac t  \epsilon \terminal{\tau} = \lambda\frac t  \epsilon \terminal{\tau} +(1-\lambda)\frac t  \epsilon \terminal{\tau} \ge \lambda \rho + (1-\lambda)\sigma\); in particular
\begin{equation}
  \frac t  \epsilon \terminal{\tau} \ge (1-\epsilon)\rho + \epsilon \sigma = \tau\;.
\end{equation}
\end{proof}
The following proof follows Ref.~\cite{JGW2025}.
\begin{proof}[Proof of Eq.~(\ref{eq:cost-bound})]
   \textbf{Lower bound.}
  \begin{align}
    \cost[\epsilon](\rho)&= \inf_{\tau \approx_\epsilon \rho}\dmax\divergence{\tau}{\Delta\pa*{\tau}} \\
    & \ge \inf_{\begin{smallmatrix}
      \tau \approx_\epsilon \rho \\
      \sigma \in \states
    \end{smallmatrix}}\dmax\divergence{\tau}{\Delta\pa*{\sigma}} \defeq \dmax^\epsilon\divergence{\rho}{\free}\;.
  \end{align}

  \textbf{Upper bound.} Fix \(\omega\approx_{\epsilon-\delta}\!\rho\). Since the \(\delta\)-neighborhood of \(\omega\) is contained within the \(\epsilon\)-neighborhood of \(\rho\) we have 
  \begin{equation}
    \cost[\delta](\omega)  = \inf_{\tau \approx_\delta \omega}\dmax\divergence{\tau}{\Delta\pa*{\tau}} \ge \inf_{\tau \approx_\epsilon \rho}\dmax\divergence{\tau}{\Delta\pa*{\tau}} = \cost[\epsilon](\rho)\;.
  \end{equation}
  Thus, by Claim~\ref{claim:initial-cost-higher-bound}, \(\dmax \divergence{\omega}{\free} + \log (1/\delta) \ge \cost[\delta](\omega) \ge \cost[\epsilon](\rho)\).
  Since \(\omega\) was an arbitrary state in the \((\epsilon-\delta)\)-neighborhood of \(\rho\), we conclude 
  \begin{equation}
    \cost [\epsilon] (\rho) \le \inf _{\omega \approx_{\epsilon-\delta}\rho} \dmax \divergence{\omega}{\free} + \log(1/\delta) \defeq \dmax^{\epsilon-\delta} \divergence{\rho}{\free}+ \log(1/\delta)\;.
  \end{equation}
\end{proof}
\begin{corollary}\label{appendix:asymptotic-cost}
  Let \(\rho\in \states\) and \(\epsilon\in (0,1)\). 
  Then, \(\displaystyle \cost[\epsilon]_\infty\pa*{\rho} \defeq \lim_{n \to \infty} \frac 1 n \cost[\epsilon]\pa*{\rho^{\otimes n}} = D \divergence{\rho}{\Delta(\rho)}\). 
\end{corollary}
\begin{proof}
  By the fact that \(D \divergence{\rho}{\free} = D \divergence{\rho}{\Delta(\rho)}\) and the generalized Quantum Stein's lemma~\cite{BP2010b,HY2025,lami2024} it follows that \(\lim_{n \to \infty} \frac 1 n \dmax^\eta \divergence{\rho^{\otimes n}}{\free}  = D \divergence{\rho}{\Delta(\rho)}\) for all \(\eta \in (0,1)\)~\footnote{The generalized quantum Stein's lemma is most commonly formulated in terms of the hypothesis-testing divergence. Here we present it in an equivalent form using the smoothed max-divergence \(\dmax^\eta\); this form is also known as the generalized asymptotic equipartition property~\cite{Gour2025}.}.
  Taking \(\eta= \epsilon\) and \(\eta = \epsilon- \delta\) for some fixed \(\delta\in (0,\epsilon)\) and regularizing Eq.~(\ref{eq:cost-bound}) provides the result.
\end{proof}

\section{The Distillable Instability Yield}
This appendix is dedicated to the distillation task.
In Sec.~\ref{appendix:one-shot-yield} we calculate the \(\epsilon\)-one-shot distillable yield; in Sec.~\ref{appendix:asymptotic-distillable-fragility} we connect it with the standard hypothesis testing divergence of the resource and use this connection to compute the \(\epsilon\)-achievable asymptotic yield rate. In Sec.~\ref{appendix:catalytic} we compute the one-shot yield in the paradigm of catalysis-assisted conversions.

Let us define:
\begin{equation}
  \scalarTerminal_m \coloneqq \pc*{0 \le \Gamma \le I \smallgiven \condExpectation(\Gamma) = 2^{-m} I} =\pc*{0 \le \Gamma \le I \smallgiven \tr[\sigma\Gamma]= 2^{-m} \quad \forall\sigma \in \free} \;,
\end{equation}
so that \(\bigcup_{m \in [0,\infty]} \scalarTerminal_m  = \pc{0 \le \Gamma \le I \smallgiven \condExpectation(\Gamma) \propto I} \defeq  \scalarTerminal \) and hence,
\begin{equation}\label{eq:characterization-of-restricted-HT-monotone}
  h^\epsilon(\rho) \defeq D_{\!H}^{\epsilon,\scalarTerminal }\divergence{\rho}{\sigma}\defeq
  -\log\!\!\!\!\!\min_{\begin{smallmatrix}
\Gamma\in\scalarTerminal\\
\tr[\rho\Gamma]\ge 1-\epsilon
\end{smallmatrix}}\!\!\!\!\!\!
\tr[\sigma\Gamma]= \max \pc{m\in [0,\infty]\smallgiven \exists \Gamma\in \scalarTerminal_m \;,\; \tr [\rho\Gamma] \ge 1- \epsilon} \;,
\end{equation}
(we recall that \(\sigma\) is an arbitrary free state in the definition above).

\subsection{Calculation of the One-Shot Distillable Yield}\label{appendix:one-shot-yield}
The \(\epsilon\)-one-shot yield of \(\rho\in \states(A)\) is the supremum over \(m\ge 0\) such that \(\rho \convto[\epsilon] \currency[m]\).
Recall that \(\currency[m]\) is described by the athermal state \(\ket 0\! \bra 0\) on a two-level \emph{classical} system \(B\) with Gibbs state \(\gibbs = 2^{-m} \ket 0\! \bra 0 +(1-2^{-m})\ket 1\! \bra 1\).

Therefore, the conversion \(\rho \convto[\epsilon] \currency[m]\) is achieved via a free binary measurement (POVM) channel \(\mathcal{M}\in \free(A \to B)\).
Any such channel is determined by a single \emph{quantum effect}, \(0\le \Gamma\le I^A\):
\begin{equation}\label{measurement-channel}
  \mathcal{M} (\tau) = \tr [\tau \Gamma] \ket 0\! \bra 0 + \tr[\tau(I-\Gamma)]\ket 1\! \bra 1,\; \qquad  \tau\in \states(A)\;.
\end{equation}
The condition that \(\mathcal{M}(\rho)\approx_{\epsilon} \ket 0\! \bra 0\) means that \(\tr [\rho \Gamma] \ge 1-\epsilon\).

As discussed in the main text, destruction covariance of \(\mathcal{M}\) means
\begin{equation*}
  \Delta\circ \mathcal{M} = \mathcal{M}\circ \Delta\;.
\end{equation*}
Since \(\Delta\) on the output bit prepares \(\gibbs\), this requires \(\mathcal{M}(\Delta\pa*{\sigma})= \gibbs\) for all states \(\sigma\).
This fixes \(\tr [\Delta\pa*{\sigma} \Gamma] = 2^{-m}\) for all states \(\sigma\), i.e., forces that \(\Gamma\in \scalarTerminal_m\).

The supremum over \(m\in [0, \infty)\) for which there exist \(\Gamma\in \scalarTerminal_m\) with \(\tr \rho \Gamma \ge 1- \epsilon\) is, by Eq.~(\ref{eq:characterization-of-restricted-HT-monotone}), \(h^\epsilon \pa{\rho}\).

\subsection{The Asymptotic Regime}\label{appendix:asymptotic-distillable-fragility}

In this section we connect the one-shot, catalytic, and battery-assisted instability yields with \(D_{\!H}^{\epsilon} \divergence{\rho}{\free}\). The argument follows Ref.~\cite{JGW2025}.

\begin{lemma} \label{thm:unrestricted-HT-yield-connection}
  Let \(\rho\in \states\) and \(\epsilon \in [0,1) \). Then:
  \begin{enumerate}[label=(\roman*)] 
    \item\label{item:one-shot-yield-bound}
    If \(D_{\!H}^{\epsilon} \divergence{\rho}{\free} \ge 2\), then \(\yield[\epsilon](\rho)\ge D_{\!H}^{\epsilon} \divergence{\rho}{\free} -1\).
  \item\label{item:asymptotic-yield-convergence}
  \(\displaystyle \yield[\epsilon]_\infty\pa*{\rho} \defeq \lim _{n \to \infty} \frac 1 n \yield[\epsilon]\pa*{\rho^{\otimes n}} = 
  \begin{cases}
\dmin \divergence{\rho}{\free} & \epsilon =0 \;,\\
    D \divergence{\rho}{\Delta(\rho)} & \epsilon \in (0,1)\;.\;
  \end{cases} \) 
  \end{enumerate}  
\end{lemma}
\begin{proof} 
  We will first establish item~\ref{item:one-shot-yield-bound} and then use it to derive item~\ref{item:asymptotic-yield-convergence}.

  Item~\ref{item:one-shot-yield-bound}.
By Eq.~(\ref{eq:hypothesis-heisenberg}) in the main text and the assumption,
there exists an effect \(\Gamma\) with \(p=\norm{\expectation{\Gamma}}_\infty \le 1/4\). 
By item~\ref{item:effect-enlargement} of Claim~\ref{effects-claim}, there exists an effect \( \Gamma\le \Lambda \le I\)---so that \(\tr[\rho\Lambda] \ge 1- \epsilon\)---with \(\expectation{\Lambda}=qI\), for \(q \le 2p\).  

Therefore, by Eq.~(\ref{eq:characterization-of-restricted-HT-monotone}), we have \(h^\epsilon(\rho)\ge -\log (2p) = D_{\!H}^{\epsilon} \divergence{\rho}{\free} - 1\).
 Since \(h^\epsilon = \yield[\epsilon]\), we are done.

Item~\ref{item:asymptotic-yield-convergence}.
Clearly, \(\yield[\epsilon](\rho)\le D_{\!H}^{\epsilon} \divergence{\rho}{\free}\). 
Together with the trivial bound \(\yield[\epsilon](\rho)\ge0\), item~\ref{item:one-shot-yield-bound} implies \(\yield[\epsilon](\rho)\ge D_{\!H}^{\epsilon} \divergence{\rho}{\free} - 2\) for all  \(\rho\in \states\). So, 
\begin{equation}\label{eq:chain-to-regularize}
  D_{\!H}^{\epsilon} \divergence{\rho}{\free} - 2\le \yield[\epsilon](\rho)\le D_{\!H}^{\epsilon} \divergence{\rho}{\free}\;.
\end{equation}
At \(\epsilon=0\),
  \(D_H^0\divergence{\rho}{\free}
  \defeq
  D_{\min}\divergence{\rho}{\free}\),
which is additive. Hence regularizing Eq.~\eqref{eq:chain-to-regularize} gives
\begin{equation}
  \yield[0]_\infty(\rho)
  =
  D_{\min}\divergence{\rho}{\free}.
\end{equation}
At \(\epsilon\in(0,1)\), the generalized quantum Stein lemma~\cite{BP2010b,HY2025,lami2024} gives
\begin{equation}
  \frac1n D_H^\epsilon\divergence{\rho^{\otimes n}}{\free}
  \xrightarrow{n\to\infty}
  D\divergence{\rho}{\free}
  =
  D\divergence{\rho}{\Delta(\rho)},
\end{equation}
and Eq.~\eqref{eq:chain-to-regularize} therefore yields
\begin{equation}
  \yield[\epsilon]_\infty(\rho)
  =
  D\divergence{\rho}{\Delta(\rho)}.
\end{equation}
This proves item~\ref{item:asymptotic-yield-convergence}.
\end{proof} 
\subsection{Catalytic Yield}\label{appendix:catalytic}
After finding the yield in the asymptotic paradigm, let us find it in the catalytic and battery-assisted paradigms.
\LemmaCatalytic*
\begin{remark}
  Recall the definitions of the catalytic and battery-assisted instability yield discussed in the main text:
  \begin{align*}
    \yieldcat[\epsilon](\rho)
    & \coloneqq
    \sup\pc*{ r \smallgiven \exists\,\tau\in\states\!:\ \rho\otimes\tau \convto[\epsilon] \currency[r]\otimes\tau}\;,\\
    \yieldbat[\epsilon](\rho)
    & \coloneqq    \sup \pc*{m-t \smallgiven \rho\otimes\currency[t]\convto[\epsilon]\currency[m]}\;.
  \end{align*}
  To see that the battery-assisted yield is a form of catalytic yield, note that \(\currency[r]\otimes\currency[t] \sim \currency[m]\) whenever \(m=t+r\) for \(t,r,m \in (0,\infty)\) (additivity of the currency; explained in Ref.~\cite{WW2019} and proved in Sec.~\ref{appendix:numeraire}). Therefore,
\(  
    \yieldbat[\epsilon](\rho) = \sup \pc* {r \smallgiven \exists t\;:\; \rho\otimes\currency[t] \convto[\epsilon] \currency[r] \otimes \currency[t]}
 \).
\end{remark}
\begin{proof}[Proof of Eq.~(\ref{eq:catalytic})]

By definition of the one-shot yield, \(\yieldbat[\epsilon](\rho) = \sup_{t} \pb[\big]{\yield[\epsilon](\rho\otimes\currency[t]) - t}\) and thus \(\yieldbat[\epsilon](\rho)\ge \yield[\epsilon](\rho\otimes\currency[1]) - 1\). To establish Eq.~(\ref{eq:catalytic}), we therefore need to show:
\begin{align}\label{eq:battery-yield-achievability}
    D_{\!H}^{\epsilon} \divergence{\rho}{\free} & \le \yield[\epsilon](\rho\otimes\currency[1]) - 1 \;,\qquad \text{and} \\
\label{eq:battery-yield-bound}
    \rho\otimes\currency[t] \convto[\epsilon] \currency[m] & \Rightarrow  D_{\!H}^{\epsilon} \divergence{\rho}{\free} \ge m-t\;. 
  \end{align}
  To prove Eq.~(\ref{eq:battery-yield-achievability}), set \(m = D_{\!H}^{\epsilon} \divergence{\rho}{\free}\)
and let \(\Gamma^A\) be the optimizing effect satisfying \(\tr [\rho \Gamma] \ge 1-\epsilon\) and \(\norm*{\expectation{ \Gamma}}_\infty = 2^{-m}\).
Let \(B\) be a classical thermodynamic bit with Gibbs state \(\gibbs \coloneqq \frac 1 2 \ket 0\! \bra 0 + \frac 1 2 \ket 1\! \bra 1\).
Recall that \(\currency[1]\) is the state \(\ket 0\! \bra 0^B\).

Item~\ref{item:corollary-for-depolarizing-system} of Claim~\ref{effects-claim} constructs an effect \(\Upsilon^{AB}\in \scalarTerminal_{m+1}\) satisfying \(\Upsilon^{AB} \ge \Gamma^A \otimes \ket 0\! \bra 0^B \).
    Thus, 
    \begin{equation}
      \tr [(\rho \otimes \currency[1])\Upsilon] \ge \tr [(\rho\otimes\ket 0\! \bra 0)( \Gamma \otimes \ket 0\! \bra 0)] = \tr [\rho  \Gamma] \ge 1-\epsilon \;.
    \end{equation}
    Therefore, the one-shot yield formula gives
     \begin{equation}
      \yield[\epsilon](\rho\otimes\currency[1]) 
     = D_{\!H,\scalarTerminal}^\epsilon \divergence{\rho\otimes\currency[1]}{\free} 
     \ge m+1 \defeq D_{\!H}^{\epsilon}\divergence{\rho}{\free} +1\;,
     \end{equation}
proving Eq.~(\ref{eq:battery-yield-achievability}).

To conclude Eq.~(\ref{eq:catalytic}), it is left to prove Eq.~(\ref{eq:battery-yield-bound}).
Here, we assume  \(\rho\otimes\currency[t] \convto[0] \p \approx_\epsilon \currency[m]\) where we recall that \(\currency[m]\) is the athermal state \(\ket 0\! \bra 0\) with respect to the Gibbs state  \(\gamma = 2^{-m} \ket 0\! \bra 0 +(1-2^{-m})\ket 1\! \bra 1\) and \(\p\) is some classical state.
Therefore, \(\tr [\p \ket 0\! \bra 0 ]\ge 1-\epsilon\) and hence \(\ket 0\! \bra 0\) is an admissible effect in the optimization of \(D_{\!H}^{\epsilon} \divergence{\p} \free = D_{\!H}^{\epsilon} \divergence{\p} {\gamma}\). As discussed in the main text,
\begin{equation*}
  D_{\!H}^{\epsilon}\divergence{\p}{\gamma}
  =
  -\log\!\!\!\!\!\min_{\begin{smallmatrix}
  0\le\Lambda\le I\\
  \tr[\p\Lambda]\ge 1-\epsilon
  \end{smallmatrix}}\!\!\!\!\!
  \tr[\gamma\Lambda]\;.
\end{equation*}
Using the admissible effect \(\ket 0\! \bra 0\), we obtain
\begin{align}
    D_{\!H}^{\epsilon} \divergence{\p} \free & \ge - \log \tr [\ket 0\! \bra 0 \gamma] \\
    \hspace{-100pt}\explain{\(\gamma=2^{-m} \ket 0\! \bra 0 +(1-2^{-m})\ket 1\! \bra 1\)}&= m\;. 
\end{align}
Thus, since \(D_{\!H}^{\epsilon} \divergence{\cdot}{\free}\) is a monotone and \(\rho \otimes\currency[t]\convto[0]\p\),
\begin{align}
  m &\le D_{\!H}^{\epsilon} \divergence{\p}{\free} \\ 
  & \le  D_{\!H}^{\epsilon} \divergence{\rho\otimes \currency[t]}{\free} \\
 \hspace{-40pt}\explain{Claim~\ref{claim:hypo-splits-with-battery} below}& = D_{\!H}^{\epsilon} \divergence{\rho}{\free} + t\;. 
\end{align}
This establishes Eq.~(\ref{eq:battery-yield-bound}) and hence completes the proof of Eq.~(\ref{eq:catalytic}).
\end{proof}

\begin{proof}[Proof of Eq.~(\ref{eq:exact-catalytic})]
By Eq.~(\ref{eq:catalytic}), it remains to show
\(\dmin \divergence{\rho}{\free} \ge \yieldcat(\rho)\).

  To do this, note that if \(\rho \otimes \tau \succ \currency[r] \otimes \tau\) then by additivity and monotonicity of \(\dmin \divergence{\cdot}{\free}\) (see the remark opening Sec.~\ref{appendix:extremals}) we have
  \[\dmin \divergence{\rho}{\free} + \dmin \divergence{\tau}{\free} = \dmin \divergence{\rho\otimes\tau}{\free} \ge \dmin \divergence{\currency[r] \otimes\tau}{\free} = r + \dmin \divergence{\tau}{\free}\;,\]
  which implies \(\dmin \divergence{\rho}{\free} \ge r\).
  Thus, by the definition of the catalytic yield,
  \(\dmin \divergence{\rho}{\free} \ge \yieldcat(\rho)\).
\end{proof}

\subsection{Auxiliary Claims}\label{appendix:aux}
\begin{claim}\label{claim:hypo-splits-with-battery}
 \(D_{\!H}^{\epsilon} \divergence{\rho\otimes \currency[t]}{\free} = D_{\!H}^{\epsilon} \divergence{\rho}{\free} + t\) for all \(\rho\in \states(A)\) and \(t\in (0,\infty)\).
\end{claim}
\begin{proof}
   \(\currency[t]\) is the athermal state \(\ket 0\! \bra 0\) on a thermodynamic system \(B\) with Gibbs state \(\gamma=2^{-t} \ket 0\! \bra 0 +(1-2^{-t}) \ket 1\! \bra 1\). 

  Let \(\sigma\in \free(A)\). 
  For \(0\le \Gamma \le I^{AB}\) an effect on the composite system,  we have
  \begin{align}
    \tr[(\rho^A\otimes \ket 0\! \bra 0^B)\Gamma] & = \tr [\rho \Gamma_0]\;,
    \\ 
    \tr [(\sigma \otimes \gamma) \Gamma] & = 2^{-t} \tr[\sigma \Gamma_0] + (1-2^{-t}) \tr[\sigma \Gamma_1] \\
    & \ge 2^{-t} \tr[\sigma \Gamma_0]\;,
  \end{align}
 where \(\Gamma_x = \tr_B [\ket x\! \bra x^B \Gamma^{AB}]\) for \(x=0,1\) are effects on \(A\).
 Suppose now that \(\tr[(\rho\otimes\currency[t]) \Gamma] \ge 1-\epsilon\). Then, 
  \begin{equation}\label{eq:inequality-for-hypo-divergence}
    \tr[(\sigma \otimes \gamma) \Gamma] \ge 2^{-t} \tr[\sigma \Gamma_0] \ge 2^{-t}\inf \pc*{\tr [\sigma \Lambda'] \smallgiven 0\le \Lambda'\le I^A\;, \tr[\rho\Lambda'] \ge 1-\epsilon}\;.
  \end{equation}
  On the other hand, given an effect \(0\le\Lambda\le I^A\), 
  \begin{align}
    \tr[(\rho^A\otimes \ket 0\! \bra 0^B)(\Lambda^A\otimes \ket 0\! \bra 0^B)] & = \tr [\rho \Lambda]\;,
    \\ 
    \tr [(\sigma \otimes \gamma) (\Lambda^A\otimes \ket 0\! \bra 0^B)] & = 2^{-t} \tr[\sigma \Lambda]\;,
  \end{align}
  and thus, if \(\tr[\rho \Lambda] \ge 1-\epsilon\), then
  \begin{equation}\label{eq:inequality-for-hypo-divergence-reverse}
    2^{-t} \tr[\sigma \Lambda] \ge \inf \pc*{\tr [(\sigma \otimes \gamma) \Gamma'] \smallgiven 0\le \Gamma'\le I^{AB}\;, \tr[(\rho^A\otimes \ket 0\! \bra 0^B)\Gamma'] \ge 1-\epsilon}\;.
  \end{equation}
  Since Eqs.~(\ref{eq:inequality-for-hypo-divergence}) and~(\ref{eq:inequality-for-hypo-divergence-reverse}) are true for all effects \(\Gamma^{AB}\) and \(\Lambda^A\), taking the infimum over admissible effects in both directions and then applying \(-\log\) yields
  \begin{equation}
    D_{\!H}^\epsilon \divergence{\rho}{\sigma} +t = D_{\!H}^\epsilon \divergence{\rho\otimes \currency[t]}{\sigma\otimes\gamma}\;.
  \end{equation}
Because \(\sigma\in \free(A)\) was arbitrary and \(\free(AB) = \pc*{\sigma\otimes\gamma|\sigma\in \free(A)}\) (by locality of destruction and since \(\Delta^{\!B}(\cdot) = \tr [\cdot] \gamma\)), this concludes the proof. 
\end{proof}

\begin{claim}\label{effects-claim}
  Let \(\Gamma\) be a quantum effect on system \(A\), and set \(p\coloneqq\norm*{\expectation{\Gamma}}_\infty=2^{-m}\).
  The following hold:
  \begin{enumerate}[label=(\roman*)] 
    \item\label{item:sce-catalytic} Given a number \(t \in (0,\infty)\) satisfying \(2^{-(t+m)} \le 1-2^{-t}\) (in particular, any \(t\ge 1\)) and an effect \(\Lambda^B \in \scalarTerminal_t\) on some system \(B\), there exists a composite effect \(\Upsilon^{AB} \in \scalarTerminal_{m+t}\) satisfying \(\Upsilon \ge  \Gamma \otimes \Lambda\).
    \item\label{item:corollary-for-depolarizing-system} If \(B\) is the two-level system and \(\Delta^{\!B}\) is the depolarization \((\Delta^B(\rho) = \u^B)\), then there exists an effect \(\Upsilon^{AB}\in \scalarTerminal_{m+1}\) satisfying \(\Upsilon \ge \Gamma^A \otimes \ket 0\! \bra 0^B\).
    \item\label{item:effect-enlargement} If \(p\le\frac14\), then there exists an effect \(\Gamma\le\Lambda\le I\) such that \(\expectation{\Lambda}=qI\) for some \(q\le \frac{p}{1-2p}\) (in particular \(q\le 2p\)).
  \end{enumerate}
\end{claim}
\begin{proof}
  We prove all items by explicit construction.

   
   Item~\ref{item:sce-catalytic}.
   Define
\begin{equation}\label{eq:def-composite-effect}
  \Upsilon \coloneqq \Gamma \otimes \Lambda  + \frac{2 ^{-t}}{1- 2^{-t}}\pa*{p I - \expectation{ \Gamma}} \otimes (I - \Lambda)\;.
\end{equation}
    We need to prove \(\Gamma\otimes\Lambda \overset{\text{\ref{item:Upsilon≥}}}{\le} \Upsilon \overset{\text{\ref{item:Upsilon≤I}}}{\le} I \) and \(\condExpectation(\Upsilon) \overset{\text{\ref{item:∆(Upsilon)=}}}{=} 2^{-(m+t)}I\):
    \begin{enumerate}[label=(\alph*)] 
        \item\label{item:Upsilon≥}
        \(\Gamma \le I\) hence
         \(p I \ge \expectation{ \Gamma}\), so \(\frac{2 ^{-t}}{1- 2^{-t}}\pa*{p I - \expectation{ \Gamma}} \otimes (I - \Lambda) \in \positiveOperators(AB)\) and hence \(\Upsilon \ge  \Gamma\otimes \Lambda\).
        \item\label{item:Upsilon≤I}
         \(p I - \expectation{ \Gamma} \le p I \defeq 2^{-m} I \), hence 
        \begin{equation}
          \begin{aligned}
              \frac{2 ^{-t}}{1- 2^{-t}}\pa*{p I - \expectation{ \Gamma}} 
              &
              \le \frac{2 ^{-(t+m)}}{1- 2^{-t}}I   
              \le I\;.\end{aligned}
        \end{equation}
        where the rightmost inequality follows by the assumption that \(2^{-(t+m)} \le 1-2^{-t}\).
        Therefore, by the definition in Eq.~(\ref{eq:def-composite-effect}), \(\Upsilon \le \Gamma \otimes \Lambda  + I \otimes (I - \Lambda) \le I\).
        \item\label{item:∆(Upsilon)=}
        Finally,  \vspace{-20pt}       
         \begin{align}
                \expectation{\Upsilon}
                &  = \expectation{ \Gamma} \otimes \explainabove{\expectation{\Lambda}}{=2^{-t}I}
                 +\frac{2 ^{-t}}{1- 2^{-t}}\pa[\big]{p \explainabove{\expectation{I}}{=I} - \explainabove{\expectation{(\expectation{ \Gamma})}}{= \expectation{ \Gamma}}} 
                \otimes
                 (\explainabove{\expectation{I}}{=I} - \explainabove{\expectation{\Lambda}}{=2^{-t}I} )
                \\ &
                = 2 ^{-t} \expectation{ \Gamma} \otimes I + \frac{2 ^{-t}}{1- 2^{-t}} (p I - \expectation{ \Gamma}) \otimes (1-2 ^{-t}) I
                \\ & 
                = 2 ^{-t} \expectation{ \Gamma} \otimes I + 2 ^{-t} (p I - \expectation{ \Gamma}) \otimes  I
                \\ &
= 2^{-t} p I = 2^{-(m+t)} I 
            \end{align}
\end{enumerate} 
Item~\ref{item:corollary-for-depolarizing-system}. Apply item~\ref{item:sce-catalytic} for \(\Lambda= \ket 0\! \bra 0^B\) (for which \(t=1\)).

Item~\ref{item:effect-enlargement}.
Write \(\norm{\mathcal N}\) for the operator norm induced by \(\norm{\cdot}_\infty\) on maps \(\mathcal N:\algebra\to\algebra\), where \(\algebra=\ima\condExpectation\).
For every effect \(0\le Z\le I\), define the positive map \(\mathcal E_Z:\algebra\to\algebra\) by \(\mathcal E_Z(X)\coloneqq\expectation{ZXZ}\).
Set \(R\coloneqq(I-\Gamma)^{1/2}\) and define \(\Theta\coloneqq\id-\mathcal E_R\).

To expand \(\Theta\), denote the complementary effect \(K\coloneqq I-R\). For \(X\in\algebra\),
\begin{align}
  \Theta(X)
  &=X-\expectation{(I-K)X(I-K)}
  \\
  \explain{\(X=\expectation{X}\)}
  &=\expectation{KX}+\expectation{XK}-\mathcal E_K(X)\;. \label{eq:3-summands}
\end{align}
Since \(0\le K \le I\), functional calculus gives \(K^2 \le K \defeq I -(I-\Gamma)^{1/2} \le \Gamma\); positivity of \(\condExpectation\) therefore gives
\begin{equation}\label{eq:effect-K-expectation-bounds}
  \norm*{\expectation{K^2}}_\infty
  \le\norm*{\expectation{K}}_\infty
  \le \norm*{\expectation{\Gamma}}_\infty \defeq p\;.
\end{equation}
Choi--Effros submultiplicativity~\cite{ce1977} (or Bimodularity~(\ref{eq:bimodularity})) gives
\begin{align}
  \norm*{\expectation{KX}}_\infty,
  \norm*{\expectation{XK}}_\infty
  &\le\norm*{\expectation{K}}_\infty\norm{X}_\infty
  \overset{\eqref{eq:effect-K-expectation-bounds}}{\le}
  p\norm{X}_\infty\;. \label{eq:first-2-summands-bound}
\end{align}
Moreover, \(\mathcal E_K\) is positive, so its norm is attained at the identity. Hence,
\begin{align}
  \norm{\mathcal E_K}
  &=\norm*{\mathcal E_K(I)}_\infty
  =\norm*{\expectation{K^2}}_\infty \overset{\eqref{eq:effect-K-expectation-bounds}}{\le}p\;. \label{eq:third-summand-bound}
\end{align}
By Eqs.~(\ref{eq:3-summands}), (\ref{eq:first-2-summands-bound}) and (\ref{eq:third-summand-bound}),
\begin{equation}
  \norm{\Theta} \le 3p \;,\label{3p-bound}
\end{equation}
but, by assumption, \(3p\le\frac34<1\), thus \(\mathcal E_R=\id-\Theta\) is invertible, with \(\mathcal E_R^{-1}=\sum_{n\ge0}\Theta^n\).

Set \(
  Y\coloneqq\mathcal E_R^{-1}(I)=\sum_{n\ge0}\Theta^n(I)
\).
We have \(\Theta(I)=\!\explainabove{I}{\!\!\!\!\!\!\!=\expectation{I}\!\!\!\!\!\!\!}\!\!\!-\expectation{R^2}=\expectation{I-R^2}\defeq \expectation{\Gamma}\). 
  Therefore, for every \(n\ge1\), 
  \begin{equation}
    \norm*{\Theta^n(I)}_\infty\le \norm*{\Theta^{n-1}} \norm[\big]{\Theta(I)}_\infty \le (3p)^{n-1}p\;.
  \end{equation}
It follows that
\begin{align}
  \norm{Y-I}_\infty \label{eq:explaining-Y-positivity}
  &\le 
   \sum_{n\ge 1}\norm{\Theta^{n}(I)}_{\infty}  \le\sum_{n\ge1}(3p)^{n-1}p
  =\frac{p}{1-3p} 
  \;.
\end{align}
By assumption, \(p\le 1/4\), so \(\norm{Y-I}_\infty \le 1\); since \(\Theta\) is Hermitian preserving, \(Y\defeq\sum_{n\ge 0}\Theta^{n}(I)\) is Hermitian, so this means that \(Y \ge 0\). Therefore,
\begin{equation}
  0\le RYR \le\norm{Y}_\infty R^2 \defeq \norm{Y}_\infty(I-\Gamma)\;. \label{RYR-bound}
\end{equation}
Furthermore,
\begin{align}
  \label{item2} \norm{Y}_\infty & \le \norm{Y-I}_\infty + \norm{I}_\infty
  \overset{\!\!\!\text{(\ref{eq:explaining-Y-positivity})}\!\!\!}{\le}\frac{p}{1-3p} +1
  =\frac{1-2p}{1-3p}\;. 
\end{align} 
Now define \(\displaystyle \Lambda\coloneqq I-\frac{RYR}{\norm{Y}_\infty}\);
Eq.~(\ref{RYR-bound}) implies \(\Gamma\le\Lambda\le I\).
Moreover,
\begin{align}
  \expectation{\Lambda}
  &=I-\frac{\explainabove{\expectation{RYR}}{=
  \mathcal E_R(Y)=I}}{\norm{Y}_\infty}
  =\pa*{1-\frac{1}{\norm{Y}_\infty}}I
  \eqqcolon qI\;.
\end{align}
Finally,
\begin{align}
  q
  &=1-\frac{1}{\norm{Y}_\infty}
  \overset{\text{(\ref{item2})}}{\le}
  1-\frac{1-3p}{1-2p}
  =\frac{p}{1-2p}\;,
\end{align}
which proves the claim.
\end{proof}

\clearpage
\section*{List of Symbols}
For convenience, we collect the notation used most heavily in the main text and in this supplemental material.

\begingroup
\small
\renewcommand{\arraystretch}{1.15}
\newcommand{\NotationEntry}[2]{\parbox[t]{0.30\textwidth}{#1} & \parbox[t]{0.62\textwidth}{#2}\\}

\begin{center}
\begin{tabular}{@{}ll@{}}
\NotationEntry{\textbf{Symbol}}{\textbf{Definition}}
\NotationEntry{\(A,B,H\ldots\)}{Quantum systems represented by a finite dimensional Hilbert space.}
\NotationEntry{\(\linearOperators(H)\)}{Algebra of linear operators on \(H\).}
\NotationEntry{\(\positiveOperators(A)\)}{Positive semidefinite operators on system \(A\).}
\NotationEntry{\(\rho,\sigma,\tau,\omega,\dots\in \states(A)\)}{Quantum states on system \(A\).}
\NotationEntry{\(\mathcal{N},\mathcal{M},\mathcal{P},\ldots \in \CPTP(A \to B)\)}{Quantum channels from \(A\) to \(B\).}
\NotationEntry{\(X^*, \mathcal{N}^*\)}{Adjoint of a linear operator/Heisenberg dual of a quantum channel.}
\NotationEntry{\(I, I^A\)}{Identity operator on the relevant system, or explicitly on \(A\).}
\NotationEntry{\(\u,\u^A\)}{Uniform (maximally mixed) state on the relevant system, or explicitly on \(A\).}
\NotationEntry{\(\id^A\)}{Identity channel of the system \(A\).}
\NotationEntry{\(\tr, \tr_A\)}{Trace/partial trace of the system \(A\).}
\NotationEntry{\(X^p\)}{operator power, always taken on the support of \(X\).}
\NotationEntry{\(X \gg Y\)}{the support of \(Y\) is contained in the support of \(X\).}
\NotationEntry{\(\norm{X}_p\)}{Schatten \(p\)-norm, or quasinorm when \(p\in (0,1)\).}
\NotationEntry{\(\Delta\)}{Destruction channel in the Schr\"odinger picture.}
\NotationEntry{\(\algebra=\ima \Delta^{\!*}\)}{Fixed-point algebra of \(\Delta^{\!*}\).}
\NotationEntry{\(\tpExpectation\)}{Trace-preserving conditional expectation onto \(\algebra\).}
\NotationEntry{\(\bar{\omega}\)}{A free state \(\omega\) after ``twisting''/projecting it into \(\algebra\) (\(\bar{\omega} \coloneqq{\Delta(I)}^{-1}\omega=\bar{\Delta}(\omega)\)).}
\NotationEntry{\(\free(A)\)}{Free states on \(A\) (the \(\Delta\)-fixed states).}
\NotationEntry{\(\free(A\to B)\)}{Free (destruction-covariant) channels from \(A\) to \(B\).}
\NotationEntry{\(\currency[m]\)}{Currency state with value \(m\).}
\NotationEntry{\(\rho \convto[\epsilon] \sigma\)}{\(\epsilon\)-approximate conversion from \(\rho\) to \(\sigma\) by a free channel.}
\NotationEntry{\(\yield[\epsilon](\rho),\cost[\epsilon](\rho)\)}{\(\epsilon\)-one-shot distillable yield and dilution cost of \(\rho\).}
\NotationEntry{\(\yieldbat[\epsilon](\rho),\yieldcat[\epsilon](\rho)\)}{Battery-assisted and catalytic one-shot yields.}
\NotationEntry{\(\monotone_\infty(\rho)\)}{Regularized version of a monotone \(\monotone\), when the limit exists.}
\NotationEntry{\(\D\divergence{\rho}{\sigma}\)}{Generic quantum divergence.}
\NotationEntry{\(\mean\)}{Generic binary action on positive operators with properties specified around Eq.~(\ref{def:mean}); often a Kubo-Ando mean.}
\NotationEntry{\(\#_\lambda\)}{Operator geometric mean with weight \(\lambda\). An instance of \(\mean\).}
\NotationEntry{\(M_\D^\mean(\rho)\)}{Instability monotone induced by a divergence \(\D\) and a binary action \(\mean\).}
\NotationEntry{\(D_{\alpha,z}\divergence{\rho}{\sigma}\)}{\(\alpha\)-\(z\) R\'enyi divergence.}
\NotationEntry{\(Q_{\alpha,z}\divergence{\rho}{\sigma}\)}{Trace quantity underlying \(D_{\alpha,z}\divergence{\rho}{\sigma}\).}
\NotationEntry{\(\dmin,\dmax\)}{Min- and max-relative entropies.}
\NotationEntry{\(F_X(S)\) or \(F(S)\)}{Trace functional \(\tr\!\big[(S^{r/2}XS^{r/2})^z\big]\) used to characterize optimizers.}
\NotationEntry{\(\opt\)}{Optimum; interpreted as \(\min\) or \(\max\) depending on the parameter regime (the optimum is attained, see Sec.~\ref{appendix:optimization}).}
\NotationEntry{\(\region\)}{DPI region for the parameters \((\alpha,z)\).}
\NotationEntry{\(\region_F^{\smile},\region_F^{\frown},\region_F\)}{Convexity/concavity/total parameter region for the functional \(F_X\).}
\NotationEntry{\(\mathfrak{N}_r(X)\)}{Admissible region of optimizer candidates for \(F_X\); determined by a support condition.}
\NotationEntry{\(D_{\!H}^{\epsilon}\divergence{\rho}{\sigma}\)}{Hypothesis-testing divergence.}
\NotationEntry{\(\scalarTerminal_m,\scalarTerminal\)}{Restricted families of effects used in the hypothesis-testing formulation.}
\NotationEntry{\(h^\epsilon(\rho)\)}{Monotone arising from hypothesis-testing when tests are restricted to \(\scalarTerminal\).}
\NotationEntry{\(\monotone_{\alpha,z}^\lambda(\rho)\)}{Three-parameter additive R\'enyi instability monotone.}
\end{tabular}
\end{center}

\endgroup

\end{document}